\documentclass[acmtocl,acmnow,runninghead]{acmtrans2m}
\usepackage{amscd}
\usepackage{amssymb}
\usepackage{amsmath}
\usepackage{stmaryrd}
\usepackage{mathrsfs}
\usepackage{times}
\usepackage{latexsym}
\usepackage{hhline}
\usepackage{multicol}          
\usepackage{multirow}          
\usepackage{rotating}          
\usepackage[latin1]{inputenc}  
\usepackage{calc}              
\usepackage{graphicx}          
\usepackage{ifthen}            
\usepackage{subfigure}         
\usepackage{pst-all}           
\usepackage{pst-poly}          
\usepackage{multido}           
\usepackage{colortbl}          
\usepackage{array}             
\usepackage{booktabs}          
\usepackage{ctable}            
\usepackage[normalem]{ulem}    
\newtheorem{theorem}{Theorem}[section]

\newtheorem{corollary}[theorem]{Corollary}
\newtheorem{property}[theorem]{Property}

\newtheorem{lemma}[theorem]{Lemma}
\newcommand{\set}[1]{\left\{ #1 \right\}}
\newcommand{\pspar}{\psset{arrows=->, arrowscale=2, linewidth=1pt, linecolor=black, linestyle=solid, rowsep=1.5cm, colsep=1.5cm, levelsep=1.5cm, nodesep=3pt, labelsep=3pt, arcangle=30, border=0.05cm}}
\newcommand{\msout}[1]{\text{\sout{\ensuremath{#1}}}}
\newdef{example}[theorem]{Example}
\newdef{definition}[theorem]{Definition}
\newdef{remark}[theorem]{Remark}
\newdef{algorithm}[theorem]{Algorithm}
\title{Parsing methods streamlined}
\author{Luca Breveglieri \quad Stefano Crespi Reghizzi \quad Angelo Morzenti \\ \\ Dipartimento di Elettronica, Informazione e Bioingegneria (\emph{DEIB}) \\ Politecnico di Milano \\ Piazza Leonardo Da Vinci $32$, I-$20133$ Milano, Italy \\ email: $\{$ \text{luca.breveglieri, stefano.crespireghizzi, angelo.morzenti} $\}$@polimi.it}
\begin{abstract}
This paper has the goals ($1$) of unifying top-down parsing with shift-reduce parsing to yield a single simple and consistent framework, and ($2$) of producing provably correct parsing methods, deterministic as well as tabular ones, for extended context-free grammars (\emph{EBNF}) represented as state-transition networks. Departing from the traditional way of presenting as independent algorithms the deterministic bottom-up $\emph{LR} \, (1)$, the top-down $\emph{LL} \, (1)$ and the general tabular (Earley) parsers, we unify them in a coherent minimalist framework. We present a simple general construction method for \emph{EBNF} $\emph{ELR} \, (1)$ parsers, where the new category of convergence conflicts is added to the classical shift-reduce and reduce-reduce conflicts; we prove its correctness and show two implementations by deterministic push-down machines and by vector-stack machines, the latter to be also used for Earley parsers. Then the Beatty's theoretical characterization of $\emph{LL} \, (1)$ grammars is adapted to derive the extended $\emph{ELL} \, (1)$ parsing method, first by minimizing the $\emph{ELR} \, (1)$ parser and then by simplifying its state information. Through using the same notations in the $\emph{ELR} \, (1)$ case, the extended Earley parser is obtained. Since all the parsers operate on compatible representations, it is feasible to combine them into mixed mode algorithms.
\end{abstract}
\category{parsing}{of formal languages}{theory}
\terms{language parsing algorithm, syntax analysis, syntax analyzer}
\begin{document}
\maketitle\today\footnotetext{See \emph{Formal Languages and Compilation}, S. Crespi Reghizzi, L. Breveglieri, A. Morzenti, Springer, London, $2$nd edition, (planned) $2014$}
\newpage
\tableofcontents
\newpage
\section{Introduction}\label{sectIntroduction}
In many applications such as compilation, program analysis, and document or natural language processing, a language defined by a formal grammar has to be processed using the classical approach based on syntax analysis (or parsing) followed by syntax-directed translation. Efficient deterministic parsing algorithms have been invented in the $1960$'s and improved in the following decade; they are described in  compiler related textbooks (for instance \cite{AhoLamSethiUl2006,GruneJacobs2004,Crespi09}), and efficient implementations are available and widely used. But research in the last decade or so has focused on issues raised by technological advances, such as  parsers for XML-like data-description languages, more efficient general tabular parsing for  context-free grammars, and  probabilistic parsing for natural language processing,  to mention just some  leading research lines.
\par
This paper has the goals ($1$) of unifying top-down parsing with shift-reduce parsing and, to a minor extend, also with tabular Earley parsing, to yield a single simple and consistent framework, and (2) of producing  provably correct parsing  methods, deterministic and also tabular ones, for extended context-free grammars (\emph{EBNF}) represented as state-transition networks.
\par
We address the first goal. Compiler and language developers, and those who are familiar with the technical aspects of parsing, invariably feel that there ought to be room for an improvement of the classical  parsing methods:  annoyingly similar yet incompatible notions are used in shift-reduce (i.e., $\emph{LR} \, (1)$), top-down (i.e., $\emph{LL} \, (1)$) and  tabular Earley parsers. Moreover, the parsers are presented as independent algorithms, as indeed they were when first invented, without taking  advantage of the known grammar inclusion properties ($\emph{LL} \, (k) \subset \emph{LR} \, (k) \subset \text{context-free}$) to tighten and simplify constructions and proofs. This may be a consequence of the excellent quality of the original presentations (particularly but not exclusively \cite{Knuth1965LR(k),Rosenkrantz:Stearns:ic:1970,Earley70}) by distinguished scientists, which made revision and systematization less necessary.
\par
Our first contribution is to conceptual economy and clarity. We have reanalyzed the traditional deterministic parsers on the view provided by Beatty's \cite{Beatty82} rigorous characterization of $\emph{LL} \, (1)$ grammars as a special case of the $\emph{LR} \, (1)$. This allows us to show how to transform shift-reduce parsers into the top-down parsers by merging and simplifying parser states and by anticipating parser decisions. The result is that only one set of technical definitions suffices to present all parser types.
\par
Moving to the second goal, the best known shift-reduce and tabular parsers do not accept the extended form of \emph{BNF} grammars known as \emph{EBNF} (or \emph{ECFG} or also as regular right-part \emph{RRPG}) that make use of regular expressions, and are widely used in the language reference manuals and are popular among designers of top-down parsers using recursive descent methods. We  systematically use a graphic representation of \emph{EBNF} grammars by means of state-transition networks (also known as syntax charts. Brevity prevents to discuss in detail the long history of the research on parsing methods for \emph{EBNF} grammars. It suffices to say that the existing top-down deterministic method is well-founded and very popular, at least since its use in the elegant recursive descent Pascal compiler \cite{Wirth75} where \emph{EBNF} grammars are represented by a network of finite-state machines, a formalism already in use since at least \cite{journals/jacm/Lomet73}). On the other hand, there have been numerous interesting but non-conclusive proposals for shift-reduce methods for \emph{EBNF} grammars, which operate under more or less general assumptions and are implemented by various types of push-down parsers (more of that in Section \ref{sectRelatedWorkELR}). But a  recent survey  \cite{conf/lata/Hemerik09} concludes rather negatively:
\begin{quotation}\noindent
What has been published about \emph{LR}-like parsing theory is so complex that not many feel tempted to use it; \ldots it is a striking phenomenon that the ideas behind recursive descent parsing of \emph{ECFG}s can be grasped and applied immediately, whereas most of the literature on \emph{LR}-like parsing of \emph{RRPG}s is very difficult to access. \ldots Tabular parsing seems feasible but is largely unexplored.
\end{quotation}
We have decided to represent extended grammars as transition diagram systems, which are of course equivalent to grammars, since any regular expression can be easily mapped to its finite-state recognizer; moreover, transition networks (which we dub machine nets) are often more readable than grammar rules. We stress that \emph{all} our constructions operate on machine nets that represent \emph{EBNF} grammars and, unlike many past proposals, do not make restrictive assumptions on the form of regular expressions.
\par
For \emph{EBNF} grammars, most past shift-reduce methods had met with a difficulty: how to formulate a condition that ensures deterministic parsing in the presence of recursive invocations and of cycles in the transition graphs. Here we offer a simple and rigorous formulation that adds to the two classical conditions (neither shift-reduce nor reduce-reduce conflicts) a third one: no convergence conflict. Our shift-reduce parser is presented in two variants, which use different devices for identifying the right part or handle (typically a substring of unbounded length) to be reduced:  a deterministic pushdown automaton, and  an implementation, to be named vector-stack machine,  using unbounded integers as pointers into the stack.
\par
The vector-stack device is also used in our last development: the tabular Earley-like parser for \emph{EBNF} grammars. At last, since all parser types described operate on uniform assumptions and use compatible notations, we suggest the possibility to combine them into mixed-mode algorithms.
\par
After half a century research on parsing, certain facts and properties that had to be formally proved in the early studies have become obvious, yet the endemic presence of errors or inaccuracies (listed in \cite{conf/lata/Hemerik09}) in published  constructions for \emph{EBNF} grammars, warrants that all new constructions be proved correct. In the interest of  readability and brevity, we first present the enabling properties and the constructions semi-formally also relying on significant examples; then, for the properties and constructions that are new, we provide correctness proofs in the appendix.
\par
The paper is organized as follows. Section \ref{sectionPreliminaries} sets the terminology and notation for grammars and transition networks. Section \ref{sectionShiftReduce} presents the construction of shift-reduce $\emph{ELR} \, (1)$ parsers. Section \ref{sectDeterministicTopDownParsing} derives top-down $\emph{ELL} \, (1)$ parsers, first by transformation of shift-reduce parsers and then also directly. Section \ref{sectionTabular} deals with tabular Earley parsers.
\section{Preliminaries} \label{sectionPreliminaries} \label{SectGramNetwork}
The concepts and terminology for grammars and automata are classical (e.g., see \cite{AhoLamSethiUl2006,Crespi09}) and we only have to introduce some specific notations.
\par
A \emph{BNF} or context-free \emph{grammar} $G$ is specified by the \emph{terminal alphabet} $\Sigma$, the \emph{set of nonterminal symbols} $V$, the \emph{set of rules} $P$ and the \emph{starting symbol} or \emph{axiom} $S$. An element of $\Sigma \cup V$ is called a \emph{grammar symbol}. A rule has the form $A \to \alpha$, where $A$, the \emph{left part}, is a nonterminal and $\alpha$, \emph{the right part}, is a possibly empty (denoted by $\varepsilon$) string over $\Sigma \cup V$. Two rules such as $A \to \alpha$ and $A \to \beta$ are called \emph{alternative} and can be shortened to $A \to \alpha \; \mid \; \beta$.
\par
An \emph{Extended BNF} (\emph{EBNF}) grammar $G$ generalizes the rule form by allowing the right part to be a \emph{regular expression} (r.e.) over $\Sigma \cup V$. A r.e. is a formula that as operators uses union  (written ``$|$''), concatenation, Kleene star and parentheses.
The language defined by a r.e. $\alpha$ is denoted $R \, (\alpha)$. For each nonterminal $A$ we assume, without any loss of generality, that $G$ contains exactly one rule $A \to \alpha$.
\par
A \emph{derivation} between strings is a relation $u \, A \, v \Rightarrow u \, w \, v$, with $u$, $w$ and $v$ possibly empty strings, and $w \in R \, (\alpha)$. The derivation is \emph{leftmost} (respectively \emph{rightmost}), if $u$ does not contain a nonterminal (resp. if $v$ does not contain a nonterminal). A series of derivations is denoted by $\stackrel \ast \Rightarrow$. A derivation $A\stackrel \ast \Rightarrow A \, v$, with $v \neq \varepsilon$, is called \emph{left-recursive}. For a derivation $u \Rightarrow v$, the reverse relation is named \emph{reduction} and denoted by $v \leadsto u$.
\par
The language $L \, (G)$ generated by grammar $G$ is the set  $L \, (G) = \{ \; x \in \Sigma^\ast \; \vert \quad S \stackrel \ast \Rightarrow x \; \}$. The language $L \, (A)$ generated by nonterminal $A$ is the set $L_A \, (G) = \{ \; x \in \Sigma^\ast \; \vert \quad A \stackrel \ast \Rightarrow x \; \}$. A language is \emph{nullable} if it contains the empty string. A nonterminal that generates a nullable language is also called nullable.
\par
Following a tradition dating at least to \cite{journals/jacm/Lomet73}, we are going to represent a grammar rule $A \to \alpha$ as a graph: the state-transition graph of a finite automaton to be named \emph{machine}\footnote{To avoid confusion we call ``machines'' the finite automata of grammar rules, and we reserve the term ``automaton'' for the pushdown automaton that accepts language $L \, (G)$.} $M_A$ that recognizes a regular expression $\alpha$. The collection $\mathcal{M}$ of all such graphs for the set $V$ of nonterminals, is named a \emph{network of finite machines} and is a graphic representation of a grammar (see Fig. \ref{figuraReteMacchineEsEsprAritm}). This has well-known advantages: it offers a pictorial representation, permits to directly handle $EBNF$ grammars and maps quite nicely on the recursive-descent parser implementation.
\par
In the simple case when $\alpha$ contains just terminal symbols, machine $M_A$ recognizes the language $L_A \, (G)$. But if $\alpha$ contains a nonterminal $B$, machine $M_A$ has a state-transition edge labeled with $B$. This can be thought of as the invocation of the machine $M_B$ associated to rule $B \to \beta$; and if nonterminals $B$ and $A$ coincide, the invocation is recursive.
\par
It is convenient although not necessary, to assume that the machines are deterministic, at no loss of generality since a nondeterministic finite state machine can always be made deterministic.
\begin{figure}[h!]
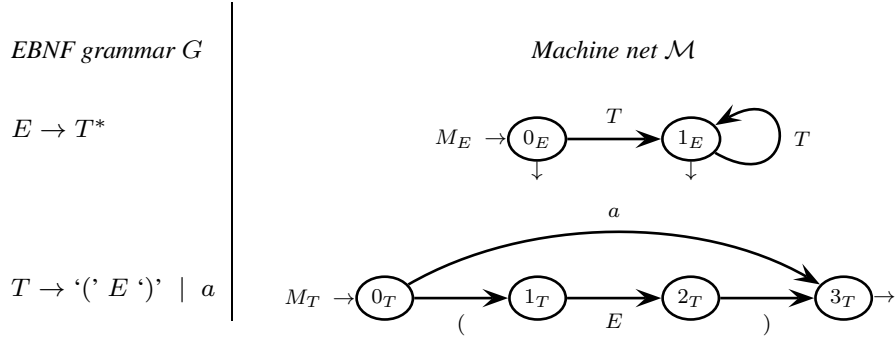

\begin{center}
\def\arraystretch{3}\begin{tabular}{l|p{1cm}c}
\scalebox{1.2}{\emph{EBNF grammar $G$}} && \scalebox{1.2}{\emph{Machine net $\mathcal{M}$}} \\

\scalebox{1.2}{$E \to T^\ast$}

&&

\pspar\psset{border=0pt,nodesep=0pt,labelsep=5pt,colsep=1.25cm}
\begin{psmatrix}
\ovalnode{0E}{$0_E$} & \ovalnode{1E}{$1_E$}

\nput[labelsep=0pt]{180}{0E}{$M_E \ \to$}

\nput[labelsep=0pt]{-90}{1E}{$\downarrow$}

\nput[labelsep=0pt]{-90}{0E}{$\downarrow$}

\ncline{0E}{1E} \naput{$T$}

\nccurve[angleA=-30,angleB=30,ncurvA=7,ncurvB=7]{1E}{1E} \nbput{$T$}
\end{psmatrix}

\\ \\

\scalebox{1.2}{$T \to \text{`$($'} \; E \; \text{`$)$'} \; \mid \; a$}

&&

\pspar\psset{border=0pt,nodesep=0pt,labelsep=5pt,,colsep=1.25cm}
\begin{psmatrix}
\ovalnode{0T}{$0_T$} & \ovalnode{1T}{$1_T$} & \ovalnode{2T}{$2_T$} & \ovalnode{3T}{$3_T$}

\nput[labelsep=0pt]{180}{0T}{$M_T \ \to$}

\nput[labelsep=0pt]{0}{3T}{$\to$}

\ncline{0T}{1T} \nbput{$($}

\ncline{1T}{2T} \nbput{$E$}

\ncline{2T}{3T} \nbput{$)$}

\ncarc{0T}{3T} \naput{$a$}
\end{psmatrix}

\end{tabular}
\end{center}
\caption{$EBNF$ grammar $G$ (axiom $E$) and machine network $\mathcal{M} = \set{ \, M_E, \, M_T \, }$ (axiom $M_E$) of the running example \ref{exRunningELRpart1}; initial (respectively final) states are tagged by an incoming (resp. outgoing) dangling dart.} \label{figRunningExELRpart1} \label{figuraReteMacchineEsEsprAritm}
\end{figure}
\begin{definition}Recursive net of finite deterministic machines. \label{retiRicorsMacch}
\begin{itemize}
\item Let $G$ be an $EBNF$ grammar with nonterminal set $V = \{ \; S, \, A, \, B, \, \ldots \; \}$ and
grammar rules $S \to \sigma$, $A \to \alpha$, $B \to \beta$, \, \ldots
\item
Symbols $R_S$, $R_A$, $R_B$, \, \ldots \, denote the regular languages over alphabet $\Sigma \cup V$, respectively defined by the r.e. $\sigma$, $\alpha$, $\beta$, \, \ldots
\item
Symbols $M_S$, $M_A$, $M_B$, \, \ldots \, are the names of finite deterministic machines that accept the corresponding regular languages $R_S$, $R_A$, $R_B$, \, \ldots \, .  As usual, we assume the machines to be \emph{reduced}, in the sense that every state is reachable from the initial state and reaches a final state.  The set of all machines, i.e., the \emph{(machine) net}, is denoted by $\mathcal{M}$.
\item
To prevent confusion, the names of the states of any two machines are made different by appending the machine name as subscript. The set of states of machine $M_A$ is  $Q_A = \{ \; 0_A, \; \ldots, \; q_A, \; \ldots \; \}$, the initial state is $0_{A}$ and the set of final states is $F_A \subseteq Q_A$. The \emph{state set of the net} is the union of all states $Q = \bigcup_{M_A \in \mathcal{M}} Q_A$. The transition function of every machine is denoted by the same symbol $\delta$, at no risk of confusion as the state sets are disjoint.
\item
For state $q_A$ of machine $M_A$, the symbol $R \, (M_A, \, q_A)$ or for brevity $R \, (q_A)$, denotes the \emph{regular language} of alphabet $\Sigma \cup V$ accepted by the machine starting from state $q_A$. If $q_A$ is the initial state $0_A$, language $R \, (0_A) \equiv R_A$ includes every string that labels a path, qualified as \emph{accepting}, from the initial to a final state.
\item Normalization disallowing reentrance into initial states. To simplify the parsing algorithms, we stipulate that for every machine $M_A$
no edge $q_A \stackrel c \longrightarrow 0_A$ exists, where $q_A \in Q_A$ and $c$ is a grammar symbol. In words, no edges may enter the initial state.
\end{itemize}
\end{definition}
The above normalization ensures that the initial state $0_A$ is not visited again within a computation that stays inside machine $M_A$;  clearly, any machine can be so normalized by adding one state and a few transitions, with a negligible overhead. Such minor adjustment greatly simplifies the reduction moves of bottom-up parsers.
\par
For an arbitrary r.e. $\alpha$, several well-known algorithms such as MacNaughton-Yamada and Berry-Sethi (described for instance in \cite{Crespi09}) produce a machine recognizing the corresponding regular language. In practice, the r.e. used in the right parts of grammars are so simple that they can be immediately translated by hand into an equivalent machine. We neither assume nor forbid that the machines be  minimal with respect to the number of states. In facts, in syntax-directed translation it is not always desirable to use the minimal machine, because different semantic actions may be required in two states that would be indistinguishable by pure language theoretical definitions.
\par
If the grammar is purely $BNF$ then each right part has the form $\alpha = \alpha_1 \; \mid \; \alpha_2 \; \mid \; \ldots \; \mid \; \alpha_k$, where every alternative $\alpha_i$ is a finite string; therefore $R_A$ is a finite language and any machine for $R_A$ has an acyclic graph, which can be made into a tree if we accept to have a non-minimal machine. In the general case, the graph of machine $M_A$ representing rule $A \to \alpha$ is not acyclic. In any case there is a one-to-one mapping between the strings in language $R_A$ and the set of accepting paths of machine $M_A$. Therefore, the  net $\mathcal{M} = \set{ \; M_S, \; M_A, \; \ldots \; }$ is essentially a notational variant of a grammar $G$, as witnessed by the common practice to include both  $EBNF$ productions  and  syntax diagrams in language specifications.
We indifferently denote the language as $L \, (G)$ or as $L \, (\mathcal{M})$.
\par
We need also the context-free \emph{terminal} language defined by the net, starting from a state $q_A$ of machine $M_A$ possibly other than the initial one:
\[
L \, (M_A, \, q_A) \equiv L \, (q_A) = \set{ \; y \in \Sigma^\ast \; \vert \quad \eta \in R \, (q_A) \ \text{and} \ \eta \overset \ast \Rightarrow y \; }
\]
In the formula, $\eta$ is a string over terminals and nonterminals, accepted by machine $M_A$ starting from state $q_A$. The derivations originating from $\eta$ produce the terminal strings of language $L \, (q_A)$. In particular, from previous definitions it follows that:
\[
L \, (M_A, \, 0_A) \equiv L \, (0_A) =  L_{A} \, (G) \quad \text{and} \quad L \, (M_S, \, 0_S) \equiv L \, (0_S) \equiv L \, (\mathcal{M}) = L \, (G)
\]
\noindent
\begin{example}Running example.\label{exRunningELRpart1}\label{reteEsprAritm}
\par
The $EBNF$ grammar $G$ and machine net $\mathcal{M}$ are shown in Figure \ref{figRunningExELRpart1}.
The language generated can be viewed as obtained from the language of well-parenthesized strings, by allowing character $a$ to replace a well-parenthesized substring.
All machines are deterministic and the initial states are not reentered. Most features needed to exercise different aspects of parsing are present:  self-nesting, iteration, branching, multiple final states and the nullability of a nonterminal.
\par
To illustrate, we list a language defined by its net and component machines, along with their aliases:
\[
\def\arraystretch{1.5}
\begin{array}{rcl}
R \, (M_E, \, 0_E) & = & R \, (0_E) = T^\ast \\
R \, (M_T, \, 1_T) & = & R \, (1_T) =  E \; \Big) \\
L \, (M_E, \, 0_E) & = & L \, (0_E) = L \, (G) = L \, (\mathcal{M}) \\
& = & \set { \; \varepsilon, \; a, \; a\,a, \; ( \, ), \; a\,a\,a, \; ( \, a \, ), \; a \, ( \, ), \; ( \, ) \, a, \; ( \, ) \, ( \, ), \; \ldots \; } \\
L \, (M_T, \, 0_T) & = & L \, (0_T) = L_T \, (G) = \set{ \; a, \; ( \, ), \; ( \, a \, ), \; ( \, a \, a \, ), \; ( \, ( \, ) \, ), \; \ldots \; }
\end{array}
\]
\end{example}
\par
To identify machine states, an alternative  convention  quite used for $BNF$ grammars, relies on \emph{marked grammar rules}: for instance the states of machine $M_T$ have the aliases:
\begin{align*}
0_T & \equiv T \to \bullet \, a \; \mid \; \bullet \, ( \, E \,) &  1_T & \equiv T \to ( \, \bullet \, E \, ) \\
2_T & \equiv T \to ( \, E \, \bullet \, ) & 3_T & \equiv a \, \bullet \; \mid \; T \to ( \, E \, ) \, \bullet
\end{align*}
where the bullet is a character not in $\Sigma$.
\par
We need to define the set of \emph{initial} characters for the strings recognized starting from a given state.
\noindent
\begin{definition}\label{defInitials}Set of initials.
\[
Ini \, (q_A) =  Ini \, \big( \; L \, (q_A) \; \big) = \set{ \; a \in \Sigma \; \vert \quad a \; \Sigma^\ast \cap L \, (q_A) \neq \emptyset \; }
\]
The set can be computed by applying the following logical clauses until a fixed point is reached. Let $a$ be a terminal, $A$, $B$, $C$ nonterminals, and $q_A$, $r_A$ states of the same machine. The clauses are:
\begin{align*}
a \in Ini \, (q_A) & \quad \text{if} \ \exists \; \text{edge} \ q_A \stackrel a \longrightarrow r_A \\
a \in Ini \, (q_A) & \quad \text{if} \ \exists \; \text{edge} \ q_A \stackrel { B } \longrightarrow r_A \ \land \ a \in Ini \, (0_B) \\
a \in Ini \, (q_A) & \quad \text{if} \ \exists \; \text{edge} \ q_A \stackrel { B } \longrightarrow r_A \ \land \ L \, (0_B) \; \text{is nullable} \ \land \
a \in Ini \, (r_A)
\end{align*}
\end{definition}
To illustrate, we have $Ini \, (0_E) = Ini \, (0_T) = \{ \; a, \; ( \; \}$.
\subsection{Derivation for machine nets}\label{subSectDerivation}
For machine nets and $EBNF$ grammars, the preceding definition of derivation, which models a rule such as $E \to T^*$ as the infinite set of $BNF$ alternatives $E \to \varepsilon \; \mid \; T \; \mid \; T \, T  \; \mid \; \ldots$, has shortcomings because a derivation step, such as $E \Rightarrow T\,T\,T$, replaces a nonterminal $E$ by a string of possibly \emph{unbounded} length; thus a multi-step computation inside a machine is equated to just one derivation step. For application to parsing, a more analytical definition is needed to split such large step into a series of state transitions.
\par
We recall that a $BNF$ grammar is \emph{right-linear} ($RL$) if the rule form is $A \to a \, B$ or $A \to \varepsilon$, where $a \in \Sigma$ and $B \in V$. Every finite-state machine can be represented by an equivalent $RL$ grammar that has the machine states  as nonterminal symbols.
\subsubsection{Right-linearized grammar}\label{rightLinearizedGram}
Each machine $M_A$ of the net can be replaced with an equivalent right linear $RL$ grammar, to be used to provide a rigorous semantic for the derivations constructed by our parsers.
It is straightforward to write the $RL$ grammar, named $\hat{G}_A$, equivalent with respect to the regular language $R \, (M_A, \, 0_A)$. The nonterminals of $\hat{G}_A$ are the states of $Q_A$, and the axiom is  $0_A$; there exists a rule $p_A \to X \, r_A$ if an edge $p_A \stackrel X \longrightarrow r_A$ is in $\delta$, and the empty rule $p_A \to \varepsilon$ if $p_A$ is a final state.
\par
Notice that $X$ can be a nonterminal $B$ of the original grammar, therefore a rule of $\hat{G}_A$  may have the form $p_A \to B \, r_A $, which is still $RL$ since the first symbol of the right part is viewed  as a ``terminal'' symbol for grammar $\hat{G}_A$. With this provision, the identity $L \, \left( \, \hat{G}_A \, \right) = R \, (M_A, \, 0_A)$ clearly holds.
\par
Next, for every $RL$ grammar of the net we replace with $0_B$ each symbol $B \in V$ occurring in a rule such as $p_A \to  B \, r_A$, and thus we obtain rules of the form $ p_A \to 0_B \, r_A$. The resulting $BNF$ grammar, non-$RL$, is denoted $\hat{G}$ and named the \emph{right-linearized grammar} of the net: it has terminal alphabet $\Sigma$, nonterminal set $Q$ and axiom $0_S$. The right parts have length zero or two, and may contain two nonterminal symbols, thus the grammar is \textit{not} $RL$. Obviously, $\hat{G}$ and $G$ are equivalent, i.e., they both generate language $L \, (G)$.
\begin{example}\label{reteEsprAritmRightLin}Right-linearized grammar $\hat{G}$ of the running example.
\noindent
\begin{equation*}
\begin{array}{l l l l l}
\hat{G}_E \colon & 0_E \to 0_T 1_E \; \mid \; \varepsilon \qquad \qquad & 1_E \to 0_T \; 1_E \; \mid \; \varepsilon \qquad \\ \\
\hat{G}_T \colon & 0_T \to a \; 3_T \; \mid \; ( \; 1_T & 1_T \to 0_E \; 2_T & 2_T \to ) \; 3_T \qquad \qquad & 3_T \to \varepsilon
\end{array}
\end{equation*}
As said, in right-linearized grammars we choose to name nonterminals  by their alias states; 
an instance of non-$RL$ rule is $1_T \to 0_E \, 2_T$.
\end{example}
Using $\hat{G}$ instead of $G$, we obtain derivations the steps of which are elementary state transitions instead of an entire sub-computation on a machine. An example should suffice.
\noindent
\begin{example}\label{exLeftDerivForNet}Derivation.
\par
For grammar $G$ the ``classical'' leftmost derivation:
\begin{equation}\label{eqClassicDerivation}
E \underset G \Rightarrow T \, T \underset G \Rightarrow a \, T \underset G \Rightarrow  a \, ( \, E \, ) \underset G \Rightarrow a \, ( \, )
\end{equation}
is expanded into the series of truly atomic derivation steps of the right-linearized grammar:
\begin{equation}\label{eqRLGderivation}
\begin{array}{rclclcl}
0_E & \underset {\hat{G}} \Rightarrow & 0_T \, 1_E & \underset {\hat{G}} \Rightarrow & a \, 1_E & \underset {\hat{G}} \Rightarrow & a \, 0_T \, 1_E \\
& \underset {\hat{G}} \Rightarrow & a \, ( \, 1_T \, 1_E & \underset {\hat{G}} \Rightarrow & a \, ( \, 0_E \, 2_T \, 1_E & \underset {\hat{G}} \Rightarrow & a \, ( \, \varepsilon \, 2_T \, 1_E \\
& \underset {\hat{G}} \Rightarrow & a \, ( \, ) \, 3_T \, 1_E & \underset {\hat{G}} \Rightarrow & a \, ( \, ) \, \varepsilon \, 1_E & \underset {\hat{G}} \Rightarrow & a \, ( \, ) \, \varepsilon  \\
& = & a \, ( \, ) \\
\end{array}
\end{equation}
\end{example}
We may also work bottom-up and consider  reductions such as $ a \, (\, 1_T \, 1_E \leadsto  a \, 0_T \, 1_E$.
\par
As said, the right-linearized grammar is only used in our proofs to assign a precise semantic to the parser steps, but has otherwise no use as a readable specification of the language to be parsed. Clearly, an $RL$ grammar has many more rules than the original $EBNF$ grammar and is  less readable than a syntax diagram or machine net, because it needs to introduce a plethora of nonterminal names to identify the machine states.
\subsection{Call sites, machine activation, and look-ahead} \label{callSitesActivationLookahead}
An edge labeled with a nonterminal, $q_A \stackrel {B} \longrightarrow r_A$, is named a \emph{call site} for machine $M_B$, and $r_A$ is the corresponding \emph{return state}. Parsing can be viewed as a process that at call sites activates a machine, which on its own graph performs scanning operations and further calls until it reaches a final state, then performs a reduction and returns. Initially the axiom machine, $M_S$, is activated by the program that invokes the parser. At any step in the derivation, the nonterminal suffix of the derived string contains the current state of the active machine followed by the return points of the suspended machines; these are ordered from right to left according to their activation sequence.
\par
\begin{example}\label{derivAndRetPoints}Derivation and machine return points.
\par
Looking at derivation \eqref{eqRLGderivation}, we find
$0_E \stackrel * \Rightarrow a \, ( \, 0_E \, 2_T \, 1_E$: machine $M_E$ is active and its current state is $0_E$; previously machine $M_T$ was suspended and will resume in state $2_T$; an earlier activation of $M_E$ was also suspended and will resume in state $1_E$.
\end{example}
Upon  termination of $M_B$  when $M_A$ resumes in return state $r_A$, the collection of all the first legal tokens that can be scanned is named the \emph{look-ahead set} of this activation of $M_B$; this intuitive concept is made more precise in the following definition of a candidate. By inspecting the next token, the parser can avoid invalid machine call actions. For uniformity, when the input string has been entirely scanned, we assume the next token to be the special character $\dashv$ (string terminator or end-marker).
\par
A \emph{$1$-candidate} (or $1$-item), or simply a \emph{candidate} since we exclusively deal with look-ahead length $1$, is a pair $\langle q_B, \, a \rangle$ in $Q \times \left( \, \Sigma \cup \{ \dashv \} \, \right)$. The intended meaning is that token $a$ is a legal look-ahead token for the current activation of machine $M_B$. We reformulate the classical \cite{Knuth1965LR(k)} notion of closure function for a machine net and we use it to compute the set of legal candidates.
\noindent
\begin{definition}\label{defLookaheadSet}Closure functions.
\par
The initial activation of machine $M_S$ is encoded by candidate $\langle 0_S, \, \dashv \rangle$.
\par
Let $C$ be a set candidates initialized to $\{\langle 0_S, \, \dashv \rangle\}$. The closure of $C$ is the function defined by applying the following clauses until a fixed point is reached:
\begin{align}
c & \in closure \, (C) \quad \text{if} \ c \in C \nonumber \\
\label{eqDef1-closure}
\langle 0_B, \, b \rangle & \in closure \, (C) \quad \text{if} \ \exists \; \langle q, \, a \rangle \in closure \, (C) \ \text{and} \ \exists \; \text{edge} \; q \stackrel {B} \longrightarrow r \; \text{in} \; \mathcal{M} \\
\phantom{\langle 0_B, \, b \rangle} & \phantom{\in closure \, (C)} \quad \ \text{and} \ b \in {Ini} \, \big( \; L \, (r) \cdot a \; \big) \nonumber
\end{align}
\end{definition}
Thus closure functions  compute the set of machines reachable from a given call site through one or more invocations, without any intervening state transition.
\par
For conciseness, we group together the candidates that have the same state and we write $\big\langle q, \, \{ \, a_1, \, \ldots, \, a_k \, \} \big\rangle$ instead of $\{ \, \langle q, \, a_1 \rangle, \, \ldots, \, \langle q, \, a_k \rangle \, \}$. The collection $\set{ \, a_1, \, a_2, \, \ldots, \, a_k \, }$ is termed \emph{look-ahead set} and by definition it cannot be empty.
\par
We list a few values of closure function for the  grammar of Ex. \ref{exRunningELRpart1}:
\begin{center}
$\def\arraystretch{1.5}\begin{array}{l|ccc}
& \multicolumn{3}{l}{\text{function} \; closure} \\ \hline	
\langle 0_E, \, \dashv \rangle & \langle 0_E, \, \dashv \rangle \qquad & \big\langle 0_T, \set{ \, \dashv, \, a, \, ( \, } \big\rangle \\
\langle 1_T, \, \dashv \rangle & \langle 1_T, \dashv \rangle \qquad & \big\langle 0_E, \, ) \, \big\rangle \qquad & \big\langle 0_T, \{ \, a, \, (, \, ) \, \} \big\rangle \\
\end{array}$
\end{center}
\section{Shift-reduce parsing} \label{sectionShiftReduce}
We show how to construct deterministic bottom-up parsers directly for $EBNF$ grammars represented by machine nets. As we deviate from the classical Knuth's method, which operates on pure $BNF$ grammars, we call our method $ELR \, (1)$ instead of $LR \, (1)$. For brevity, whenever some passages are identical or immediately obtainable from classical ones, we do not spend much time to justify them. On the other hand, we  include correctness proofs of the main constructions because  past works on Extended $BNF$ parsers have been found not rarely to be flawed \cite{conf/lata/Hemerik09}. At the end of this section we briefly compare our method with  older ones.
\par
An $ELR \, (1)$ parser is a deterministic pushdown automaton ($DPDA$) equipped with a set of states named \textit{macrostates} (for short m-states) to avoid confusion with net states. An m-state consists of a set of $1$-candidates (for brevity candidate).
\par
The automaton performs moves of two types. A \emph{shift} action reads the current input character (i.e., a token) and applies the $PDA$ state-transition function to compute the next m-state; then the token and the next m-state are pushed on the stack. A \emph{reduce} action is applied when the grammar symbols from the stack top match a recognizing path on a machine $M_A$ and the current token is admitted by the look-ahead set. For the parser to be deterministic in any configuration, if a shift is permitted then reduction should be impossible, and in any configuration at most one should be possible. A reduce action grows bottom-up the syntax forest, pops the matched part of the stack, and pushes the nonterminal symbol recognized and the next m-state. The $PDA$ accepts the input string if the last move reduces to the axiom and the input is exhausted; the latter condition can be expressed by saying that the special end-marker character $\dashv$ is the current token.
\par
The presence of convergent paths in a machine graph complicates reduction moves because two such paths may require to pop  different stack segments (so-called reduction handles); this difficulty is acknowledged in  past research on  shift-reduce methods for EBNF grammars, but the proposed solutions differ in the technique used and in generality. To implement such reduction moves, we enrich the stack organization with pointers, which enable the parser to trace back a recognizing path while popping the stack. Such a pointer can be implemented in two ways: as a bounded integer offset that identifies a candidate in the previous stack element, or as an unbounded integer pointer to a distant stack element. In the former case the parser still qualifies as $DPDA$ because the stack symbols are taken from a a finite set. Not so in the latter case, where the pointers are unbounded integers; this organization is an indexable stack to be called a \emph{vector-stack}  and will be also used by Earley parsers.
\subsection{Construction of $\emph{ELR} \, (1)$ parsers}
Given an $EBNF$ grammar represented by a  machine net, we show how to construct an $ELR \, (1)$ parser if certain conditions are met. The method operates in three phases:
\begin{enumerate}
\item From the net we construct a $DFA$, to be called a \emph{pilot} automaton.\footnote{Its traditional lengthy name is ``recognizer
of viable $LR \, (1)$ prefixes''.} A pilot state, named \emph{macro-state} (m-state),  includes  a non empty set of candidates, i.e., of  pairs of  states and terminal tokens (look-ahead set).
\item The pilot is examined to check the conditions for deterministic bottom-up parsing; the check involves an inspection of the components
of each m-state and of the transitions outgoing from it. Three types of failures may occur: \emph{shift-reduce} or \emph{reduce-reduce} conflicts, respectively signify that in a parser configuration both a shift and a reduction are possible or multiple reductions; a \emph{convergence} conflict occurs when two different parser computations that share a look-ahead character, lead to the same machine state.
\item If the test is passed, we construct the deterministic $PDA$, i.e., the parser, through using the pilot $DFA$ as its finite-state control and adding  the operations needed for managing reductions of unbounded length.
\end{enumerate}
At last, it would be a simple  exercise to encode the $PDA$ in a programming language.
\par
For a candidate $\left\langle p_A, \, \rho \right\rangle$ and a terminal or nonterminal symbol $X$, the \emph{shift}\footnote{Also known as ``go to'' function.} under $X$ (qualified as \emph{terminal/nonterminal} depending on $X$ being a terminal or non- ) is:
\[\def\arraystretch{1.5}
\begin{cases}
\vartheta \, \left( \, \left\langle p_A, \, \rho \right\rangle, \, X \, \right) = \left\langle q_A, \, \rho \right\rangle & \text{if edge} \ p_A \stackrel {X} {\longrightarrow} q_A \ \text{exists} \\
\text{the empty set} & \text{otherwise}
\end{cases}
\]
For a set $C$ of candidates, the shift under a symbol $X$ is the union of the shifts of the candidates in $C$.
\begin{algorithm}\label{algMacchPilotaELR(1)}Construction of the $ELR \, (1)$ pilot graph.
\par
The pilot is the $DFA$, named $\mathcal{P}$, defined by:
\begin{itemize}
\item the set $R$ of m-states
\item the \emph{pilot alphabet} is the union $\Sigma \cup V$ of the terminal and nonterminal alphabets
\item the initial m-state $I_0$ is the set: $I_0 = closure \, (\left\langle 0_S, \, \dashv \right\rangle)$
\item the  m-state set $R = \{ \, I_0, \, I_1, \, \ldots \, \}$ and the state-transition function\footnote{Named theta, $\vartheta$, to avoid confusion
with the traditional name delta, $\delta$, of the transition function of the net.} $\vartheta \colon R \times \left( \Sigma \cup V \right) \to R$ are computed starting from $I_0$ by the following steps:
\end{itemize}
\begin{tabbing}
\hspace{0.75cm} \= \hspace{0.75cm} \= \hspace{0.75cm} \= \hspace{0.75cm} \= \hspace{0.75cm} \= \kill
\> $R' : = \set{ \, I_0 \, }$ \\
\> \textbf{do} \\
\>\> $R := R'$ \\
\>\> \textbf{for} \ $\big( \, \text{all the m-states $I \in R$ and symbols $X \in \Sigma \cup V$} \, \big)$ \ \textbf{do} \\
\>\>\> $I' := closure \, \big( \vartheta \, \left(I, \, X \right) \big)$ \\
\>\>\> \textbf{if} \ $\big( \, I' \neq \emptyset \, \big)$ \ \textbf{then} \\
\>\>\>\> add the edge $I \stackrel X \longrightarrow I'$ to the graph of $\vartheta$ \\
\>\>\>\> \textbf{if} \ $\big( \, I' \not \in R \, \big)$ \ \textbf{then} \\
\>\>\>\>\> add the m-state $I'$ to the set $R'$ \\
\>\>\>\> \textbf{end if} \\
\>\>\> \textbf{end if} \\
\>\> \textbf{end for} \\
\> \textbf{while} $\big( \, R \neq R' \, \big)$  \\
\end{tabbing}
\end{algorithm}
\subsubsection{Base closure and kernel of a m-state}
For every m-state $I$ the set of candidates is partitioned into two subsets: base and closure. The \emph{base} includes the non-initial candidates:
\[
I_{ \vert base} = \set{ \, \langle q, \, \pi \rangle \in I \; \vert \quad q \ \text{is not an initial state} \, }
\]
Clearly, for the m-state $I'$ computed at line $3$ of the algorithm, the base $I'_{ \vert base}$ coincides with the pairs computed by $\vartheta \, (I, \, X)$.
\par
The \emph{closure} contains the remaining candidates of m-state $I$:
\[
I_{ \vert closure} = \set{ \, \langle q, \, \pi \, \rangle \in I \; \vert \quad q \ \text{is  an initial state} \, }
\]
The initial m-state $I_0$ has an empty base by definition. All other m-states have a non-empty base, while the closure may be empty.
\par
The \emph{kernel} of a m-state is the projection on the first component:
\[
I_{ \vert kernel} = \set{ \, q \in Q \; \vert \quad \left\langle q, \, \pi \right\rangle \in I \, }
\]
A particular condition that may affect determinism occurs when, for two states that belong to the same m-state $I$, the outgoing transitions are defined under the same grammar symbol.
\par
\begin{definition}\label{defConvergence}Multiple Transition Property and Convergence.
\par
A pilot m-state $I$ has the \emph{multiple transition property} ($MTP$) if it includes two candidates $\langle q, \, \pi \rangle$ and $\langle r, \, \rho \rangle$, such that for some grammar symbol $X$ both transitions $\delta \, (q, \, X)$ and $\delta \, (r, \, X)$ are defined.
\par
Such a m-state $I$ and the pilot transition $\vartheta \, (I, \, X)$ are called \emph{convergent} if $\delta \, (q, \, X) = \delta \, (r, \, X)$.
\par
A convergent transition has a \emph{convergence conflict} if the look-ahead sets overlap, i.e., if $\pi \cap \rho \neq \emptyset$.
\end{definition}
To illustrate, we consider two examples.
\begin{example}Pilot of the running example.\label{exRunningPilot}
\par
The pilot graph $\mathcal{P}$ of the $EBNF$ grammar and net of Example \ref{exRunningELRpart1} (see Figure \ref{figRunningExELRpart1}) is shown in Figure \ref{figRunningExELRpilot}. In each m-state the top and bottom parts contain the base and the closure, respectively; when either part is missing, the double-line side of the m-state shows which part is present (e.g., $I_0$ has no base and $I_2$ has no closure); the look-ahead tokens are grouped by state; and final states are evidenced by encircling.
\begin{figure}[h!]
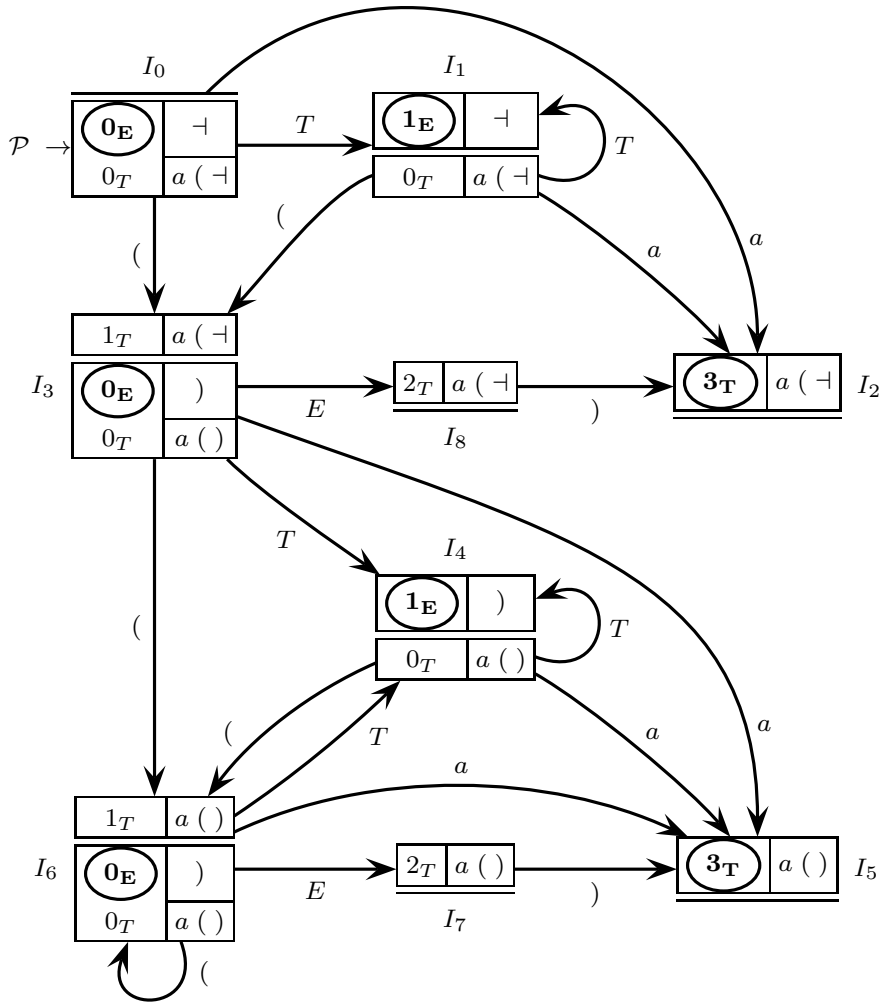

\begin{center}
\vspace{1.5cm}
\begin{minipage}{1.0\textwidth}
\begin{center}
\scalebox{1.2}{
\pspar\psset{arrows=->,border=0pt,nodesep=0pt,rowsep=1.25cm, colsep=1.5cm}
\begin{psmatrix}
\rnode{I0}{$\def\arraystretch{1.25}\begin{array}{|c|c|} \hline\hline
\ovalnode{\relax}{\mathbf{0_E}} & \dashv \\ \cline{2-2}
0_T & a \; ( \; \dashv \\ \hline
\end{array}$}
&
\rnode{I1}{$\def\arraystretch{1.25}\begin{array}{|c|c|} \hline
\ovalnode{\relax}{\mathbf{1_E}} & \dashv \\ \hline\hline
0_T & a \; ( \; \dashv \\ \hline
\end{array}$} \\

\rnode{I3}{$\def\arraystretch{1.25}\begin{array}{|c|c|} \hline
1_T & a \; ( \; \dashv \\
\hline \hline \ovalnode{o}{\mathbf{0_E}} & ) \\ \cline{2-2}
0_T & a \; ( \; ) \\ \hline
\end{array}$}
&
\rnode{I8}{$\def\arraystretch{1.25}\begin{array}{|c|c|} \hline
2_T & a \; ( \; \dashv \\ \hline\hline
\end{array}$}
&
\rnode{I2}{$\def\arraystretch{1.25}\begin{array}{|c|c|} \hline
\ovalnode{\relax}{\mathbf{3_T}} & a \; ( \; \dashv \\ \hline\hline
\end{array}$} \\
&
\rnode{I4}{$\def\arraystretch{1.25}\begin{array}{|c|c|} \hline
\ovalnode{\relax}{\mathbf{1_E}} & ) \\ \hline\hline
0_T & a \; ( \; ) \\ \hline
\end{array}$} \\

\rnode{I6}{$\def\arraystretch{1.25}\begin{array}{|c|c|} \hline
1_T & a \; ( \; ) \\ \hline \hline
\ovalnode{\relax}{\mathbf{0_E}} & ) \\ \cline{2-2}
0_T & a \; ( \; ) \\ \hline
\end{array}$}
&
\rnode{I7}{$\def\arraystretch{1.25}\begin{array}{|c|c|} \hline
2_T & a \; ( \; ) \\ \hline\hline
\end{array}$}
&
\rnode{I5}{$\def\arraystretch{1.25}\begin{array}{|c|c|} \hline
\ovalnode{\relax}{\mathbf{3_T}} & a \; ( \; ) \\ \hline\hline
\end{array}$}

\nput[labelsep=5pt]{90}{I0}{$I_0$}

\nput[labelsep=5pt]{90}{I1}{$I_1$}

\nput[labelsep=5pt]{0}{I2}{$I_2$}

\nput[labelsep=5pt]{180}{I3}{$I_3$}

\nput[labelsep=5pt]{90}{I4}{$I_4$}

\nput[labelsep=5pt]{0}{I5}{$I_5$}

\nput[labelsep=5pt]{180}{I6}{$I_6$}

\nput[labelsep=5pt]{-90}{I7}{$I_7$}

\nput[labelsep=5pt]{-90}{I8}{$I_8$}

\nput[labelsep=0pt]{180}{I0}{$\mathcal{P} \; \to$}

\ncline{I0}{I1} \naput{$T$}

\nccurve[angleA=45,angleB=90,ncurvA=1,ncurvB=0.8]{I0}{I2} \aput(0.85){$a$}

\nccurve[angleA=-30,angleB=130,ncurvA=0.5,ncurvB=0.5]{I1}{I2} \naput{$a$}

\nccurve[angleA=-20,angleB=20,ncurvA=3,ncurvB=3]{I1}{I1} \nbput{$T$}

\ncline{I0}{I3} \nbput{$($}

\nccurve[angleA=-160,angleB=45,ncurvA=0.5,ncurvB=0.5]{I1}{I3} \nbput{$($}

\ncline{I3}{I8} \nbput{$E$}

\ncline{I8}{I2} \nbput{$)$}

\nccurve[angleA=-20,angleB=90,ncurvA=1,ncurvB=0.8]{I3}{I5} \aput(0.85){$a$}

\nccurve[angleA=-30,angleB=130,ncurvA=0.5,ncurvB=0.5]{I4}{I5} \naput{$a$}

\nccurve[angleA=-45,angleB=145,ncurvA=0.5,ncurvB=0.5]{I3}{I4} \nbput{$T$}

\ncline{I3}{I6} \nbput{$($}

\nccurve[angleA=-20,angleB=20,ncurvA=3,ncurvB=3]{I4}{I4} \nbput{$T$}

\ncarc[arcangle=-5]{I6}{I4} \bput(0.75){$T$}

\ncarc[arcangle=-15]{I4}{I6} \bput(0.75){$($}

\nccurve[angleA=-70,angleB=-110,ncurvA=3,ncurvB=3]{I6}{I6} \aput(0.15){$($}

\ncline{I6}{I7} \nbput{$E$}

\ncarc[arcangle=25]{I6}{I5} \naput{$a$}

\ncline{I7}{I5} \nbput{$)$}

\end{psmatrix}}
\end{center}
\end{minipage}
\vspace{0.5cm}
\end{center}
\caption{$ELR \, (1)$ pilot graph $\mathcal{P}$ of the machine net in Figure \ref{figRunningExELRpart1}.}\label{figRunningExELRpilot}
\end{figure}
None of the edges of the graph is convergent.
\end{example}
\par
Two m-states, such as $I_1$ and $I_4$, having the same kernel, i.e., differing just for some look-ahead sets, are called \emph{kernel-equivalent}. Some simplified parser constructions to be later introduced, rely on kernel equivalence to reduce the number of m-states. We observe that for any two kernel-equivalent m-states $I$ and $I'$, and for any grammar symbol $X$, the m-states $\vartheta \, (I, \, X)$ and $\vartheta \, (I', \, X)$ are either both defined or neither one, and are kernel-equivalent.
\par
To illustrate the notion of convergent transition with and without conflict, we refer to Figure \ref{figPilotWithConvergence}, where m-states $I_{10}$ and $I_{11}$ have convergent transitions, the latter with a conflict.
\subsection{$\emph{ELR} \, (1)$ condition}
The presence of a final candidate in a m-state tells the parser that a reduction move ought to be considered. The look-ahead set specifies which tokens should occur next, to confirm the decision to reduce. For a machine net, more than one reduction may be applied in the same final state, and to choose the correct one, the parser stores additional information in the stack, as later explained. We formalize the conditions ensuring that all parser decisions are deterministic.
\begin{definition}$ELR \, (1)$ condition.\label{defELR1condition}
\par
A grammar or machine net  meets condition $ELR \, (1)$ if the corresponding pilot satisfies the following conditions:
\begin{description}
\item[Condition 1 ] Every m-state $I$ satisfies the next two clauses:
\par
no \emph{shift-reduce conflict}:
\begin{equation}\label{eqNoReduceShiftConfl}
\text{for all candidates $\langle q, \, \pi \rangle \in I$ s.t. $q$ is final and for all edges $I \stackrel a \longrightarrow I' \colon a \not \in \pi$}
\end{equation}
no \emph{reduce-reduce conflict}:
\begin{equation}\label{eqNoReduceReduceConfl}
\text{for all candidates $\langle q, \, \pi \rangle, \; \langle r, \, \rho \rangle \in I$ s.t. $q$ and $r$ are final $\colon \pi \cap \rho = \emptyset$}
\end{equation}
\item[Condition 2 ] No transition of the pilot graph has a convergence conflict.
\end{description}
\end{definition}
The pilot in Figure \ref{figRunningExELRpilot} meets  conditions \eqref{eqNoReduceShiftConfl} and \eqref{eqNoReduceReduceConfl}, and no edge of $\vartheta$ is convergent.
\subsubsection{$\emph{ELR} \, (1)$ versus classical $\emph{LR} \, (1)$ definitions}
First, we discuss the relation between this definition and the classical one \cite{Knuth1965LR(k)} in the case that the grammar is $BNF$, i.e., each nonterminal $A$ has finitely many alternatives $A \to \alpha \; \mid \; \beta \; \mid \; \ldots$. Since the alternatives do not contain star or union operations, a straightforward (nondeterministic) $NFA$ machine, $N_A$,  has an acyclic graph shaped as a tree with as many legs originating from the initial state $0_A$ as there are alternative rules for $A$. Clearly the graph of $N_A$ satisfies the no reentrance hypothesis for the initial state. In general $N_A$ is not minimal. No m-state in the classical $LR \, (1)$ pilot machine can exhibit the multiple transition property, and the only requirement for parser determinism comes from clauses \eqref{eqNoReduceShiftConfl} and \eqref{eqNoReduceReduceConfl} of Def. \ref{defELR1condition}.
\par
In our representation, machine $M_A$ may differ from $N_A$ in two ways. First, we assume that machine $M_A$ is deterministic, out of convenience not of necessity. Thus consider a nonterminal $C$ with two alternatives, $C \to \textit{if} \ E \ \textit{then} \ I \; \mid \; \textit{if} \ E \ \textit{then} \ I \ \textit{else} \ I$. Determinization has the effect of normalizing the alternatives by left factoring the longest common prefix, i.e., of using the equivalent $EBNF$ grammar $C \to \ \textit{if} \ E \ \textit{then} \ I \left( \, \varepsilon \; \mid \; \textit{else} \ I \, \right)$.
\par
Second, we allow (and actually recommend) that the graph of $M_A$ be minimal with respect to the number of states. In particular, the final states
of $N_A$ are merged together by state reduction if they are \emph{undistinguishable} for the $DFA$ $M_A$. Also a non-final state of $M_A$ may correspond to multiple states of $N_A$.
\par
Such state reduction may cause some pilot edges to become convergent. Therefore, in addition to checking conditions \eqref{eqNoReduceShiftConfl} and \eqref{eqNoReduceReduceConfl}, Def. \ref{defELR1condition} imposes that any convergent edge be free from conflicts. Since this point is quite subtle, we illustrate it by the next example.
\begin{example}State-reduction and convergent transitions.\label{exConvergenceProperty}
\par
Consider the equivalent $BNF$ and $EBNF$ grammars and the corresponding machines in Fig. \ref{figNetWithConvergence}.
\begin{figure}[h!]
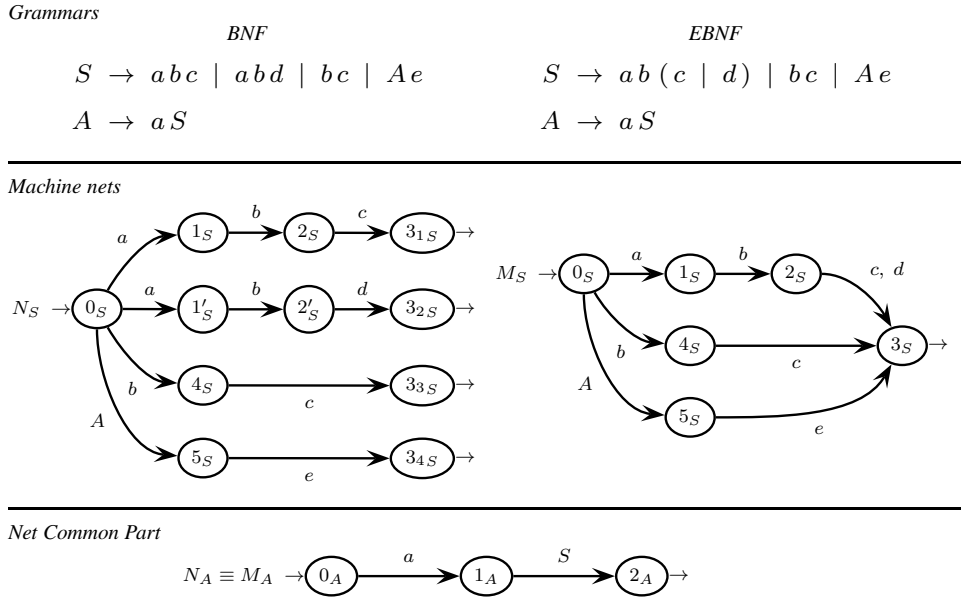

\begin{center}
\begin{flushleft}
\emph{Grammars}
\end{flushleft}
\begin{tabular}{cp{0.5cm}c}
\emph{BNF} && \emph{EBNF} \\

\scalebox{1.2}{$\def\arraystretch{1.5}\begin{array}{rcl}
S & \to & a\,b\,c \; \mid \; a\,b\,d \; \mid \; b\,c \; \mid \; A\,e \\
A & \to & a\,S
\end{array}$}

&&

\scalebox{1.2}{$\def\arraystretch{1.5}\begin{array}{rcl}
S & \to & a\,b\,\left(\,c \; \mid \; d\,\right) \; \mid \; b\,c \; \mid \; A\,e \\
A & \to & a\,S
\end{array}$}
\end{tabular}
\vspace{0.25cm}
\hrule
\begin{flushleft}
\emph{Machine nets}
\end{flushleft}
\vspace{0.0cm}
\begin{tabular}{cp{0.5cm}c}
\begin{minipage}{0.45\textwidth}
\qquad \scalebox{0.9}{
\pspar\psset{border=0pt,nodesep=0pt,labelsep=5pt,colsep=23pt,rowsep=12pt}
\begin{psmatrix}
& \ovalnode{1S}{$1_S$} & \ovalnode{2S}{$2_S$} & \ovalnode{3.1}{$3{_1}_S$} & & \\

\ovalnode{0S}{$0_S$} & \ovalnode{1'S}{$1'_S$} & \ovalnode{2'S}{$2'_S$} & \ovalnode{3.2}{$3{_2}_S$} \\

& \ovalnode{4S}{$4_S$}  && \ovalnode{3.3}{$3{_3}_S$} \\

& \ovalnode{5S}{$5_S$} && \ovalnode{3.4}{$3{_4}_S$}

\nput[labelsep=0pt]{180}{0S}{$N_S \ \to$}

\nput[labelsep=0pt]{0}{3.1}{$\to$}

\nput[labelsep=0pt]{0}{3.2}{$\to$}

\nput[labelsep=0pt]{0}{3.3}{$\to$}

\nput[labelsep=0pt]{0}{3.4}{$\to$}

\nccurve[angleA=60,angleB=-175]{0S}{1S} \naput{$a$}

\ncline{0S}{1'S} \naput{$a$}

\nccurve[angleA=-60,angleB=175]{0S}{4S} \bput(0.6){$b$}

\nccurve[angleA=-90,angleB=175]{0S}{5S} \nbput{$A$}

\ncline{1S}{2S} \naput{$b$}

\ncline{1'S}{2'S} \naput{$b$}

\ncline{2S}{3.1} \naput{$c$}

\ncline{2'S}{3.2} \naput{$d$}

\ncline{4S}{3.3} \nbput{$c$}

\ncline{5S}{3.4} \nbput{$e$}
\end{psmatrix}}
\end{minipage}

&&

\begin{minipage}{0.45\textwidth}
\scalebox{0.9}{
\pspar\psset{border=0pt,nodesep=0pt,labelsep=5pt,colsep=23pt,rowsep=12pt}
\begin{psmatrix}

\ovalnode{0S}{$0_S$} & \ovalnode{1S}{$1_S$} & \ovalnode{2S}{$2_S$} \\

& \ovalnode{4S}{$4_S$} & & \ovalnode{3S}{$3_S$} \\

& \ovalnode{5S}{$5_S$}

\nput[labelsep=0pt]{180}{0S}{$M_S \ \to$}

\nput[labelsep=0pt]{0}{3S}{$\to$}

\ncline{0S}{1S} \naput{$a$}

\nccurve[angleA=-60,angleB=175]{0S}{4S} \bput(0.6){$b$}

\nccurve[angleA=-90,angleB=175]{0S}{5S} \nbput{$A$}

\ncline{1S}{2S} \naput{$b$}

\nccurve[angleA=0,angleB=120]{2S}{3S} \naput{$c, \; d$}

\ncline{4S}{3S} \nbput{$c$}

\nccurve[angleA=0,angleB=-120]{5S}{3S} \nbput{$e$}
\end{psmatrix}}
\end{minipage}
\end{tabular}
\vspace{0.35cm}
\hrule
\begin{flushleft}
\emph{Net Common Part}
\end{flushleft}
\vspace{0.0cm}
\begin{minipage}{0.9\textwidth}
\begin{center}
\scalebox{0.9}{
\pspar\psset{border=0pt,nodesep=0pt,labelsep=5pt}
\begin{psmatrix}
\ovalnode{0A}{$0_A$} & \ovalnode{1A}{$1_A$} & \ovalnode{2A}{$2_A$}

\nput[labelsep=0pt]{180}{0A}{$N_A \equiv M_A \ \to$}

\nput[labelsep=0pt]{0}{2A}{$\to$}

\ncline{0A}{1A} \naput{$a$}

\ncline{1A}{2A} \naput{$S$}

\end{psmatrix}}
\end{center}
\end{minipage}
\end{center}
\caption{$BNF$ and $EBNF$ grammars and networks $\set{ \, N_S, \, N_A \, }$ and $\set{ \, M_S, \, M_A \equiv N_A \, }$.} \label{figNetWithConvergence}
\end{figure}
After determinizing, the states $3{_1}_S$, $3{_2}_S$, $3{_3}_S$, $3{_4}_S$ of machine $N_S$ are equivalent and are merged into the state $3_S$ of machine $M_S$. Turning our attention to the $LR$ and $ELR$ conditions, we find that the $BNF$ grammar has a reduce-reduce conflict caused by the derivations below:
\[
S \Rightarrow A \; e \Rightarrow a \; S \; e \Rightarrow a \; b \; c \; e \quad \text{and} \quad
S \Rightarrow A \; e \Rightarrow a \; S \; e \Rightarrow a \; a \; b \; c \; e
\]
\begin{figure}[h!]
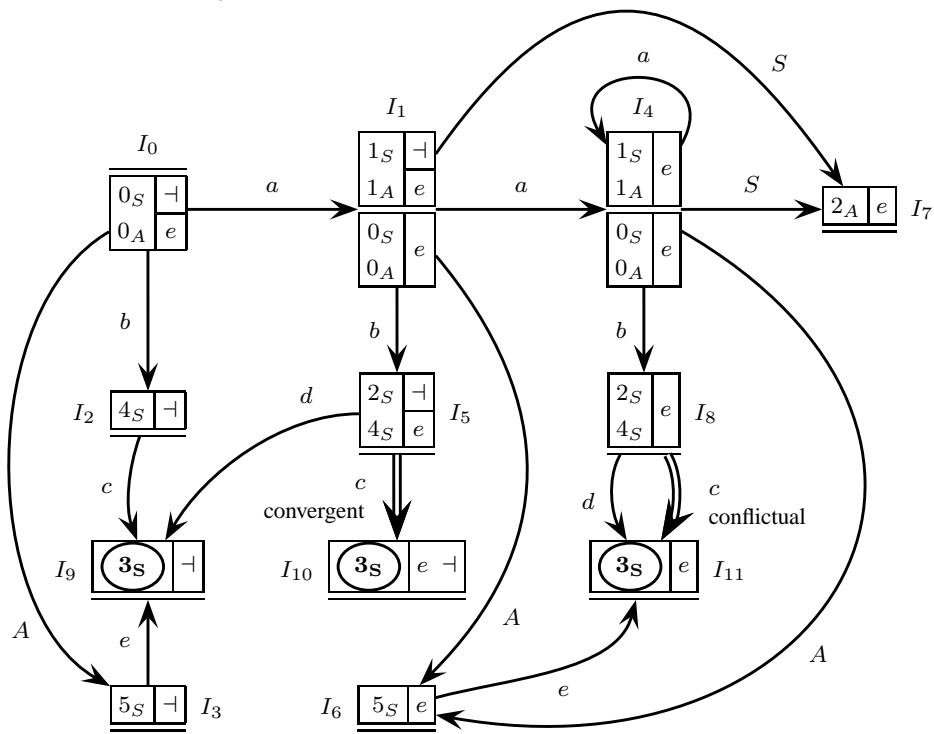

\begin{minipage}{1.0\textwidth}
\vspace{1.25cm}
\begin{center}
\scalebox{1.1}{
\pspar\psset{border=0pt,nodesep=0pt,labelsep=5pt, arrows=->, colsep=1.5cm, rowsep=1cm}
\begin{psmatrix}

\rnode{I0}{$\def\arraystretch{1.25}\begin{array}{|c|c|} \hline\hline
0_S & \dashv \\ \cline{2-2}
0_A & e \\ \hline
\end{array}$}

\nput{90}{I0}{$I_0$}

&

\rnode{I1}{$\def\arraystretch{1.25}\begin{array}{|c|c|} \hline
1_S & \dashv \\ \cline{2-2}
1_A & e \\ \hline\hline
0_S & \multirow{2}{*}{$e$} \\
0_A & \\\hline
\end{array}$}

\nput{90}{I1}{$I_1$}

&

\rnode{I4}{$\def\arraystretch{1.25}\begin{array}{|c|c|} \hline
1_S & \multirow{2}{*}{$e$} \\
1_A & \\ \hline\hline
0_S & \multirow{2}{*}{$e$} \\
0_A & \\ \hline
\end{array}$}

\nput{90}{I4}{$I_4$}

&

\rnode{I7}{$\def\arraystretch{1.25}\begin{array}{|c|c|} \hline
2_A & e \\ \hline\hline
\end{array}$}

\nput{0}{I7}{$I_7$}

\\

\rnode{I2}{$\def\arraystretch{1.25}\begin{array}{|c|c|} \hline
4_S & \dashv \\ \hline\hline
\end{array}$}

\nput{180}{I2}{$I_2$}

&

\rnode{I5}{$\def\arraystretch{1.25}\begin{array}{|c|c|} \hline
2_S & \dashv \\ \cline{2-2}
4_S & e \\ \hline\hline
\end{array}$}

\nput{0}{I5}{$I_5$}

&

\rnode{I8}{$\def\arraystretch{1.25}\begin{array}{|c|c|} \hline
2_S & \multirow{2}{*}{$e$} \\
4_S & \\ \hline\hline
\end{array}$}

\nput{0}{I8}{$I_8$}

\\

\rnode{I9}{$\def\arraystretch{1.25}\begin{array}{|c|c|} \hline
\ovalnode{\relax}{\mathbf{3_S}} & \dashv \\ \hline\hline
\end{array}$}

\nput{180}{I9}{$I_9$}

&

\rnode{I10}{$\def\arraystretch{1.25}\begin{array}{|c|c|} \hline
\ovalnode{\relax}{\mathbf{3_S}} & e \; \dashv \\ \hline\hline
\end{array}$}

\nput{180}{I10}{$I_{10}$}

&

\rnode{I11}{$\def\arraystretch{1.25}\begin{array}{|c|c|} \hline
\ovalnode{\relax}{\mathbf{3_S}} & e \\ \hline\hline
\end{array}$}

\nput{0}{I11}{$I_{11}$}

\\

\rnode{I3}{$\def\arraystretch{1.25}\begin{array}{|c|c|} \hline
5_S & \dashv \\ \hline\hline
\end{array}$}

\nput{0}{I3}{$I_3$}

&

\rnode{I6}{$\def\arraystretch{1.25}\begin{array}{|c|c|} \hline\
5_S & e \\ \hline\hline
\end{array}$}

\nput{180}{I6}{$I_6$}

\ncline{I0}{I1} \naput{$a$}

\ncline{I0}{I2} \nbput{$b$}

\nccurve[angleA=-150,angleB=150]{I0}{I3} \bput(0.8){$A$}

\ncline{I1}{I4} \naput{$a$}

\ncline{I1}{I5} \nbput{$b$}

\nccurve[angleA=-50,angleB=45,ncurvA=0.9,ncurvB=0.9]{I1}{I6} \aput(00.8){$A$}

\nccurve[angleA=55,angleB=125,ncurvA=1.25,ncurvB=1.25]{I1}{I7} \aput(0.8){$S$}

\nccurve[angleA=60,angleB=120,ncurvA=2.75,ncurvB=2.75]{I4}{I4} \nbput{$a$}

\ncline{I4}{I8} \nbput{$b$}

\nccurve[angleA=-30,angleB=-10,ncurv=1.5]{I4}{I6} \aput(0.6){$A$}

\ncline{I4}{I7} \naput{$S$}

\ncarc[arrowscale=1.25,doubleline=true,arcangle=30]{I8}{I11} \naput{$\begin{array}{l} c \\ \text{conflictual} \end{array}$}

\ncarc[arcangle=-30]{I8}{I11} \nbput{$d$}

\ncarc[arcangle=-20]{I2}{I9} \nbput{$c$}

\ncline{I3}{I9} \naput{$e$}

\nccurve[angleA=-180,angleB=60]{I5}{I9} \bput(0.2){$d$}

\ncline[arrowscale=1.25,doubleline=true]{I5}{I10} \nbput{$\begin{array}{r} c \\ \text{convergent} \end{array}$}

\nccurve[angleA=15,angleB=-105]{I6}{I11} \nbput{$e$}

\end{psmatrix}}
\end{center}
\vspace{0.25cm}
\end{minipage}
\caption{Pilot graph of machine net $\set{ \, M_S, \, M_A \, }$ in Fig. \ref{figNetWithConvergence}; the double-line edges are convergent.}\label{figPilotWithConvergence}
\end{figure}
\par
On the other hand, for the $EBNF$ grammar the pilot $\mathcal{P}_{\{ \, M_S, \, M_A \, \}}$, shown in Fig. \ref{figPilotWithConvergence}, has the two convergent edges highlighted as double-line arrows, one with a conflict and the other without. Arc $I_8 \stackrel c \to I_{11}$ violates the $ELR \,(1)$ condition because the look-aheads of $2_S$ and $4_S$ in $I_8$ are not disjoint.
\par
Notice in Fig.\ref{figPilotWithConvergence} that, for explanatory purposes, in m-state $I_{10}$ the two candidates $\langle 3_S, \, \dashv \rangle$ and $\langle 3_S, \, e \rangle$ deriving from a convergent transition with no conflict, have been kept separate and are not depicted as only one candidate with a single look-ahead, as it is usually done with candidates that have the same state.
\end{example}
We observe that in general a state-reduction in the machines of the net has the effect, in the pilot automaton, to transform reduce-reduce violations into convergence conflicts.
\par
Next, we prove the essential property that justifies the practical value of our theoretical development: an $EBNF$ grammar is $ELR \, (1)$ if, and only if, the equivalent right-linearized grammar (defined in Sect. \ref{rightLinearizedGram}) is $LR \, (1)$.
\begin{theorem}\label{theorELR1}
Let $G$ be an $EBNF$ grammar represented by a machine net $\mathcal{M}$ and let $\hat{G}$ be the equivalent right-linearized grammar. Then net $\mathcal{M}$ meets the $ELR \, (1)$ condition if, and only if, grammar $\hat{G}$ meets the $LR \, (1)$ condition.
\end{theorem}
The proof is in the Appendix.
\par
Although this proposition may sound  intuitively obvious to knowledgeable readers, we believe a formal proof is  due:  past proposals to extend $LR \, (1)$ definitions to the $EBNF$ (albeit often  with restricted types of regular expressions), having omitted formal proofs, have  been later found to be inaccurate (see  Sect. \ref{sectRelatedWorkELR}). The fact that shift-reduce conflicts are preserved by the two pilots is quite easy to prove. The less obvious part of the proof concerns the correspondence between convergence conflicts in the $ELR \, (1)$ pilot and reduce-reduce conflicts in the right-linearized grammar. To have a grasp without reading the proof, the convergence conflict in Fig. \ref{figPilotWithConvergence}  corresponds to the reduce-reduce conflict in the m-state $I_{11}$ of Fig. \ref{figMultiCandidatesInBaseRightLin}.
\par
We address a possible criticism to the significance of Theorem \ref{theorELR1}: that, starting from an $EBNF$ grammar, several equivalent $BNF$ grammars can be obtained by removing the regular expression operations in different ways. Such grammars may or may not be $LR \, (1)$, a fact that would seem to make somewhat arbitrary our definition of $ELR \, (1)$, which is based on the highly constrained left-linearized form. We defend the significance and generality of our choice on two grounds. First, our original grammar specification is not a set of r.e.'s, but a set of machines ($DFA$), and the choice to transform the $DFA$ into a right-linear grammar is standard and almost obliged because, as already shown by \cite{Heilbrunner1979}, the other standard form - left-linear - would exhibit conflicts in most cases. Second, the same author proves that if a right-linearized grammar equivalent to $G$ is $LR \, (1)$, then \emph{every} right-linearized grammar equivalent to $G$, provided it is not ambiguous, is $LR \, (1)$; besides,  he shows that this definition of $ELR \, (1)$ grammar dominates all the preexisting alternative definitions. We believe that also the new definitions of later years are dominated by the present one.
\par
To illustrate the discussion, it helps us consider a simple example where the machine net is $ELR \, (1)$, i.e., by Theorem \ref{theorELR1} $\hat{G}$ is $LR \, (1)$, yet another equivalent grammar obtained by a very natural transformation, has conflicts.
\begin{example}
A phrase $S$ has the structure $E \, ( \, s \, E)^\ast$, where a construct $E$ has either the form $b^+ \, b^{n} \, e^{n}$ or $b^{n} \, e^{n} \, e$ with $n \geq 0$. The language is defined by the $ELR \, (1)$ net below:
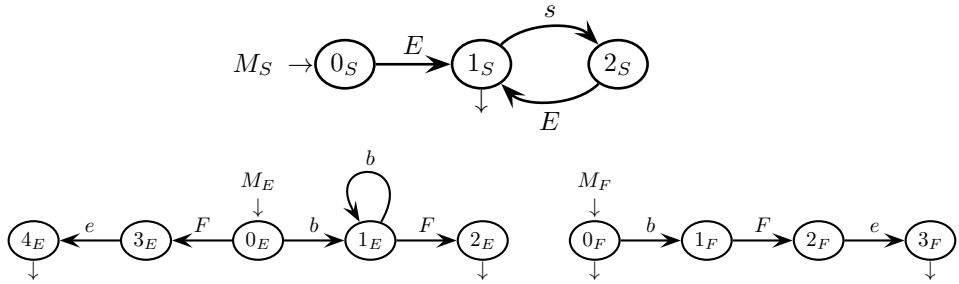
\begin{center}
\vspace{0.5cm}
\scalebox{1.0}{
\pspar\psset{border=0pt,nodesep=0pt,colsep=1cm}
\begin{psmatrix}
\ovalnode{0S}{$0_S$} & \ovalnode{1S}{$1_S$} & \ovalnode{2S}{$2_S$}

\nput[labelsep=0pt]{180}{0S}{$M_S \ \to$}

\nput[labelsep=0pt]{-90}{1S}{$\downarrow$}

\ncline{0S}{1S} \naput{$E$}

\nccurve[angleA=45,angleB=135]{1S}{2S} \naput{$s$}

\nccurve[angleA=-135,angleB=-45]{2S}{1S} \naput{$E$}

\end{psmatrix}}
\end{center}
\vspace{1.0cm}
\begin{center}
\scalebox{0.8}{
\pspar\psset{border=0pt,nodesep=0pt,colsep=1cm}
\begin{psmatrix}
\ovalnode{4E}{$4_E $} & \ovalnode{3E}{$3_E $} &  \ovalnode{0E}{$0_E $} & \ovalnode{1E}{$1_E $} & \ovalnode{2E}{$2_E$}

\nput[labelsep=0pt]{90}{0E}{$\begin{array}{c} M_E \\ \downarrow \end{array}$}

\nput[labelsep=0pt]{-90}{2E}{$\downarrow$}

\nput[labelsep=0pt]{-90}{4E}{$\downarrow$}

\ncline{0E}{1E} \naput{$b$}

\ncline{1E}{2E} \naput{$F$}

\ncline{0E}{3E} \nbput{$F$}

\nccurve[angleA=60,angleB=120,ncurvA=7,ncurvB=7]{1E}{1E} \nbput{$b$}

\ncline{3E}{4E} \nbput{$e$}

&

\ovalnode{0F}{$0_F $} & \ovalnode{1F}{$1_F $} & \ovalnode{2F}{$2_F$} & \ovalnode{3F}{$3_F $}

\nput[labelsep=0pt]{90}{0F}{$\begin{array}{c} M_F \\ \downarrow \end{array}$}

\nput[labelsep=0pt]{-90}{0F}{$\downarrow$}

\nput[labelsep=0pt]{-90}{3F}{$\downarrow$}

\ncline{0F}{1F} \naput{$b$}

\ncline{1F}{2F} \naput{$F$}

\ncline{2F}{3F} \naput{$e$}

\end{psmatrix}}
\vspace{0.25cm}
\end{center}
On the contrary, there is a conflict in the equivalent grammar:
\[
S \to E\,s\,S \; \mid \; E \qquad \; E \to B\,F \; \mid \; F\,e \qquad F \to b\,E\,f \; \mid \; \varepsilon \qquad B \to b\,B \; \mid \; b
\]
caused by the indecision whether to reduce $b^+$ to $B$ or shift. The right-linearized grammar postpones any reduction decision as long as possible and avoids conflicts.
\end{example}
\subsubsection{Parser algorithm}
Given the pilot $DFA$ of an $ELR \, (1)$ grammar or machine net, we explain how to obtain a deterministic pushdown automaton $DPDA$ that recognizes and parses the sentences.
\par
At the cost of some repetition, we recall the three sorts of abstract machines involved: the net $\mathcal{M}$ of $DFA$'s $M_S, \, M_A, \, \ldots$ with state set $Q = \{ \, q, \, \ldots \, \}$ (states are drawn as circular nodes); the pilot $DFA$ $\mathcal{P}$ with m-state set $R = \{ \, I_0, \, I_1, \, \ldots \, \}$ (states are drawn as rectangular nodes); and the $DPDA$ $\mathcal{A}$ to be next defined. As said, the $DPDA$ stores in the stack the series of m-states entered during the computation, enriched with additional information used in the parsing steps. Moreover, the m-states are interleaved with terminal or nonterminal grammar symbols. The current m-state, i.e., the one on top of stack, determines the next move: either a shift that scans the next token, or a reduction of a topmost stack segment (also called reduction handle) to a nonterminal identified by a final candidate included in the current m-state. The absence of shift-reduce conflicts makes the choice between shift and reduction operations deterministic. Similarly, the absence of reduce-reduce conflicts allows the parser to uniquely identify the final state of a machine. However, this leaves open the problem to determine the stack segment to be reduced. For that two designs will be presented: the first uses a finite pushdown alphabet; the second uses unbounded integer pointers and, strictly speaking, no longer qualifies as a pushdown automaton.
\par
First, we specify the \emph{pushdown stack alphabet}. Since for a given net $\mathcal{M}$ there are finitely many different candidates, the number of m-states is bounded and the number of candidates in any m-state is also bounded by $C_{Max} = \vert Q \vert \times \left( \, \vert \Sigma \vert + 1 \, \right)$.
The $DPDA$ stack elements are of two types: grammar symbols and stack m-states (\textit{sms}). An $sms$, denoted by $J$, contains an ordered set of triples of the form $(\textit{state}, \, \textit{look-ahead}, \, \textit{candidate identifier } (cid))$, named \emph{stack candidates}, specified as:
\[
\langle q_A, \, \pi, \, cid \rangle \ \text{where} \ q_A \in Q, \quad \pi \subseteq \Sigma \cup \{ \, \dashv \, \} \quad
\text{and} \ 1 \leq cid \leq C_{Max} \ \text{or} \ cid = \bot
\]
For readability, a $cid$ value will be prefixed by a $\sharp$ marker.
\par
The parser makes use of a surjective mapping from the set of $sms$ to the set of m-states, denoted by $\mu$, with the property that $\mu(J) = I $ if, and only if, the set of stack candidates of $J$, deprived of the candidate identifiers, equals the  set of candidates of $I$. For notational convenience, we stipulate that identically subscripted symbols $J_k$ and $I_k$ are related by $\mu(J_k) = I_k$.
As said, in the stack  the $sms$ are interleaved with grammar symbols.
\begin{algorithm}\label{AlgELR1parser}$ELR \, (1)$ parser as $DPDA$ $\mathcal{A}$.
\par
Let $J[0] \, a_1 \, J[1] \, a_2 \, \ldots \, a_k \, J[k]$ be the current stack, where $a_i$ is a grammar symbol and the top element is $J[k]$.
\begin{description}
\item[\textit{Initialization}] The analysis starts by pushing on the stack the $sms$:
\[
J_0 = \{ \, s \; \vert \quad s = \langle q, \, \pi, \, \bot \rangle \ \text{for every candidate} \ \langle q, \, \pi \rangle \in I_0 \, \}
\]
(thus $\mu(J_0) = I_0$, the initial m-state of the pilot).
\item[\textit{Shift move}] Let the top $sms$ be $J$, $I = \mu(J)$ and the current token be $a \in \Sigma$. Assume that, by inspecting $I$, the pilot has decided to shift and let $\vartheta \, (I, \, a) = I'$.	
The shift move does:
\begin{enumerate}
\item push token $a$ on stack and get next token
\item  push on stack the $sms$ $J'$ computed as follows:
\begin{eqnarray}\label{eqShiftMove1}
J' & = & \left\{ \, \langle {q_A}', \, \rho, \, \sharp i \rangle \; \vert \quad \langle {q_A}, \, \rho, \, \sharp j \rangle \ \text{is at position} \ i \ \text{in} \ J \; \land \; {q_A} \stackrel a \to {q_A}' \in \delta \, \right\} \\
\label{eqShiftMove2}
& \cup & \left\{ \, \langle {0_A}, \, \sigma, \, \bot \rangle \; \vert \quad\langle {0_A}, \, \sigma \rangle \in I'_{|closure} \, \right\}
\end{eqnarray}
(thus  $\mu(J') = I'$)
\end{enumerate}
Notice that the last condition in \eqref{eqShiftMove1} implies that ${q_A}'$ is a state of the base of $I'$.
\item[\textit{Reduction move} (non-initial state)]
The stack is $J[0] \, a_1 \, J[1] \, a_2 \, \ldots \, a_k \, J[k]$ and let the corresponding m-states be $I[i] = \mu \left( \, J[i] \, \right)$ with $0 \leq i \leq k$.
\par 	
Assume that, by inspecting  $I[k]$, the pilot chooses the reduction candidate $c = \langle q_A, \, \pi \rangle \in I[k]$, where $q_A$ is a final but non-initial state. Let $t_k = \langle q_A, \, \rho, \, \sharp i_k \rangle \in J[k]$ be the (only) stack candidate such that the current token $a \in \rho$.
\par
From  $i_k$, a $cid$ chain starts, which links $t_k$ to a stack candidate $t_{k-1} = \langle p_A, \, \rho, \, \sharp i_{k-1} \rangle \in J[k-1]$, and so on until a stack candidate $t_{h} \in J[h]$ is reached that has  $cid=\bot$ (therefore its state is initial)
\[
t_{h} = \langle 0_A, \, \rho, \, \bot \rangle
\]
The reduction move does:
\begin{enumerate}
\item grow the syntax forest by applying  reduction $a_{h+1} \, a_{h+2} \, \ldots \, a_k \leadsto A$
\item pop the stack symbols in the following order: $J[k] \, a_k \, J[k-1] \, a_{k-1} \, \ldots \, J[h+1] \, a_{h+1}$
\item execute the nonterminal shift move $\vartheta \, ( \, I[h], \, A \, )$ (see below).
\end{enumerate}
\item[\textit{Reduction move} (initial state)] It differs from the preceding case in that the chosen candidate is
$c = \langle 0_A, \, \pi, \, \bot \rangle$. The parser move grows the syntax forest by the reduction $\varepsilon \leadsto A$ and performs the nonterminal shift move corresponding to $\vartheta \, ( \, I[k], \, A \, )$.
\item[\textit{Nonterminal shift move}] It is the same as a shift move, except that the shifted symbol, $A$, is a nonterminal.
The only  difference is that the parser does not read the next input token at line (1) of \textit{Shift move}.
\item[\textit{Acceptance}] The parser accepts and halts when the stack is $J_0$, the move is the nonterminal shift defined by
$\vartheta \, ( \, I_0, \, S \, )$ and the current token is $\dashv$.
\end{description}
\end{algorithm}
For shift moves, we note that the m-states computed by Alg. \ref{algMacchPilotaELR(1)}, may contain multiple stack candidates that have the same state; this happens whenever edge $I \to \theta \, (I, \, a)$ is convergent. It may help to look at the situation of a shift move (eq. \eqref{eqShiftMove1} and \eqref{eqShiftMove2}) schematized in Figure \ref{figSchematizedShift}.
\begin{figure}
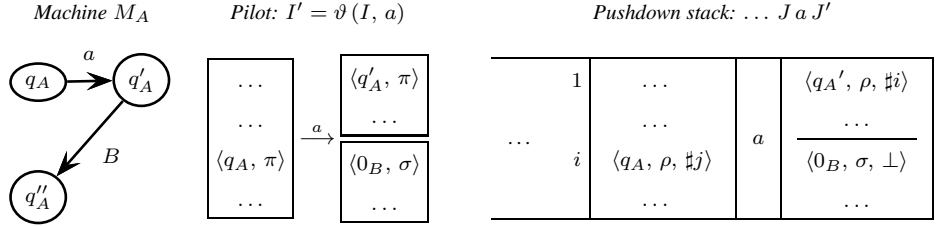

\begin{center}
\begin{tabular}{cccc}
\emph{Machine} $M_A$ & \emph{Pilot:} $I' = \vartheta \, ( I, \, a )$ && \emph{Pushdown stack:} $\ldots \, J \, a \, J'$ \\ \\
\begin{minipage}{0.175\textwidth}
\pspar\psset{colsep=0.6cm,rowsep=0.8cm,nodesep=0pt}
\begin{psmatrix}

\ovalnode{qA}{$q_A$} & \ovalnode{qA'}{${q'_A}$} \\

\ovalnode{qA''}{${q''_A}$}

\ncline{qA}{qA'} \naput[labelsep=7pt]{$a$}

\ncline{qA'}{qA''} \naput[labelsep=5pt]{$B$}

\end{psmatrix}
\end{minipage}

&

$\def\arraystretch{1.5}\begin{array}[c]{|c|} \hline
\ldots \\
\ldots \\
\langle q_A, \, \pi \rangle \\
\ldots \\ \hline
\end{array}$
$\stackrel a \longrightarrow$
$\def\arraystretch{1.5}\begin{array}[c]{|c|} \hline
\langle {q'_A}, \, \pi \rangle \\
\ldots \\ \hline\hline
{ \langle 0_B, \,  \sigma \rangle } \\
\ldots \\ \hline
\end{array}$

&&

\begin{tabular}{cc|c|c|c} \hline
\ldots & $\def\arraystretch{1.5}\begin{array}{c|c}
1 & \ldots \\
  & \ldots \\
i & \ \ \, \langle q_A, \, \rho, \, \sharp j \rangle \\
  & \ldots \\
\end{array}$ & $a$ & $\def\arraystretch{1.5}\begin{array}{c}
\langle {q_A}', \, \rho, \, \sharp i \rangle \\
\ldots \\ \hline
\langle 0_B, \, \sigma, \, \bot \rangle \\
\ldots \\
\end{array}$ \\ \hline
\end{tabular}
\end{tabular}
\end{center}
\caption{Schematization of the shift move $q_A \stackrel a \to q'_A$ (plus completion) in the parser stack with pointers.}\label{figSchematizedShift}
\end{figure}
\begin{example}\label{exRunningParsingTrace}Parsing trace.
\par
The step by step execution of the parser on input string $( \, ( \, ) \, a \, )$ produces the trace shown in Figure \ref{figExRunningParseTrace}. For clarity there are two parallel tracks: the input string, progressively replaced by the nonterminal symbols shifted on the stack; and the stack of stack m-states.
\begin{figure}[h!]
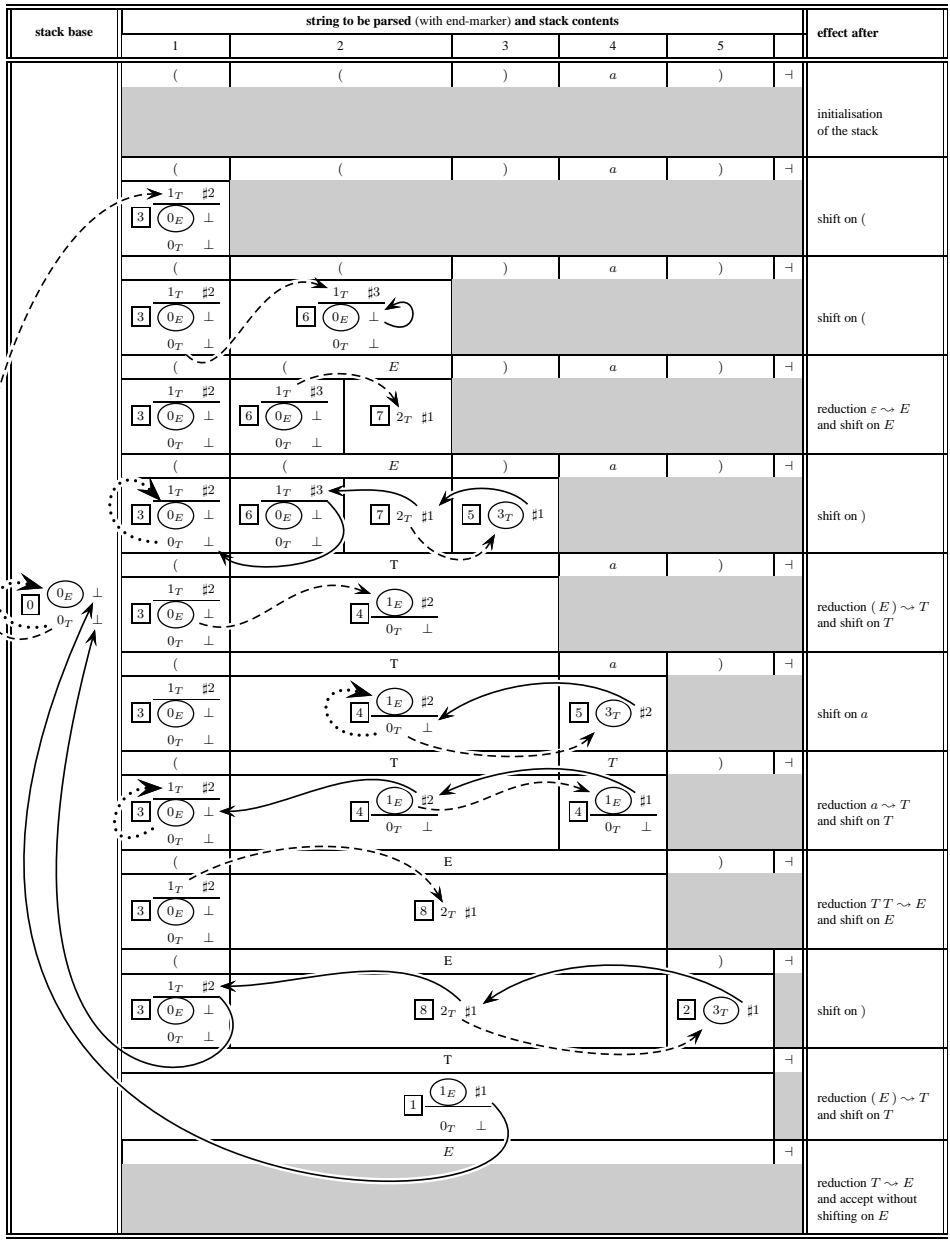

\scalebox{0.625}{\def\arraystretch{1.4}
\begin{tabular}{||c||c|c|c|c|c|c||m{70pt}||} \hline\hline
\multirow{2}{*}{\textbf{stack base}} & \multicolumn{6}{|c||}{\textbf{string to be parsed {\rm (with end-marker)} and stack contents}} & \multirow{2}{*}{\textbf{effect after}} \\ \cline{2-7}
& 1 & 2 & 3 & 4 &	5 & &  \\ \hline\hline

\multirow{47}{*}{$\framebox{$0$} \ \begin{array}{cc} \ovalnode{46}{0_E} & \rnode{52}{\bot} \\ \rnode{45}{0_T} & \rnode{44}{\bot} \end{array}$} & $($ & $($ & $)$ & $a$ & $)$ & $\dashv$ & \\ \cline{2-6}
& \multicolumn{6}{|c||}{\cellcolor[gray]{0.8} $\begin{array}{c} \\ \\ \\ \end{array}$} & initialisation \par of the stack \\ \cline{2-8}

& $($ & $($ & $)$ & $a$ & $)$ & $\dashv$ & \\ \cline{2-6}
& $\framebox{$3$} \ \begin{array}{cc} \rnode{s2}{1_T} & \sharp 2 \\ \hline \ovalnode{\relax}{0_E} & \bot \\ 0_T & \bot \end{array}$
& \multicolumn{5}{|c||}{\cellcolor[gray]{0.8}} & shift on $($ \\ \cline{2-8}

& $($ & $($ & $)$ & $a$ & $)$ & $\dashv$ & \\ \cline{2-6}
& $\framebox{$3$} \ \begin{array}{cc} 1_T & \sharp 2 \\ \hline \ovalnode{\relax}{0_E} & \bot \\ \rnode{s3}{0_T} & \bot \end{array}$
& $\framebox{$6$} \ \begin{array}{cc} \rnode{s4}{1_T} & \sharp 3 \\ \hline \ovalnode{\relax}{0_E} & \rnode{01}{\bot} \\ 0_T & \bot \end{array}$
& \multicolumn{4}{|c||}{\cellcolor[gray]{0.8}} & shift on $($ \\ \cline{2-8}

\pspar

\nccurve[angleA=-30,angleB=30,ncurvA=7,ncurvB=7]{01}{01}

\nccurve[linestyle=dashed,angleA=-45,angleB=160,ncurvA=0.75,ncurvB=1]{s3}{s4}

& $($ & $($ \hspace{2cm} $E$ & $)$ & $a$ & $)$ & $\dashv$ & \\ \cline{2-6}
& $\framebox{$3$} \ \begin{array}{cc} 1_T & \sharp 2 \\ \hline \ovalnode{\relax}{0_E} & \bot \\ 0_T & \bot \end{array}$
& $\framebox{$6$} \ \begin{array}{cc} \rnode{s5}{1_T} & \sharp 3 \\ \hline \ovalnode{\relax}{0_E} & \bot \\ 0_T & \bot \end{array}$
\quad \vline \qquad $\framebox{$7$} \ \begin{array}{cc} \rnode{s6}{2_T} & \sharp 1 \end{array}$ \quad
& \multicolumn{4}{|c||}{\cellcolor[gray]{0.8}} & reduction $\varepsilon \leadsto E$ \par and shift on $E$ \\ \cline{2-8}

\pspar

\nccurve[linestyle=dashed,angleA=30,angleB=120,ncurvA=0.75,ncurvB=0.75]{s5}{s6}

& $($ & $($ \hspace{2cm} $E$ & $)$ & $a$ & $)$ & $\dashv$ & \\ \cline{2-6}
& $\framebox{$3$} \ \begin{array}{cc} 1_T & \sharp 2 \\ \hline \ovalnode{16}{0_E} & \bot \\ \rnode{15}{0_T} & \rnode{14}{\bot} \end{array}$
& $\framebox{$6$} \ \begin{array}{cc} 1_T & \rnode{13}{\sharp 3} \\ \hline \ovalnode{\relax}{0_E} & \bot \\ 0_T & \bot \end{array}$
\quad \vline \qquad $\framebox{$7$} \ \begin{array}{cc} \rnode{s7}{2_T} & \rnode{12}{\sharp 1} \end{array}$ \quad
& $\framebox{$5$} \ \begin{array}{cc} \rnode{s8}{\ovalnode{\relax}{3_T}} & \rnode{11}{\sharp 1} \end{array}$
& \multicolumn{3}{|c||}{\cellcolor[gray]{0.8}} & shift on $)$ \\ \cline{2-8}

\pspar

\nccurve[angleA=135,angleB=45]{11}{12}

\nccurve[angleA=135,angleB=0]{12}{13}

\nccurve[angleA=-45,angleB=-45,ncurvA=1.5]{13}{14}

\nccurve[linestyle=dotted,linewidth=2pt,angleA=-180,angleB=135,ncurvA=4,ncurvB=3.5]{15}{16}

\nccurve[linestyle=dashed,angleA=-60,angleB=-120,ncurvA=1,ncurvB=1]{s7}{s8}

& $($ & \multicolumn{2}{c|}{T} & $a$ & $)$ & $\dashv$ & \\ \cline{2-6}
& $\framebox{$3$} \ \begin{array}{cc} 1_T & \sharp 2 \\ \hline \rnode{s9}{\ovalnode{\relax}{0_E}} & \bot \\ 0_T & \bot \end{array}$
& \multicolumn{2}{c|}{$\framebox{$4$} \ \begin{array}{cc} \rnode{s10}{\ovalnode{\relax}{1_E}} & \sharp 2 \\ \hline 0_T & \bot \end{array}$}
& \multicolumn{3}{|c||}{\cellcolor[gray]{0.8}} & reduction $( \, E \, ) \leadsto T$ \par and shift on $T$ \\ \cline{2-8}

\pspar

\nccurve[linestyle=dashed,angleA=-20,angleB=160,ncurvA=0.75,ncurvB=1]{s9}{s10}

& $($ & \multicolumn{2}{c|}{T} & $a$ & $)$ & $\dashv$ & \\ \cline{2-6}
& $\framebox{$3$} \ \begin{array}{cc} 1_T & \sharp 2 \\ \hline \ovalnode{\relax}{0_E} & \bot \\ 0_T & \bot \end{array}$
& \multicolumn{2}{c|}{$\framebox{$4$} \ \begin{array}{cc} \ovalnode{24}{1_E} & \sharp 2 \\ \hline \rnode{23}{0_T} & \rnode{22}{\bot} \end{array}$}
& $\framebox{$5$} \ \begin{array}{cc} \rnode{s12}{\ovalnode{\relax}{3_T}} & \rnode{21}{\sharp 2} \end{array}$
& \multicolumn{2}{|c||}{\cellcolor[gray]{0.8}} & shift on $a$ \\ \cline{2-8}

\pspar

\nccurve[angleA=145,angleB=35]{21}{22}

\nccurve[linestyle=dotted,linewidth=2pt,angleA=-165,angleB=165,ncurvA=3.5,ncurvB=3.5]{23}{24}

\nccurve[linestyle=dashed,angleA=-30,angleB=-135,ncurvA=0.5,ncurvB=0.5]{23}{s12}

& $($ & \multicolumn{2}{c|}{T} & $T$ & $)$ & $\dashv$ & \\ \cline{2-6}
& $\framebox{$3$} \ \begin{array}{cc} \rnode{35}{1_T} & \sharp 2 \\ \hline \ovalnode{34}{0_E} & \rnode{33}{\bot} \\ 0_T & \bot \end{array}$
& \multicolumn{2}{c|}{$\framebox{$4$} \ \begin{array}{cc} \rnode{s13}{\ovalnode{\relax}{1_E}} & \rnode{32}{\sharp 2} \\ \hline 0_T & \bot \end{array}$}
& $\framebox{$4$} \ \begin{array}{cc} \rnode{s14}{\ovalnode{\relax}{1_E}} & \rnode{31}{\sharp 1} \\ \hline 0_T & \bot \end{array}$
& \multicolumn{2}{|c||}{\cellcolor[gray]{0.8}} & reduction $a \leadsto T$ \par and shift on $T$ \\ \cline{2-8}

\pspar

\nccurve[angleA=145,angleB=25]{31}{32}

\nccurve[angleA=145,angleB=0]{32}{33}

\nccurve[linestyle=dotted,linewidth=2pt,angleA=-145,angleB=180,ncurvA=4,ncurvB=3]{34}{35}

\nccurve[linestyle=dashed,angleA=-15,angleB=160,ncurvA=0.75,ncurvB=1]{s13}{s14}

& $($ & \multicolumn{3}{c|}{E} & $)$ & $\dashv$ & \\ \cline{2-6}
& $\framebox{$3$} \ \begin{array}{cc} \rnode{s16}{1_T} & \sharp 2 \\ \hline \ovalnode{\relax}{0_E} & \bot \\ 0_T & \bot \end{array}$
& \multicolumn{3}{c|}{$\framebox{$8$} \ \begin{array}{cc} \rnode{s17}{2_T} & \sharp 1 \end{array}$}
& \multicolumn{2}{|c||}{\cellcolor[gray]{0.8}} & reduction $T \, T \leadsto E$ \par and shift on $E$ \\ \cline{2-8}

\pspar

\nccurve[linestyle=dashed,angleA=30,angleB=120]{s16}{s17}

& $($ & \multicolumn{3}{c|}{E} & $)$ & $\dashv$ & \\ \cline{2-6}
& $\framebox{$3$} \ \begin{array}{cc} 1_T & \rnode{43}{\sharp 2} \\ \hline \ovalnode{\relax}{0_E} & \bot \\ 0_T & \bot \end{array}$
& \multicolumn{3}{c|}{$\framebox{$8$} \ \begin{array}{cc} \rnode{s18}{2_T} & \rnode{42}{\sharp 1} \end{array}$}
& $\framebox{$2$} \ \begin{array}{cc} \rnode{s19}{\ovalnode{\relax}{3_T}} & \rnode{41}{\sharp 1} \end{array}$
& \multicolumn{1}{|c||}{\cellcolor[gray]{0.8}} & shift on $)$ \\ \cline{2-8}

\pspar

\nccurve[angleA=145,angleB=35]{41}{42}

\nccurve[angleA=135,angleB=0]{42}{43}

\nccurve[angleA=-45,angleB=-105,ncurvA=0.625,ncurvB=3]{43}{44}

\nccurve[linestyle=dotted,linewidth=2pt,angleA=-165,angleB=165,ncurvA=3.5,ncurvB=3.5]{45}{46}

\nccurve[linestyle=dashed,angleA=-30,angleB=-135,ncurvA=0.5,ncurvB=0.5]{s18}{s19}

\nccurve[linestyle=dashed,angleA=-150,angleB=180,ncurvA=0.80,ncurvB=0.75]{45}{s2}

& \multicolumn{5}{c|}{T} & $\dashv$ & \\ \cline{2-6}
& \multicolumn{5}{c|}{$\framebox{$1$} \ \def\arraystretch{2}\begin{array}{cc} \ovalnode{\relax}{1_E} & \rnode{51}{\sharp 1} \\ \hline 0_T & \bot \end{array}$}
& \multicolumn{1}{|c||}{\cellcolor[gray]{0.8}} & reduction $( \, E \, ) \leadsto T$ \par and shift on $T$ \\ \cline{2-8}

& \multicolumn{5}{|c|}{$E$} & $\dashv$ & \\ \cline{2-6}
& \multicolumn{6}{|c||}{\cellcolor[gray]{0.8} $\begin{array}{c} \\ \\ \\ \end{array}$} & reduction $T \leadsto E$ \par and accept without \par shifting on $E$
\pspar
\nccurve[angleA=-45,angleB=-115,ncurvA=0.625,ncurvB=2.25]{51}{52}
\\ \hline\hline
\end{tabular}}
\caption{Parsing steps for string $( \, ( \, ) \, a \, )$ (grammar and $ELR \, (1)$ pilot in Figures \ref{figRunningExELRpart1} and \ref{figRunningExELRpilot}); the name $J_h$ of a $sms$ maps onto the corresponding m-state $\mu(J_h) = I_h$.} \label{figExRunningParseTrace}
\end{figure}
The stack has one more entry than the scanned prefix; the suffix yet to be scanned is to the right of the stack. Inside a stack element $J[k]$, each 3-tuple is identified by its ordinal position, starting from 1 for the first 3-tuple; a value $\sharp i$ in a \emph{cid} field of  element $J[k + 1]$ encodes a pointer to the $i$-th $3$-tuple of $J[k]$.
\par
To ease reading the parser trace simulation, the m-state number appears framed in each stack element, e.g., $I_0$ is denoted \framebox{$0$}, etc.; the final candidates are encircled, e.g., \ovalnode{\relax}{$1_E$}, etc.; and to avoid clogging, look-ahead sets are not shown as they are only needed for convergent transitions, which do not occur here. But the look-aheads can always be found by inspecting the pilot graph.
\par
Fig. \ref{figExRunningParseTrace} highlights the shift moves as dashed forward arrows that link two topmost stack candidates. For instance, the first (from the Figure top) terminal shift, on $($, is $\langle 0_T, \, \bot \rangle \stackrel ( \to \langle 1_T, \, \sharp 2 \rangle$. The first nonterminal shift is the third one, on $E$, i.e., $\langle 1_T, \, \sharp 3 \rangle \stackrel E \to \langle 2_T, \, \sharp 1 \rangle$, after the null reduction $\varepsilon \leadsto E$. Similarly for all the other shifts.
\par
Fig. \ref{figExRunningParseTrace}  highlights the reduction handles, by means of solid backward arrows that link the candidates involved. For instance, see reduction $( \, E \, ) \leadsto T$: the stack configuration above shows the chain of three pointers $\sharp 1$, $\sharp 1$ and $\sharp 3$ - these three stack elements form the  handle and are popped - and finally the initial pointer $\bot$ - this stack element is the reduction origin and is not popped. The initial pointer $\bot$ marks the initial candidate $\langle 0_T, \, \bot \rangle$, which is obtained by means of a closure operation applied to candidate $\langle 0_E, \, \bot \rangle$ in the same stack element, see the dotted arrow that links them, from which the subsequent shift on $T$ starts, see the dashed arrow in the stack configuration below. The solid self-loop on candidate $\langle 0_E, \, \bot \rangle$ (at the $2^{nd}$ shift of token $($ - see the effect on the line below) highlights the null reduction $\varepsilon \leadsto E$, which does not pop anything from the stack and is immediately followed by a shift on $E$.
\par
We observe that the parser must store on stack the scanned grammar symbols, because in a reduction move, at step (2), they may be necessary for selecting the correct reduction handle and build the subtree to be added to the syntax forest.
\par
Returning to the example, the execution order of reductions is:
\[
\varepsilon\leadsto E \qquad ( \, E \, ) \leadsto T \qquad a \leadsto T \qquad T \, T \leadsto E \qquad ( \, E \, )\leadsto T \qquad T \leadsto E
\]
Pasting together the reductions, we obtain this syntax tree, where the  order of reductions is displayed:
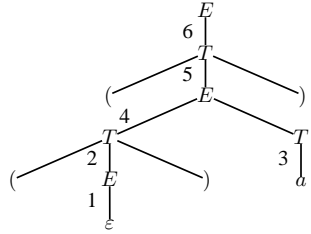
\begin{center}
 \scalebox{0.8}{%
\psset{arrows=-,linestyle=solid,levelsep=20pt, treesep=40pt}
 \pstree{\TR{$E$}}
   {
    \pstree{\TR{$T$}^6}			
				{
				\TR{$($}
      \pstree{\TR{$E$}^5}
        {

        \pstree{\TR{$T$}^4}
            {
                \TR{$($}

            \pstree{\TR{$E$}^2}
                {  \TR{$\varepsilon$}^1 }

                 \TR{$)$}
            }

        \pstree{\TR{$T$}} %
            {

                 \TR{$a$}^3
            }

        }
     \TR{$)$}
     }	

   }

  }
\end{center}
Clearly, the reduction order matches a rightmost derivation but in reversed order.
\end{example}
\par
It is worth examining more closely the case  of convergent conflict-free edges. Returning to Figure \ref{figSchematizedShift}, notice that  the stack candidates linked by a $cid$ chain are mapped onto m-state candidates that have the same machine state (here the stack candidate $\langle {q'_A}, \, \rho, \, \sharp i \rangle$ is linked via  $\sharp i$ to $\langle {q_A}, \, \rho, \, \sharp j \rangle$): the look-ahead set $\pi$ of such m-state candidates is in general a superset of the set $\rho$ included in the stack candidate, due to the possible presence of convergent transitions; the two sets $\rho$ and $\pi$  coincide when no convergent transition is taken on the pilot automaton at parsing time.
\par
An $ELR(1)$ grammar with convergent edges is studied next.
\begin{example}\label{exELRconvergentGrammar}
The net, pilot graph and the trace of a parse are shown in Figure \ref{figExELRconvergentGrammarPilotTrace}, where, for readability, the \emph{cid} values in the stack candidates are visualized as backward pointing arrows. Stack m-state $J_3$ contains two candidates which just differ in their look-ahead sets. In the corresponding pilot m-state $I_3$, the two candidates are the targets of a convergent non-conflictual m-state transition (highlighted double-line in the pilot graph).
\begin{figure}[h!]
\begin{flushleft}
\emph{Machine Net}
\end{flushleft}
\begin{center}
\vspace{0.25cm}
\scalebox{1.0}{
\pspar\psset{border=0pt,nodesep=0pt,colsep=0.725cm,rowsep=0.5cm}
\begin{psmatrix}

& & \ovalnode{2S}{$2_S$} \\

\ovalnode{0S}{$0_S$} & \ovalnode{1S}{$1_S$} & & \ovalnode{4S}{$4_S$} & & & \ovalnode{0A}{$0_A$} & \ovalnode{1A}{$1_A$} & \ovalnode{2A}{$2_A$} \\

& & \ovalnode{3S}{$3_S$}

\nput[labelsep=0pt]{180}{0A}{$M_A \ \to$}

\nput[labelsep=0pt]{0}{2A}{$\to$}

\ncline{0A}{1A} \nbput{$b$}

\ncline{1A}{2A} \nbput{$e$}

\nccurve[angleA=60,angleB=120,ncurvA=0.75,ncurvB=0.75]{0A}{2A} \naput{$e$}

\nput[labelsep=0pt]{180}{0S}{$M_S \ \to$}

\nput[labelsep=0pt]{0}{4S}{$\to$}

\ncline{0S}{1S} \naput{$a$}

\ncline{1S}{2S} \naput{$A$}

\ncline{1S}{3S} \nbput{$b$}

\ncline{2S}{4S} \naput{$d$}

\ncline{3S}{4S} \nbput{$A$}

\end{psmatrix}}
\vspace{0.25cm}
\end{center}
\hrule
\begin{flushleft}
\emph{Pilot Graph}
\end{flushleft}
\begin{center}
\vspace{0.25cm}
\scalebox{1.2}{
\pspar\psset{border=0pt,arrows=->,nodesep=0pt,colsep=1.5cm,rowsep=1cm,labelsep=5pt}
\begin{psmatrix}

\rnode{I0}{$\def\arraystretch{1.5}\begin{array}{|c|c|} \hline\hline
0_S & \dashv \\ \hline
\end{array}$} &

\rnode{I1}{$\def\arraystretch{1.5}\begin{array}{|c|c|} \hline
1_S & \dashv \\ \hline\hline
0_A & d \\ \hline
\end{array}$} &

\rnode{I2}{$\def\arraystretch{1.5}\begin{array}{|c|c|} \hline
3_S & \dashv \\ \cline{2-2}
1_A & d \\ \hline\hline
0_A & \dashv \\ \hline
\end{array}$} &

\rnode{I3}{$\def\arraystretch{1.5}\begin{array}{|c|c|} \hline
\ovalnode{\relax}{\mathbf{2_A}} & d \; \dashv \\ \hline\hline
\end{array}$} \\

&

\rnode{I4}{$\def\arraystretch{1.5}\begin{array}{|c|c|} \hline
2_S& \dashv \\ \hline\hline
\end{array}$} &

\rnode{I5}{$\def\arraystretch{1.5}\begin{array}{|c|c|} \hline
\ovalnode{\relax}{\mathbf{4_S}} & \dashv \\ \hline\hline
\end{array}$} \\

\nput[labelsep=5pt]{90}{I0}{$I_0$}

\nput[labelsep=5pt]{90}{I1}{$I_1$}

\nput[labelsep=5pt]{90}{I2}{$I_2$}

\nput[labelsep=5pt]{90}{I3}{$I_3$}

\nput[labelsep=5pt]{-90}{I4}{$I_4$}

\nput[labelsep=5pt]{-90}{I5}{$I_5$}

\ncline{I0}{I1} \naput{$a$}

\ncline{I1}{I2} \naput{$b$}

\ncline[doubleline=true,arrowscale=1.5]{I2}{I3} \naput{$e$} \nbput{\parbox{1.5cm}{\centering convergent \par edge}}

\ncline{I2}{I5} \naput{$A$}

\ncline{I1}{I4} \nbput{$A$}

\ncline{I4}{I5} \nbput{$d$}

\end{psmatrix}}
\end{center}
\hrule
\begin{flushleft}
\emph{Parse Traces}
\end{flushleft}
\begin{center}
\vspace{0.25cm}
\scalebox{1.2}{
\pspar\psset{border=0.0cm,nodesep=3pt,labelsep=5pt,rowsep=0.25cm,colsep=0.625cm}
\begin{psmatrix}
$J_0$ & $a$ & $J_1$ & $b$ & $J_2$ & $e$ & $J_3$ & $d$ & $\dashv$ \\

\rnode{I0}{$\def\arraystretch{1.5}\begin{array}{|c|c|} \hline\hline
0_S & \rnode{0s1p}{\dashv} \\ \hline
\end{array}$}

&&

\rnode{I1}{$\def\arraystretch{1.5}\begin{array}{|c|c|} \hline
\rnode{1s1s}{1_S} & \rnode{1s1p}{\dashv} \\ \hline\hline
\rnode{1s2s}{0_A} & \rnode{1s2p}{d} \\ \hline
\end{array}$}

&&

\rnode{I2}{$\def\arraystretch{1.5}\begin{array}{|c|c|} \hline
\rnode{2s1s}{3_S} & \dashv \\ \cline{2-2}
\rnode{2s2s}{1_A} & \rnode{2s2p}{d} \\ \hline\hline
0_A & \rnode{2s3p}{\dashv} \\ \hline
\end{array}$}

&&

\rnode{I3}{$\def\arraystretch{1.5}\begin{array}{|c|c|} \hline
\rnode{3s1s}{\ovalnode{\relax}{\mathbf{2_A}}} & d \\ \cline{2-2}
\rnode{3s2s}{\ovalnode{\relax}{\mathbf{2_A}}} & \dashv \\ \hline\hline
\end{array}$} \\

\nccurve[angleA=-165,angleB=0,ncurvA=0.75,ncurvB=0.75]{1s1s}{0s1p}

\nccurve[angleA=-165,angleB=0,ncurvA=0.75,ncurvB=0.75]{2s1s}{1s1p}

\nccurve[angleA=-165,angleB=0,ncurvA=0.75,ncurvB=0.75]{2s2s}{1s2p}

\nccurve[angleA=-165,angleB=0,ncurvA=0.75,ncurvB=0.75]{3s1s}{2s2p}

\nccurve[angleA=-165,angleB=0,ncurvA=0.75,ncurvB=0.75]{3s2s}{2s3p}

\\

$J_0$ & $a$ & $J_1$ & $A$ & $J_4$ & $d$ & $J_5$ & $\dashv$ \\

\rnode{I0}{$\def\arraystretch{1.5}\begin{array}{|c|c|} \hline\hline
0_S & \rnode{0s1p}{\dashv} \\ \hline
\end{array}$}

&&

\rnode{I1}{$\def\arraystretch{1.5}\begin{array}{|c|c|} \hline
\rnode{1s1s}{1_S} & \rnode{1s1p}{\dashv} \\ \hline\hline
0_A & d \\ \hline
\end{array}$}

&&

\rnode{I2}{$\def\arraystretch{1.5}\begin{array}{|c|c|} \hline
\rnode{2s1s}{2_S} & \rnode{2s1p}{\dashv} \\ \hline\hline
\end{array}$}

&&

\rnode{I3}{$\def\arraystretch{1.5}\begin{array}{|c|c|} \hline
\rnode{3s1s}{\ovalnode{\relax}{\mathbf{4_S}}} & \dashv \\ \hline\hline
\end{array}$}

\nccurve[angleA=-165,angleB=0,ncurvA=0.75,ncurvB=0.75]{1s1s}{0s1p}

\nccurve[angleA=165,angleB=0,ncurvA=0.75,ncurvB=0.75]{2s1s}{1s1p}

\ncline{3s1s}{2s1p}

\end{psmatrix}}
\begin{center}
\vspace{0.25cm}
reductions $b\,e \leadsto A$ (above) and $a\,A\,d \leadsto S$ (below)
\end{center}
\end{center}
\caption{$ELR \, (1)$ net, pilot with convergent edges (double line) and parsing trace of string  $a\,b\,e\,d\,\dashv$.}\label{figExELRconvergentGrammarPilotTrace}
\end{figure}
\end{example}
\subsection{Simplified parsing for \emph{BNF} grammars}
For $LR \, (1)$ grammars that do not use regular expressions in the rules, some features of the $ELR \, (1)$ parsing algorithm become superfluous. We briefly discuss them to highlight the differences between extended and basic shift-reduce parsers. Since the graph of every machine is a tree, there are no edges entering the same machine state, which rules out the presence of convergent edges in the pilot. Moreover, the alternatives of a nonterminal, $A \to \alpha \; \mid \; \beta \; \mid \; \ldots$, are recognized in distinct final states $q_{\alpha, \, A}$, $q_{\beta, \, A}$, $\ldots$ of machine $M_A$. Therefore, if the candidate chosen by the parser is $q_{\alpha, \, A}$, the $DPDA$ simply pops $2 \times \vert \alpha \vert$ stack elements and performs the reduction $\alpha \leadsto A$. Since pointers to a preceding stack candidate are no longer needed, the stack m-states coincide with the pilot ones.
\par
Second, for a related reason the interleaved grammar symbols are no longer needed on the stack, because the pilot of an $LR \, (1)$ grammar has the well-known property that all the edges entering an m-state carry the same terminal/nonterminal label. Therefore the reduction handle is uniquely determined by the final candidate of the current m-state.
\par
The simplifications have some effect on the formal notation: for $BNF$ grammars the use of machine nets and their states becomes subjectively less attractive than the classical notation based on marked grammar rules.
\subsection{Parser implementation using an indexable stack} \label{VectorStackSection}
Before finishing with bottom-up parsing, we present an alternative implementation of the parser for $EBNF$ grammars (Algorithm \ref{AlgELR1parser}, p. \pageref{AlgELR1parser}), where the memory of the analyzer is an array of elements such that any element can be directly accessed by means of an integer index, to be named a \emph{vector-stack}. The reason for presenting the new implementation is twofold: this technique is compatible with the implementation of non-deterministic tabular parsers (Sect. \ref{sectionTabular}) and is potentially faster. On the other hand, a vector-stack as data-type is more general than a pushdown stack, therefore this parser cannot be viewed as a pure $DPDA$.
\par
As before, the elements in the vector-stack are of two alternating types: \emph{vector-stack m-states} $vsms$ and grammar symbols. A $vsms$, denoted by $J$, is a set of  triples, named \emph{vector-stack candidates}, the form of which $\langle q_A, \, \pi, \, elemid \rangle$ simply differs from the earlier stack candidates because the  third component is a positive integer named \emph{element identifier} instead of a $cid$. Notice also that now the set is not ordered. The surjective mapping from vector-stack m-states to pilot m-states is denoted $\mu$, as before. Each $elemid$ points back to the vector-stack element containing the initial state of the current machine, so that, when a reduction move is performed, the length of the string to be reduced - and the reduction handle - can be obtained directly without inspecting the stack elements below the top one. Clearly the value of $elemid$ ranges from 1 to the maximum vector-stack height.
\begin{algorithm}\label{AlgELR1parserVectorStack}$ELR(1)$ parser automaton $\mathcal{A}$ using a vector-stack.
\par
Let $J[0] \, a_1 \, J[1] \, a_2 \, \ldots \, a_k \, J[k]$ be the current stack, where $a_i$ is a grammar symbol and the top element is $J[k]$.
\begin{description}
\item[\textit{Initialization}] The analysis starts by pushing on the stack the $sms$:
\[
J_0 = \{ \, s \; \vert \quad s = \langle q, \, \pi, \, 0 \rangle \ \text{for every candidate} \ \langle q, \, \pi \rangle \in I_0 \, \}
\]
(thus $\mu(J_0) = I_0$, the initial m-state of the pilot).
\item[\textit{Shift move}] Let the top $vsms$ be $J$, $I = \mu(J)$ and the current token be $a \in \Sigma$.
Assume that, by inspecting $I$, the pilot has decided to shift and let $\vartheta \, (I, \, a) = I'$.
\par
The shift move does:
\begin{enumerate}
\item push token $a$ on stack and get next token
\item  push on stack the $sms$ $J'$ (more precisely, $k := k+1$ and $J[k] := J'$) computed as follows:
\begin{eqnarray}\label{VSeqShiftMove1}
J' & = & \left\{ \, \langle {q_A}', \, \rho, \, j \rangle \; \vert \quad \langle {q_A}, \, \rho, \, j \rangle \ \text{is in} \ J \; \land \; {q_A} \stackrel a \to {q_A}' \in \delta \, \right\} \\
\label{VSeqShiftMove2}
& \cup & \left\{ \, \langle {0_A}, \, \sigma, \, k \rangle \; \vert \quad\langle {0_A}, \, \sigma \rangle \in I'_{|closure} \, \right\}
\end{eqnarray}
(thus  $\mu(J') = I'$)
\end{enumerate}
\item[\textit{Reduction move} (non-initial state)]
The stack is $J[0] \, a_1 \, J[1] \, a_2 \, \ldots \, a_k \, J[k]$ and let the corresponding m-states be $I[i] = \mu \left( \, J[i] \, \right)$ with $0 \leq i \leq k$.
\par
Assume that, by inspecting  $I[k]$, the pilot chooses the reduction candidate $c = \langle q_A, \, \pi \rangle \in I[k]$, where $q_A$ is a final but non-initial state. Let $t_k = \langle q_A, \, \rho, \, h \rangle \in J[k]$ be the (only) stack candidate such that the current token $a \in \rho$.
\par
The reduction move does:
\begin{enumerate}
\item grow the syntax forest by applying  reduction $a_{h+1} \, a_{h+2} \, \ldots \, a_k \leadsto A$
\item pop the stack symbols in the following order: $J[k] \, a_k \, J[k-1] \, a_{k-1} \, \ldots \, J[h+1] \, a_{h+1}$
\item execute the nonterminal shift move $\vartheta \, ( \, I[h], \, A \, )$ (see below).
\end{enumerate}
\item[\textit{Reduction move} (initial state)] It differs from the preceding case in that the chosen candidate is
$c = \langle 0_A, \, \pi, \, k \rangle$. The parser move grows the syntax forest by the reduction $\varepsilon \leadsto A$ and performs the nonterminal shift move corresponding to $\vartheta \, ( \, I[k], \, A \, )$.
\item[\textit{Nonterminal shift move}] It is the same as a shift move, except that the shifted symbol, $A$, is a nonterminal.
The only  difference is that the parser does not read the next input token at line (1) of \textit{Shift move}.
\item[\textit{Acceptance}] The parser accepts and halts when the stack is $J_0$, the move is the nonterminal shift defined by
$\vartheta \, ( \, I_0, \, S \, )$ and the current token is $\dashv$.
\end{description}
\end{algorithm}
Although the algorithm uses a stack, it cannot be viewed as a $PDA$ because the stack alphabet is unbounded since a $vsms$ contains integer values.
\begin{example}\label{exRunningParsingTraceVector}Parsing trace.
\par
Figure \ref{figExRunningParseTraceVector}, to be compared with Figure \ref{figExRunningParseTrace},
shows the same step by step execution of the parser as in Example \ref{exRunningParsingTrace},
on input string  $( \, (\, ) \, a \, )$; the graphical conventions are the same.
\begin{figure}[htbp!]
\scalebox{0.650}{\def\arraystretch{1.4}
\begin{tabular}{||c||c|c|c|c|c|c||m{70pt}||} \hline\hline
\textbf{stack base} & \multicolumn{6}{|c||}{\textbf{string to be parsed {\rm (with end-marker)} and stack contents {\rm (with indices)}}} & \multirow{2}{*}{\textbf{effect after}} \\ \cline{1-7}
0 & 1 & 2 & 3 & 4 &	5 & &  \\ \hline\hline

\multirow{47}{*}{$\framebox{$0$} \ \begin{array}{cc} \ovalnode{46}{0_E} & \rnode{52}{0} \\ \rnode{45}{0_T} & \rnode{44}{0} \end{array}$} & $($ & $($ & $)$ & $a$ & $)$ & $\dashv$ & \\ \cline{2-6}
& \multicolumn{6}{|c||}{\cellcolor[gray]{0.8} $\begin{array}{c} \\ \\ \\ \end{array}$} & initialisation \par of the stack \\ \cline{2-8}

& $($ & $($ & $)$ & $a$ & $)$ & $\dashv$ & \\ \cline{2-6}
& $\framebox{$3$} \ \begin{array}{cc} 1_T & 0 \\ \hline \ovalnode{\relax}{0_E} & 1 \\ 0_T & 1 \end{array}$
& \multicolumn{5}{|c||}{\cellcolor[gray]{0.8}} & shift on $($ \\ \cline{2-8}

& $($ & $($ & $)$ & $a$ & $)$ & $\dashv$ & \\ \cline{2-6}
& $\framebox{$3$} \ \begin{array}{cc} 1_T & 0 \\ \hline \ovalnode{\relax}{0_E} & 1 \\ 0_T & 1 \end{array}$
& $\framebox{$6$} \ \begin{array}{cc} 1_T & 1 \\ \hline \ovalnode{\relax}{0_E} & \rnode{01}{2} \\ 0_T & 2 \end{array}$
& \multicolumn{4}{|c||}{\cellcolor[gray]{0.8}} & shift on $($ \\ \cline{2-8}

\pspar

\nccurve[angleA=-30,angleB=30,ncurvA=10,ncurvB=10]{01}{01}

& $($ & $($ \hspace{2cm} $E$ & $)$ & $a$ & $)$ & $\dashv$ & \\ \cline{2-6}
& $\framebox{$3$} \ \begin{array}{cc} 1_T & 0 \\ \hline \ovalnode{\relax}{0_E} & 1 \\ 0_T & 1 \end{array}$
& $\framebox{$6$} \ \begin{array}{cc} 1_T & 1 \\ \hline \ovalnode{\relax}{0_E} & 2 \\ 0_T & 2 \end{array}$
\quad \vline \qquad $\framebox{$7$} \ \begin{array}{cc} 2_T & 1 \end{array}$ \quad
& \multicolumn{4}{|c||}{\cellcolor[gray]{0.8}} & reduction $\varepsilon \leadsto E$ \par and shift on $E$ \\ \cline{2-8}

& $($ & $($ \hspace{2cm} $E$ & $)$ & $a$ & $)$ & $\dashv$ & \\ \cline{2-6}
& $\framebox{$3$} \ \begin{array}{cc} 1_T & 0 \\ \hline \ovalnode{16}{0_E} & 1 \\ \rnode{15}{0_T} & \rnode{14}{1} \end{array}$
& $\framebox{$6$} \ \begin{array}{cc} 1_T & \rnode{13}{1} \\ \hline \ovalnode{\relax}{0_E} & 2 \\ 0_T & 2 \end{array}$
\quad \vline \qquad $\framebox{$7$} \ \begin{array}{cc} 2_T & \rnode{12}{1} \end{array}$ \quad
& $\framebox{$5$} \ \begin{array}{cc} \ovalnode{\relax}{3_T} & \rnode{11}{1} \end{array}$
& \multicolumn{3}{|c||}{\cellcolor[gray]{0.8}} & shift on $)$ \\ \cline{2-8}

\pspar

\nccurve[angleA=-135,angleB=-45,ncurvA=0.5]{11}{14}

\nccurve[linestyle=dotted,linewidth=2pt,angleA=-180,angleB=135,ncurvA=4,ncurvB=3.5]{15}{16}

& $($ & \multicolumn{2}{c|}{T} & $a$ & $)$ & $\dashv$ & \\ \cline{2-6}
& $\framebox{$3$} \ \begin{array}{cc} 1_T & 0 \\ \hline \ovalnode{\relax}{0_E} & 1 \\ 0_T & 1 \end{array}$
& \multicolumn{2}{c|}{$\framebox{$4$} \ \begin{array}{cc} \ovalnode{\relax}{1_E} & 1 \\ \hline 0_T & 2 \end{array}$}
& \multicolumn{3}{|c||}{\cellcolor[gray]{0.8}} & reduction $( \, E \, ) \leadsto T$ \par and shift on $T$ \\ \cline{2-8}

& $($ & \multicolumn{2}{c|}{T} & $a$ & $)$ & $\dashv$ & \\ \cline{2-6}
& $\framebox{$3$} \ \begin{array}{cc} 1_T & 0 \\ \hline \ovalnode{\relax}{0_E} & 1 \\ 0_T & 1 \end{array}$
& \multicolumn{2}{c|}{$\framebox{$4$} \ \begin{array}{cc} \ovalnode{24}{1_E} & 1 \\ \hline \rnode{23}{0_T} & \rnode{22}{2} \end{array}$}
& $\framebox{$5$} \ \begin{array}{cc} \ovalnode{\relax}{3_T} & \rnode{21}{2} \end{array}$
& \multicolumn{2}{|c||}{\cellcolor[gray]{0.8}} & shift on $a$ \\ \cline{2-8}

\pspar

\nccurve[angleA=145,angleB=35]{21}{22}

\nccurve[linestyle=dotted,linewidth=2pt,angleA=-165,angleB=165,ncurvA=3.5,ncurvB=3.5]{23}{24}

& $($ & \multicolumn{2}{c|}{T} & $T$ & $)$ & $\dashv$ & \\ \cline{2-6}
& $\framebox{$3$} \ \begin{array}{cc} \rnode{35}{1_T} & 0 \\ \hline \ovalnode{34}{0_E} & \rnode{33}{1} \\ 0_T & 1 \end{array}$
& \multicolumn{2}{c|}{$\framebox{$4$} \ \begin{array}{cc} \ovalnode{\relax}{1_E} & \rnode{32}{1} \\ \hline 0_T & 2 \end{array}$}
& $\framebox{$4$} \ \begin{array}{cc} \ovalnode{\relax}{1_E} & \rnode{31}{1} \\ \hline 0_T & 3 \end{array}$
& \multicolumn{2}{|c||}{\cellcolor[gray]{0.8}} & reduction $a \leadsto T$ \par and shift on $T$ \\ \cline{2-8}

\pspar

\nccurve[angleA=145,angleB=10]{31}{33}

\nccurve[linestyle=dotted,linewidth=2pt,angleA=-145,angleB=180,ncurvA=4,ncurvB=3]{34}{35}

& $($ & \multicolumn{3}{c|}{E} & $)$ & $\dashv$ & \\ \cline{2-6}
& $\framebox{$3$} \ \begin{array}{cc} 1_T & 0 \\ \hline \ovalnode{\relax}{0_E} & 1 \\ 0_T & 1 \end{array}$
& \multicolumn{3}{c|}{$\framebox{$8$} \ \begin{array}{cc} 2_T & 0 \end{array}$}
& \multicolumn{2}{|c||}{\cellcolor[gray]{0.8}} & reduction $T \, T \leadsto E$ \par and shift on $E$ \\ \cline{2-8}

& $($ & \multicolumn{3}{c|}{E} & $)$ & $\dashv$ & \\ \cline{2-6}
& $\framebox{$3$} \ \begin{array}{cc} 1_T & \rnode{43}{0} \\ \hline \ovalnode{\relax}{0_E} & 1 \\ 0_T & 1 \end{array}$
& \multicolumn{3}{c|}{$\framebox{$8$} \ \begin{array}{cc} 2_T & \rnode{42}{0} \end{array}$}
& $\framebox{$2$} \ \begin{array}{cc} \ovalnode{\relax}{3_T} & \rnode{41}{0} \end{array}$
& \multicolumn{1}{|c||}{\cellcolor[gray]{0.8}} & shift on $)$ \\ \cline{2-8}

\pspar

\nccurve[linestyle=dotted,linewidth=2pt,angleA=-165,angleB=165,ncurvA=3.5,ncurvB=3.5]{45}{46}

& \multicolumn{5}{c|}{T} & $\dashv$ & \\ \cline{2-6}
& \multicolumn{5}{c|}{$\framebox{$1$} \ \def\arraystretch{2}\begin{array}{cc} \ovalnode{\relax}{1_E} & \rnode{51}{0} \\ \hline 0_T & 1 \end{array}$}
& \multicolumn{1}{|c||}{\cellcolor[gray]{0.8}} & reduction $( \, E \, ) \leadsto T$ \par and shift on $T$ \\ \cline{2-8}

& \multicolumn{5}{|c|}{$E$} & $\dashv$ & \\ \cline{2-6}
& \multicolumn{6}{|c||}{\cellcolor[gray]{0.8} $\begin{array}{c} \\ \\ \\ \end{array}$} & reduction $T \leadsto E$ \par and accept without \par shifting on $E$
\pspar
\nccurve[angleA=-120,angleB=-105,ncurvA=0.625,ncurvB=2.0]{41}{44}
\nccurve[angleA=-45,angleB=-115,ncurvA=1.0,ncurvB=2.375]{51}{52}
\\ \hline\hline
\end{tabular}}
\caption{Tabulation of parsing steps of the parser using a vector-stack for string $( \, ( \, ) \, a \, )$ generated by the grammar in Figure \ref{figRunningExELRpart1} having the $ELR \, (1)$ pilot of Figure \ref{figRunningExELRpilot}.}
\label{figExRunningParseTraceVector}
\end{figure}
\par
\begin{figure}[htbp!]
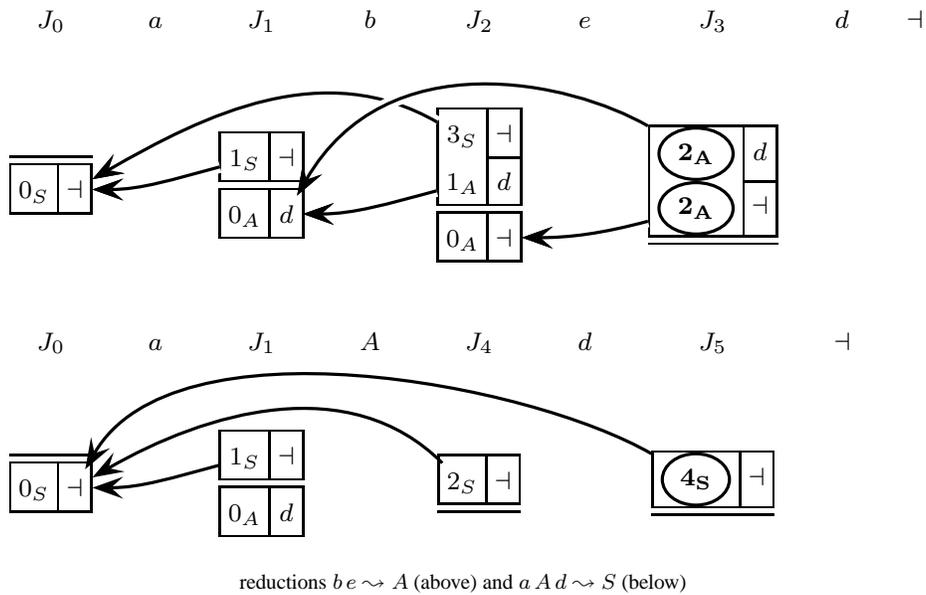

\begin{flushleft}
\emph{Pilot Graph}
\end{flushleft}
\begin{center}
\vspace{0.25cm}
\scalebox{1.2}{
\pspar\psset{border=0pt,arrows=->,nodesep=0pt,colsep=1.5cm,rowsep=1cm,labelsep=5pt}
\begin{psmatrix}

\rnode{I0}{$\def\arraystretch{1.5}\begin{array}{|c|c|} \hline\hline
0_S & \dashv \\ \hline
\end{array}$}

& \rnode{I1}{$\def\arraystretch{1.5}\begin{array}{|c|c|} \hline
1_S & \dashv \\ \hline\hline
0_A & d \\ \hline
\end{array}$}

& \rnode{I2}{$\def\arraystretch{1.5}\begin{array}{|c|c|} \hline
3_S & \dashv \\ \cline{2-2}
1_A & d \\ \hline\hline
0_A & \dashv \\ \hline
\end{array}$}

& \rnode{I3}{$\def\arraystretch{1.5}\begin{array}{|c|c|} \hline
\ovalnode{\relax}{\mathbf{2_A}} & d \; \dashv \\ \hline\hline
\end{array}$} \\

& \rnode{I4}{$\def\arraystretch{1.5}\begin{array}{|c|c|} \hline
2_S& \dashv \\ \hline\hline
\end{array}$}

& \rnode{I5}{$\def\arraystretch{1.5}\begin{array}{|c|c|} \hline
\ovalnode{\relax}{\mathbf{4_S}} & \dashv \\ \hline\hline
\end{array}$} \\

\nput[labelsep=5pt]{90}{I0}{$I_0$}

\nput[labelsep=5pt]{90}{I1}{$I_1$}

\nput[labelsep=5pt]{90}{I2}{$I_2$}

\nput[labelsep=5pt]{90}{I3}{$I_3$}

\nput[labelsep=5pt]{-90}{I4}{$I_4$}

\nput[labelsep=5pt]{-90}{I5}{$I_5$}

\ncline{I0}{I1} \naput{$a$}

\ncline{I1}{I2} \naput{$b$}

\ncline[doubleline=true,arrowscale=1.5]{I2}{I3} \naput{$e$} \nbput{\parbox{1.5cm}{\centering convergent \par edge}}

\ncline{I2}{I5} \naput{$A$}

\ncline{I1}{I4} \nbput{$A$}

\ncline{I4}{I5} \nbput{$d$}

\end{psmatrix}}
\end{center}
\hrule
\begin{flushleft}
\emph{Parse Traces}
\end{flushleft}
\begin{center}
\vspace{0.25cm}
\scalebox{1.2}{
\pspar\psset{border=0.0cm,nodesep=3pt,labelsep=5pt,rowsep=0.25cm,colsep=0.625cm}
\begin{psmatrix}
$J_0$ & $a$ & $J_1$ & $b$ & $J_2$ & $e$ & $J_3$ & $d$ & $\dashv$ \\ \\

\rnode{I0}{$\def\arraystretch{1.5}\begin{array}{|c|c|} \hline\hline
0_S & \rnode{0s1p}{\dashv} \\ \hline
\end{array}$}

&&

\rnode{I1}{$\def\arraystretch{1.5}\begin{array}{|c|c|} \hline
\rnode{1s1s}{1_S} & \rnode{1s1p}{\dashv} \\ \hline\hline
\rnode{1s2s}{0_A} & \rnode{1s2p}{d} \\ \hline
\end{array}$}

&&

\rnode{I2}{$\def\arraystretch{1.5}\begin{array}{|c|c|} \hline
\rnode{2s1s}{3_S} & \dashv \\ \cline{2-2}
\rnode{2s2s}{1_A} & \rnode{2s2p}{d} \\ \hline\hline
0_A & \rnode{2s3p}{\dashv} \\ \hline
\end{array}$}

&&

\rnode{I3}{$\def\arraystretch{1.5}\begin{array}{|c|c|} \hline
\rnode{3s1s}{\ovalnode{\relax}{\mathbf{2_A}}} & d \\ \cline{2-2}
\rnode{3s2s}{\ovalnode{\relax}{\mathbf{2_A}}} & \dashv \\ \hline\hline
\end{array}$} \\

\psset{border=0.05cm}

\nccurve[angleA=-165,angleB=0,ncurvA=0.75,ncurvB=0.75]{1s1s}{0s1p}

\nccurve[angleA=150,angleB=30,ncurvA=0.75,ncurvB=0.75]{2s1s}{0s1p}

\nccurve[angleA=-165,angleB=0,ncurvA=0.75,ncurvB=0.75]{2s2s}{1s2p}

\nccurve[angleA=150,angleB=60,ncurvA=0.5,ncurvB=1]{3s1s}{1s2p}

\nccurve[angleA=-165,angleB=0,ncurvA=0.75,ncurvB=0.75]{3s2s}{2s3p}

\\

$J_0$ & $a$ & $J_1$ & $A$ & $J_4$ & $d$ & $J_5$ & $\dashv$ \\ \\

\rnode{I0}{$\def\arraystretch{1.5}\begin{array}{|c|c|} \hline\hline
0_S & \rnode{0s1p}{\dashv} \\ \hline
\end{array}$}

&&

\rnode{I1}{$\def\arraystretch{1.5}\begin{array}{|c|c|} \hline
\rnode{1s1s}{1_S} & \rnode{1s1p}{\dashv} \\ \hline\hline
0_A & d \\ \hline
\end{array}$}

&&

\rnode{I2}{$\def\arraystretch{1.5}\begin{array}{|c|c|} \hline
\rnode{2s1s}{2_S} & \rnode{2s1p}{\dashv} \\ \hline\hline
\end{array}$}

&&

\rnode{I3}{$\def\arraystretch{1.5}\begin{array}{|c|c|} \hline
\rnode{3s1s}{\ovalnode{\relax}{\mathbf{4_S}}} & \dashv \\ \hline\hline
\end{array}$}

\nccurve[angleA=-165,angleB=0,ncurvA=0.75,ncurvB=0.75]{1s1s}{0s1p}

\nccurve[angleA=135,angleB=30,ncurvA=0.75,ncurvB=0.75]{2s1s}{0s1p}

\nccurve[angleA=150,angleB=60,ncurvA=0.5,ncurvB=0.625]{3s1s}{0s1p}

\end{psmatrix}}
\begin{center}
\vspace{0.25cm}
reductions $b\,e \leadsto A$ (above) and $a\,A\,d \leadsto S$ (below)
\end{center}
\end{center}
\caption{Parsing steps of the parser using a vector-stack for string $a\,b\,e\,d\,\dashv$ recognized by the net in Figure \ref{figExELRconvergentGrammarPilotTrace}, with the pilot here reproduced for convenience.} \label{figConvergentEdgeNetParseTraceVector}
\end{figure}
\par
In  every stack element of type $vsms$ the second field of each candidate is
the \emph{elemid} index, which points back to some inner position of the stack; \emph{elemid} is equal to
the current stack position for all the candidates in the closure part of a stack element, and points
to some previous position for the candidates of the base. As in Figure \ref{figExRunningParseTrace}, the
reduction handles are highlighted by means of solid backward pointers, and by a dotted arrow to locate the
candidate to be shifted soon after reducing. Notice that now the arrows span a longer distance than
in Figure \ref{figExRunningParseTrace}, as the \emph{elemid} goes directly to the origin of the reduction handle.
The forward shift arrows are the same as those in Fig. \ref{figExRunningParseTrace} and are here not shown.
\par
The results of the analysis, i.e., the execution order of the reductions and the obtained syntax tree,
are identical to those of Example \ref{exRunningParsingTrace}.
\par
We illustrate the use of the vector-stack on the previous example of a net featuring a convergent edge
(net in Figure \ref{figExELRconvergentGrammarPilotTrace}): the parsing traces are shown
in Figure \ref{figConvergentEdgeNetParseTraceVector}, where the integer pointers are represented
by backward pointing solid arrows.
\end{example}
\subsection{Related work on shift-reduce parsing of \emph{EBNF} grammars.}\label{sectRelatedWorkELR}
Over the years many contributions  have been published to extend Knuth's $LR \, (k)$ method to $EBNF$ grammars. The very number of papers, each one purporting to improve over previous attempts, testifies that no clear-cut optimal solution has been found. Most papers usually start with critical reviews of related proposals, from which we can grasp the difficulties and motivations then perceived. The following discussion particularly draws from the later papers \cite{MorSas01,Kannapin2001,conf/lata/Hemerik09}.
\par
The first dichotomy concerns which format of $EBNF$ specification is taken as input: either a grammar with \emph{regular expressions} in the right parts, or a grammar with \emph{finite-automata} ($FA$) as right parts. Since it is nowadays perfectly clear that r.e. and $FA$ are interchangeable notations for regular languages, the distinction is no longer relevant. Yet some authors  insisted that a language designer should be allowed to specify syntax constructs by arbitrary r.e.'s, even ambiguous ones, which allegedly permit a more flexible mapping from syntax to semantics. In their view, which we do not share, transforming  the original r.e.'s to $DFA$'s is not entirely satisfactory.
\par
Others have imposed restrictions on the r.e.'s, for instance  limiting the depth of Kleene star nesting or forbidding common subexpressions, and so on. Although the original motivation to simplify parser construction has since vanished, it is fair to say that the r.e.'s used in language reference manuals are typically very simple, for the reason of avoiding obscurity.
\par
Others prefer to specify  right parts by using the very readable graphical notations of syntax diagrams, which are a pictorial variant of FA state-transition diagrams. Whether the $FA$'s are deterministic or non- does not really make a difference, either in terms of grammar readability or ease of parser generation. Even when the source specification includes $NFA$'s, it is as simple to transform them into $DFA$'s by the standard construction, as it is to leave the responsibility of removing finite-state non-determinism to the algorithm that constructs the $LR \, (1)$ automaton (pilot).
\par
On the other hand, we have found that an inexpensive normalization of $FA$'s (disallowing reentrance into initial states, Def. \ref{retiRicorsMacch}) pays off in terms of parser construction simplification.
\par
Assuming that the grammar is specified by a net of $FA$'s, two approaches for building a parser have been followed: (A) transform the grammar into a $BNF$ grammar and apply Knuth's $LR \, (1)$ construction, or (B) directly construct an $ELR \, (1)$ parser from the given machine net. It is generally agreed that ``approach (B) is better than approach (A) because the transformation adds inefficiency and makes it harder to determine the semantic structure due to the additional structure added by the transformation'' \cite{MorSas01}. But since approach (A) leverages on existing parser generators such as Bison, it is quite common for language reference manuals featuring syntax chart notations to include also an equivalent $BNF$ $LR \, (1)$ (or even $LALR \, (1)$) grammar. In \cite{Celentano:1981:LPT} a systematic transformation from $EBNF$ to $BNF$ is used to obtain, for an $EBNF$ grammar, an $ELR \, (1)$ parser that simulates the  classical Knuth's parser for the $BNF$ grammar.
\par
The technical difficulty of approach (B), which all authors had to deal with, is how to identify the left end of a reduction handle, since its length is variable and possibly unbounded. A list of different solutions can be found in the already cited surveys. In particular, many algorithms (including ours) use a special shift move, sometimes called \emph{stack-shift}, to record into the stack the left end of the handle when a new computation on a net machine is started.  But, whenever such algorithms permit an initial state to be reentered, a conflict between stack-shift and normal shift is unavoidable, and various  devices have been invented to arbitrate the conflict. Some add read-back states to control how  the parser should dig into the stack \cite{journals/acta/Chapman84,journals/acta/LaLonde79}, while others (e.g., \cite{journals/ipl/SassaNakata87}) use counters for the same purpose, not to mention other proposed devices. Unfortunately it was shown in \cite{journals/ipl/Galvez94,Kannapin2001} that several proposals do not precisely characterize the grammars they apply to, and in some cases may fall into unexpected errors.
\par
Motivated by the mentioned flaws of previous attempts, the paper by \cite{journals/ipl/LeeK97} aims at characterizing the $LR \, (k)$ property for $ECFG$ grammars defined  by a network of FA's. Although their definition is intended to ensure that such grammars ``can be parsed from left to right with a look-ahead of $k$ symbols'', the authors admit that ``the subject of efficient techniques for locating the left end of a handle is beyond the scope of this paper''.
\par
After such a long  history of interesting but non-conclusive proposals, our finding of a rather simple formulation  of the $ELR \, (1)$ condition leading to a naturally corresponding  parser, was rather unexpected. Our definition simply adds the treatment of convergent edges to Knuth's definition. The technical difficulties were well understood since long and we combined existing ideas into a simple and provably correct solution.
Of course, more experimental work would be needed to evaluate the performance of our algorithms on real-sized grammars.
\section{Deterministic top-down parsing}\label{sectDeterministicTopDownParsing}
A simpler and very flexible top-down parsing method, traditionally called $ELL \, (1)$,\footnote{Extended, Left to right, Leftmost, with length of look-ahead equal to one.} applies if an $ELR \, (1)$ grammar satisfies further conditions.  Although less general than the $ELR \, (1)$, this  method has several assets, primarily the ability to  anticipate parsing decisions thus offering better support for syntax-directed translation, and to be implemented by a neat modular structure made of recursive procedures that mirror the graphs of network machines.
\par
In the next sections, our presentation rigorously derives step by step the properties of top-down deterministic parsers, as we add  one by one some simple restrictions to the $ELR( \, 1)$ condition. First, we consider the Single Transition Property and how it simplifies  the shift-reduce parser: the number of m-states is reduced to the number of net states, convergent edges are no longer possible and the chain of stack pointers can be disposed.
\par
Second, we add the requirement that the grammar is not left-recursive and we obtain the traditional predictive top-down parser, which constructs the syntax tree in pre-order.
\par
At last, the direct construction of $ELL \, (1)$ parsers sums up.
\paragraph*{Historical note} In contrast to the twisted  story of $ELR \, (1)$ methods, early efforts to develop top-down parsing algorithms for $EBNF$ grammars have met with remarkable success and we do not need to critically discuss them, but just to cite the main references and to explain why our work adds value to them. Deterministic parsers operating top-down were among the first to be constructed by compilation pioneers, and their theory  for $BNF$ grammars was shortly after developed  by \cite{Rosenkrantz:Stearns:ic:1970,Knuth71}. A sound  method to extend such parsers to $EBNF$ grammars was popularized by \cite{Wirth75} recursive-descent compiler, systematized in the book \cite{LewiDeVlaHuensHuy79}, and  included in widely known compiler textbooks (e.g., \cite{AhoLamSethiUl2006}). However in such books  top-down deterministic parsing is  presented \emph{before} shift-reduce methods and independently of them, presumably because it is easier to understand. On the contrary, this section shows that  top-down parsing for $EBNF$ grammars is a corollary of the  $ELR \, (1)$ parser construction that we have just presented. Of course, for pure $BNF$ grammars the relationship between $LR \, (k)$ and $LL \, (k)$  grammar and language families has been carefully  investigated in the past, see in particular  \cite{Beatty82}. Building on the concept of multiple transitions, which we introduced for $ELR \, (1)$  analysis, we extend Beatty's characterization to the $EBNF$ case and we derive in a minimalist and provably correct way the $ELL \, (1)$ parsing algorithms. As a by-product of our unified approach, we mention in the Conclusion the use of heterogeneous parsers for different language parts.
\subsection{Single-transition property and pilot compaction}\label{subsectSingletonBaseProperty}
Given an  $ELR \, (1)$ net $\mathcal{M}$ and its $ELR \, (1)$ pilot $\mathcal{P}$, recall that a m-state has the multiple-transition property \ref{defConvergence} ($MTP$) if two identically labeled state transitions originate from two candidates present in the m-state. For brevity we also say that such m-state violates the \emph{single-transition property} ($STP$).
The next example illustrates several cases of violation.
\begin{figure}[htb!]
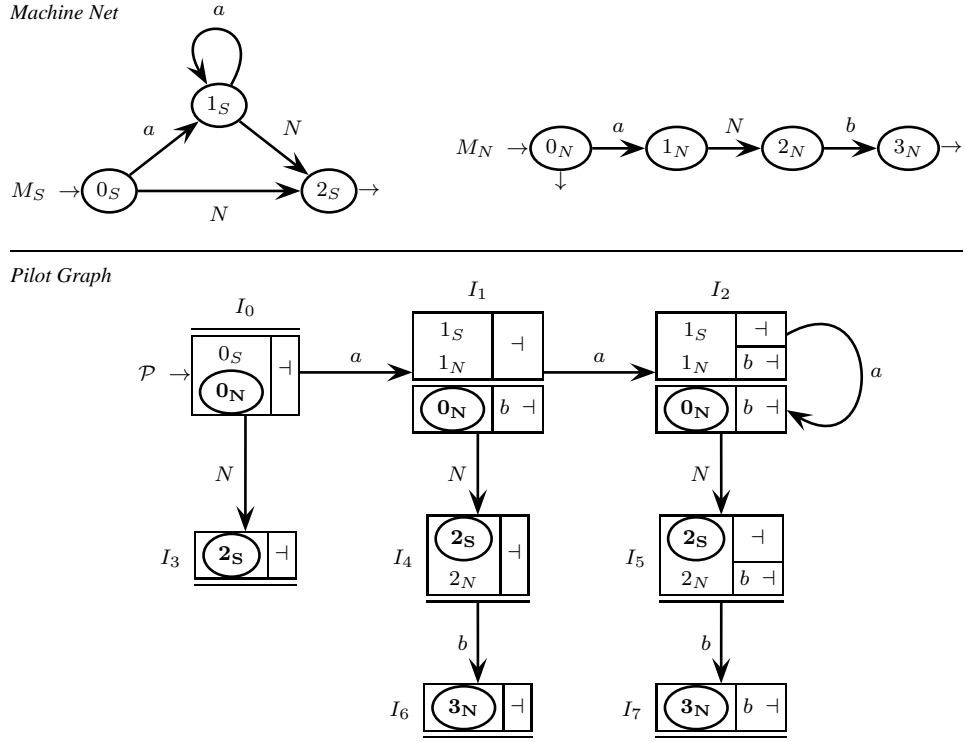

\begin{flushleft}
\emph{Machine Net}
\end{flushleft}
\begin{center}
\vspace{0.25cm}
\begin{tabular}{cp{0.0cm}c}
\begin{minipage}{0.40\textwidth}
\begin{center}
\scalebox{1.0}{
\pspar\psset{border=0pt,nodesep=0pt,colsep=20pt,labelsep=5pt,rowsep=0.5cm}

\begin{psmatrix}

& \ovalnode{1S}{$1_S$} \\

\ovalnode{0S}{$0_S$} && \ovalnode{2S}{$2_S$}

\nput[labelsep=0pt]{180}{0S}{$M_S \ \to$}

\nput[labelsep=0pt]{0}{2S}{$\to$}

\ncline{0S}{1S} \naput{$a$}

\ncline{1S}{2S} \naput{$N$}

\ncline{0S}{2S} \nbput{$N$}

\nccurve[angleA=60,angleB=120,ncurvA=7,ncurvB=7]{1S}{1S} \nbput{$a$}

\end{psmatrix}}
\end{center}
\end{minipage}

&&

\begin{minipage}{0.55\textwidth}
\begin{center}
\scalebox{1.0}{
\pspar\psset{border=0pt,nodesep=0pt,colsep=20pt,labelsep=5pt}
\begin{psmatrix}
\ovalnode{0N}{$0_N$} & \ovalnode{1N}{$1_N$} & \ovalnode{2N}{$2_N$} & \ovalnode{3N}{$3_N$}

\nput[labelsep=0pt]{180}{0N}{$M_N \ \to$}

\nput[labelsep=0pt]{-90}{0N}{$\downarrow$}

\nput[labelsep=0pt]{0}{3N}{$\to$}

\ncline{0N}{1N} \naput{$a$}

\ncline{1N}{2N} \naput{$N$}

\ncline{2N}{3N} \naput{$b$}

\end{psmatrix}}
\end{center}
\end{minipage}
\end{tabular}
\vspace{0.45cm}
\hrule
\begin{flushleft}
\emph{Pilot Graph}
\end{flushleft}
\vspace{0.1cm}
\begin{minipage}{1.0\textwidth}
\begin{center}
\scalebox{1.0}{
\pspar\psset{border=0pt,nodesep=0pt,rowsep=30pt}
\begin{psmatrix}

\rnode{I0}{$\def\arraystretch{1.25}\begin{array}{|c|c|} \hline\hline
0_S & \multirow{2}{*}{$\dashv$} \\
\ovalnode{\relax}{\mathbf{0_N}} & \\ \hline
\end{array}$} &

\rnode{I1}{$\def\arraystretch{1.25}\begin{array}{|c|c|} \hline
1_S & \multirow{2}{*}{$\dashv$} \\
1_N & \\ \hline\hline
\ovalnode{\relax}{\mathbf{0_N}} & b \; \dashv \\ \hline
\end{array}$} &

\rnode{I2}{$\def\arraystretch{1.25}\begin{array}{|c|c|} \hline
1_S & \dashv \\ \cline{2-2}
1_N & b \; \dashv \\ \hline\hline
\ovalnode{\relax}{\mathbf{0_N}} & b \; \dashv \\ \hline
\end{array}$} \\

\rnode{I3}{$\def\arraystretch{1.25}\begin{array}{|c|c|} \hline
\ovalnode{\relax}{\mathbf{2_S}} & \dashv \\ \hline\hline
\end{array}$} &

\rnode{I4}{$\def\arraystretch{1.25}\begin{array}{|c|c|} \hline
\ovalnode{\relax}{\mathbf{2_S}} & \multirow{2}{*}{$\dashv$} \\
2_N & \\ \hline\hline
\end{array}$} &

\rnode{I5}{$\def\arraystretch{1.25}\begin{array}{|c|c|} \hline
\ovalnode{\relax}{\mathbf{2_S}} & \dashv \\ \cline{2-2}
2_N & b \; \dashv \\ \hline\hline
\end{array}$} \\ &

\rnode{I6}{$\def\arraystretch{1.25}\begin{array}{|c|c|} \hline
\ovalnode{\relax}{\mathbf{3_N}} & \dashv \\ \hline\hline
\end{array}$} &

\rnode{I7}{$\def\arraystretch{1.25}\begin{array}{|c|c|} \hline
\ovalnode{\relax}{\mathbf{3_N}} & b \; \dashv \\ \hline\hline
\end{array}$}

\nput[labelsep=5pt]{90}{I0}{$I_0$}

\nput[labelsep=5pt]{90}{I1}{$I_1$}

\nput[labelsep=5pt]{90}{I2}{$I_2$}

\nput[labelsep=5pt]{180}{I3}{$I_3$}

\nput[labelsep=5pt]{180}{I4}{$I_4$}

\nput[labelsep=5pt]{180}{I5}{$I_5$}

\nput[labelsep=5pt]{180}{I6}{$I_6$}

\nput[labelsep=5pt]{180}{I7}{$I_7$}

\nput[labelsep=0pt]{180}{I0}{$\mathcal{P} \; \to$}

\ncline{I0}{I1} \naput{$a$}

\ncline{I1}{I2} \naput{$a$}

\nccurve[angleA=30,angleB=-30,ncurvA=3,ncurvB=3]{I2}{I2} \naput{$a$}

\ncline{I0}{I3} \nbput{$N$}

\ncline{I1}{I4} \nbput{$N$}

\ncline{I2}{I5} \nbput{$N$}

\ncline{I4}{I6} \nbput{$b$}

\ncline{I5}{I7} \nbput{$b$}

\end{psmatrix}}
\end{center}
\end{minipage}
\end{center}
\caption{ $ELR(1)$ net of Ex. \ref{exELRnonsingletonBase} with multiple candidates in the m-state base of $I_1$, $I_2$, $I_4$, $I_5$.}\label{figMultiCandidatesInBase}
\end{figure}
\begin{example}\label{exELRnonsingletonBase} Violations of $STP$.
\par
Three cases are examined. First, grammar $S \to a^\ast N$, $N \to a \; N \; b \; \mid \; \varepsilon$, generating the deterministic non-$LL \, (k)$ language $\set{ \; a^n \, b^m \; \vert \quad n \geq m \geq  0 \; }$, is represented by the net in Fig. \ref{figMultiCandidatesInBase}, top. The presence of two candidates in $I_{1 \vert base}$ (and also in other m-states) reveals that the parser has to carry on two simultaneous attempts at parsing until a reduction takes place, which is unique, since the pilot satisfies the $ELR \, (1)$ condition: there are neither shift-reduction nor reduction-reduction conflicts, nor convergence conflicts.
\par
Second, the net in Figure \ref{figExELRconvergentGrammarPilotTrace} illustrates the case of multiple candidates in the base of a m-state ($I_3$) that is entered by a convergent edge.
\par
Third, grammar $S \to b \, \left( \, a \; \mid \; S \, c \, \right) \; \mid \; a$ has a pilot that violates $STP$, yet it contains only one candidate in each m-state base.
\end{example}
On the other hand, if  every m-state base contains one candidate, the parser configuration space can be reduced as well as the range of parsing choices. Furthermore, $STP$ entails that there are no convergent edges in the pilot.
\par
Next we show that m-states having the same kernel, qualified for brevity as \emph{kernel-identical}, can be safely coalesced, to obtain a smaller pilot that is equivalent to the  original one and bears closer resemblance to the machine net. We hasten to say that such transformation does not work in general for an $ELR \, (1)$ pilot, but it will be proved to be correct under the $STP$ hypothesis.
\subsubsection{Merging kernel-identical m-states}\label{sectPilotCompaction}
 The \emph{merging} operation coalesces two kernel-identical m-states $I_1$ and $I_2$, suitably adjusts the pilot graph, then possibly merges more kernel-identical m-states. It is defined as follows:
\begin{algorithm}\label{algMergemMacrostates}$Merge \, ( \, I_1, \, I_2 \, )$
\par
\begin{enumerate}
\item replace $I_1$, $I_2$ by a new kernel-identical m-state, denoted by $I_{1, \, 2}$, where each look-ahead set is the union of the corresponding ones in the merged m-states:
\[
\left\langle p, \, \pi \right\rangle \in I_{1, \, 2} \iff \left\langle p, \, \pi_1 \right\rangle \in I_1 \ \text{and} \ \left\langle p, \, \pi_2 \right\rangle \in I_2 \ \text{and} \ \pi = \pi_1 \cup \pi_2
\]
\item the m-state $I_{1, \, 2}$ becomes the target for all the edges that entered $I_1$ or $I_2$:
\[
I \stackrel {X} {\longrightarrow} I_{1, \, 2} \iff I \stackrel X \longrightarrow I_{1} \ \text{or} \ I \stackrel X \longrightarrow I_{2}
\]	
\item for each pair of  edges from $I_1$ and $I_2$, labeled $X$, the target m-states (which clearly are kernel-identical) are merged:
\[
\text{if} \ \vartheta \, ( \, I_1, \, X \, ) \neq \vartheta \, ( \, I_2, \, X \, ) \ \text{then call} \ Merge \, \big( \, \vartheta \, (I_1, \, X \, ), \, \vartheta \, ( \, I_2, \, X \, ) \big)
\]
\end{enumerate}
\end{algorithm}
Clearly the merge operation  terminates with a graph with fewer nodes. The set of all kernel-equivalent m-states is an equivalence class. By applying the \emph{Merge} algorithm to the members of every equivalence class, we construct a new  graph, to be called the \emph{compact pilot} and denoted by $\mathcal{C}$.\footnote{For the compact pilot the number of m-states is the same as for the $LR \, (0)$ and $LALR \, (1)$ pilots, which are historical simpler variants of $LR \, (1)$ parsers not considered here, see for instance \cite{Crespi09}. However neither $LR \, (0)$ nor  $LALR \, (1)$ pilots have to comply with the $STP$ condition.}
\par
\begin{example}\label{exCompactLL1pilot}Compact  pilot.\\
We reproduce in Figure \ref{figCompactELL1versusELR1pilot} the machine net, the original $ELR \, (1)$ pilot and, in the bottom part, the compact pilot, where for convenience the m-states have been renumbered. Notice that the look-ahead sets have expanded: e.g., in $K_{1_T}$ the look-ahead in the first row is the union of the corresponding look-aheads in the merged m-states $I_3$ and $I_6$. We are going to prove that this loss of precision is not harmful for parser determinism thanks to the stronger constraints imposed by $STP$.
\end{example}
\begin{figure}[h!]
\begin{flushleft}
\emph{Machine Net}
\end{flushleft}
\begin{center}
\vspace{0.5cm}
\begin{tabular}{cp{0.5cm}c}
\begin{minipage}{0.3\textwidth}
\begin{center}
\scalebox{1.0}{
\pspar\psset{border=0pt,nodesep=0pt,labelsep=5pt,arrows=->,colsep=0.8cm,rowsep=10pt}
\begin{psmatrix}

\ovalnode{0E}{$0_E$} & \ovalnode{1E}{$1_E$}

\nput[labelsep=0pt]{180}{0E}{$M_E \ \to$}

\nput[labelsep=0pt]{-90}{1E}{$\downarrow$}

\nput[labelsep=0pt]{-90}{0E}{$\downarrow$}

\ncline{0E}{1E} \naput{$T$}

\nccurve[angleA=-30,angleB=30,ncurvA=7,ncurvB=7]{1E}{1E} \nbput{$T$}

\end{psmatrix}}
\end{center}
\end{minipage}

&&

\begin{minipage}{0.65\textwidth}
\begin{center}
\scalebox{1.0}{
\pspar\psset{border=0pt,nodesep=0pt,labelsep=5pt,arrows=->,colsep=0.8cm,rowsep=10pt}
\begin{psmatrix}

\ovalnode{0T}{$0_T$} & \ovalnode{1T}{$1_T$} & \ovalnode{2T}{$2_T$} & \ovalnode{3T}{$3_T$}

\nput[labelsep=0pt]{180}{0T}{$M_T \ \to$}

\nput[labelsep=0pt]{0}{3T}{$\to$}

\ncline{0T}{1T} \nbput{$($}

\ncline{1T}{2T} \nbput{$E$}

\ncline{2T}{3T} \nbput{$)$}

\ncarc[arcangle=45]{0T}{3T} \naput{$a$}
\end{psmatrix}}
\end{center}
\end{minipage}
\end{tabular}
\end{center}
\vspace{0.25cm}
\hrule
\begin{flushleft}
\emph{Pilot Graph}
\end{flushleft}
\vspace{0.5cm}
\begin{center}
\scalebox{0.85}{
\pspar\psset{arrows=->,border=0pt,nodesep=0pt,rowsep=1.0cm,colsep=1.5cm}
\begin{psmatrix}
\rnode{I0}{$\def\arraystretch{1.25}\begin{array}{|c|c|} \hline\hline
\ovalnode{\relax}{\mathbf{0_E}} & \dashv \\ \cline{2-2}
0_T & a \; ( \; \dashv \\ \hline
\end{array}$}
&
\rnode{I1}{$\def\arraystretch{1.25}\begin{array}{|c|c|} \hline
\ovalnode{\relax}{\mathbf{1_E}} & \dashv \\ \hline\hline
0_T & a \; ( \; \dashv \\ \hline
\end{array}$} \\

\rnode{I3}{$\def\arraystretch{1.25}\begin{array}{|c|c|} \hline
1_T & a \; ( \; \dashv \\
\hline \hline \ovalnode{o}{\mathbf{0_E}} & ) \\ \cline{2-2}
0_T & a \; ( \; ) \\ \hline
\end{array}$}
&
\rnode{I8}{$\def\arraystretch{1.25}\begin{array}{|c|c|} \hline
2_T & a \; ( \; \dashv \\ \hline\hline
\end{array}$}
&
\rnode{I2}{$\def\arraystretch{1.25}\begin{array}{|c|c|} \hline
\ovalnode{\relax}{\mathbf{3_T}} & a \; ( \; \dashv \\ \hline\hline
\end{array}$} \\
&
\rnode{I4}{$\def\arraystretch{1.25}\begin{array}{|c|c|} \hline
\ovalnode{\relax}{\mathbf{1_E}} & ) \\ \hline\hline
0_T & a \; ( \; ) \\ \hline
\end{array}$} \\

\rnode{I6}{$\def\arraystretch{1.25}\begin{array}{|c|c|} \hline
1_T & a \; ( \; ) \\ \hline \hline
\ovalnode{\relax}{\mathbf{0_E}} & ) \\ \cline{2-2}
0_T & a \; ( \; ) \\ \hline
\end{array}$}
&
\rnode{I7}{$\def\arraystretch{1.25}\begin{array}{|c|c|} \hline
2_T & a \; ( \; ) \\ \hline\hline
\end{array}$}
&
\rnode{I5}{$\def\arraystretch{1.25}\begin{array}{|c|c|} \hline
\ovalnode{\relax}{\mathbf{3_T}} & a \; ( \; ) \\ \hline\hline
\end{array}$}

\nput[labelsep=5pt]{90}{I0}{$I_0$}

\nput[labelsep=5pt]{90}{I1}{$I_1$}

\nput[labelsep=5pt]{0}{I2}{$I_2$}

\nput[labelsep=5pt]{180}{I3}{$I_3$}

\nput[labelsep=5pt]{90}{I4}{$I_4$}

\nput[labelsep=5pt]{0}{I5}{$I_5$}

\nput[labelsep=5pt]{180}{I6}{$I_6$}

\nput[labelsep=5pt]{-90}{I7}{$I_7$}

\nput[labelsep=5pt]{-90}{I8}{$I_8$}

\nput[labelsep=0pt]{180}{I0}{$\mathcal{P} \; \to$}

\ncline{I0}{I1} \naput{$T$}

\nccurve[angleA=45,angleB=90,ncurvA=1,ncurvB=0.8]{I0}{I2} \aput(0.85){$a$}

\nccurve[angleA=-30,angleB=130,ncurvA=0.5,ncurvB=0.5]{I1}{I2} \naput{$a$}

\nccurve[angleA=-20,angleB=20,ncurvA=3,ncurvB=3]{I1}{I1} \nbput{$T$}

\ncline{I0}{I3} \nbput{$($}

\nccurve[angleA=-160,angleB=45,ncurvA=0.5,ncurvB=0.5]{I1}{I3} \nbput{$($}

\ncline{I3}{I8} \nbput{$E$}


\ncline{I8}{I2} \nbput{$)$}

\nccurve[angleA=-20,angleB=90,ncurvA=1,ncurvB=0.8]{I3}{I5} \aput(0.85){$a$}

\nccurve[angleA=-30,angleB=130,ncurvA=0.5,ncurvB=0.5]{I4}{I5} \naput{$a$}

\nccurve[angleA=-45,angleB=145,ncurvA=0.5,ncurvB=0.5]{I3}{I4} \nbput{$T$}

\ncline{I3}{I6} \nbput{$($}

\nccurve[angleA=-20,angleB=20,ncurvA=3,ncurvB=3]{I4}{I4} \nbput{$T$}

\ncarc[arcangle=-5]{I6}{I4} \bput(0.75){$T$}

\ncarc[arcangle=-15]{I4}{I6} \bput(0.75){$($}

\nccurve[angleA=-70,angleB=-110,ncurvA=3,ncurvB=3]{I6}{I6} \aput(0.15){$($}

\ncline{I6}{I7} \nbput{$E$}

\ncarc[arcangle=25]{I6}{I5} \naput{$a$}

\ncline{I7}{I5} \nbput{$)$}

\end{psmatrix}}
\end{center}
\vspace{0.75cm}
\hrule
\begin{flushleft}
\emph{Compact Pilot Graph}
\end{flushleft}
\vspace{0.5cm}
\begin{center}
\pspar\psset{arrows=->,border=0pt,nodesep=0pt,rowsep=1.0cm,colsep=1.5cm}
\scalebox{0.85}{
\begin{psmatrix}
\rnode{I0}{$\def\arraystretch{1.25}\begin{array}{|c|c|} \hline\hline
\ovalnode{o}{\mathbf{0_E}} & \dashv \\ \cline{2-2}
0_T & a \; ( \; \dashv \\ \hline
\end{array}$}

&

\rnode{I14}{$\def\arraystretch{1.25}\begin{array}{|c|c|} \hline
\ovalnode{o}{\mathbf{1_E}} & ) \;\dashv \\ \hline\hline
0_T & a \; ( \; ) \; \dashv \\ \hline
\end{array}$} \\

\rnode{I36}{$\def\arraystretch{1.25}\begin{array}{|c|c|} \hline
1_T & a \; ( \; ) \;\dashv \\ \hline\hline
\ovalnode{o}{\mathbf{0_E}} & ) \\ \cline{2-2}
0_T & a \; ( \; ) \\ \hline
\end{array}$}

&

\rnode{I78}{$\def\arraystretch{1.25}\begin{array}{|c|c|} \hline
2_T & a \; ( \; ) \; \dashv \\ \hline\hline
\end{array}$}

&

\rnode{I25}{$\def\arraystretch{1.25}\begin{array}{|c|c|} \hline
\ovalnode{o}{\mathbf{3_T}} & a \; ( \; ) \; \dashv \\ \hline\hline
\end{array}$}

\nput[labelsep=5pt]{90}{I0}{$K_{0_E}\equiv I_0$}

\nput[labelsep=5pt]{90}{I14}{$K_{1_E}\equiv I_{1,4}$}

\nput[labelsep=5pt]{-90}{I25}{$K_{3_T}\equiv I_{2,5}$}

\nput[labelsep=5pt]{180}{I36}{$K_{1_T}\equiv I_{3,6}$}

\nput[labelsep=5pt]{-90}{I78}{$K_{2_T}\equiv I_{7,8}$}

\nput[labelsep=0pt]{180}{I0}{$\mathcal{L}\to$}

\ncline{I0}{I14} \naput{$T$}

\nccurve[angleA=45,angleB=90,ncurvA=1,ncurvB=0.8]{I0}{I25} \aput(0.85){$a$}

\nccurve[angleA=-30,angleB=130,ncurvA=0.5,ncurvB=0.5]{I14}{I25} \naput{$a$}

\nccurve[angleA=-18,angleB=18,ncurvA=3,ncurvB=3]{I14}{I14} \nbput{$T$}

\ncline{I0}{I36} \nbput{$($}

\ncarc[arcangle=-5]{I36}{I14} \bput(0.75){$T$}

\ncarc[arcangle=-15]{I14}{I36} \bput(0.75){$($}

\ncline{I36}{I78} \nbput{$E$}

\ncline{I78}{I25} \nbput{$)$}

\nccurve[angleA=-70,angleB=-110,ncurvA=3,ncurvB=3]{I36}{I36} \aput(0.15){$($}

\ncarc[arcangle=25]{I36}{I25} \naput{$a$}

\end{psmatrix}}
\vspace{0.5cm}
\end{center}
\caption{From top to bottom: machine net $\mathcal{M}$, $ELR \, (1)$ pilot graph $\mathcal{P}$ and compact pilot $\mathcal{C}$; the equivalence classes of m-states are: $\set{ \, I_0 \, }$, $\set{ \, I_1, \, I_4 \, }$, $\set{ \, I_2, \, I_5 \, }$, $\set{ \, I_3, \, I_6 \, }$ and $\set{ \, I_7, \, I_8 \, }$; the m-states of $\mathcal{C}$ are named $K_{0_E}$, \ldots, $K_{3_T}$ to evidence their correspondence with the states of net $\mathcal{M}$.}\label{figCompactELL1versusELR1pilot}
\end{figure}
\par
We anticipate from Section \ref{sectSimplifELL1parserConstr} that the compact pilot can be directly constructed from the machine net, saving some work to compute the larger $ELR \, (1)$ pilot and then to merge kernel-identical nodes.
\subsubsection{Properties of compact pilots}
We are going to show that we can safely use the compact pilot as parser controller.
\begin{property}\label{propCompactPilot}
Let $\mathcal{P}$ and $\mathcal{C}$ be respectively the $ELR \, (1)$  pilot and the compact pilot of a net $\mathcal{M}$ satisfying $STP$\footnote{A weaker hypothesis would suffice: that every m-state has at most one candidate in its base, but for simplicity we have preferred to assume also that convergent edges  are not present.}. The $ELR \, (1)$ parsers controlled by $\mathcal{P}$ and by $\mathcal{C}$ are equivalent, i.e., they recognize language $L \, (\mathcal{M})$ and construct the same syntax tree for every $x \in L \, (\mathcal{M})$.
\end{property}
\begin{proof}
We show that, for every m-state, the compact pilot has no $ELR \, (1)$ conflicts. Since the merge operation does not change the kernel, it is obvious that every $\mathcal{C}$ m-state satisfies  $STP$. Therefore non-initial reduce-reduce conflicts can be excluded, since they involve two candidates in a base,  a condition which is ruled out by $STP$.
\par
Next, suppose by contradiction that a reduce-shift conflict between a final non-initial state $q_A$ and a shift of $q_{B} \stackrel a \to q'_B$, occurs in $I_{1, \, 2}$, but neither in $I_1$ nor in $I_2$. Since $\langle q_A, \, a \rangle$ is in the base of $I_{1, \, 2}$, it must also be in the bases of $I_1$ or $I_2$, and an $a$ labeled edge originates from $I_1$ or $I_2$. Therefore the conflict was already there, in one or both m-states $I_1$, $I_2$.
\par
Next, suppose by contradiction that $I_{1, \, 2}$ has a new initial reduce-shift conflict between an outgoing $a$ labeled edge and an initial-final candidate $\langle 0_{A}, \, a \rangle$. By definition of \emph{merge}, $\langle 0_{A}, \, a \rangle$ is already in the look-ahead set of, say, $I_1$. Moreover, for any two kernel-identical m-states $I_1$ and $I_2$, for any symbol $X \in \Sigma \, \cup \, V$, either both $\vartheta \, (I_1, \, X) = I'_1$ and $\vartheta \, (I_2, \, X) = I'_2$ are defined or neither one, and the m-states $I'_1$, $I'_2$ are kernel-identical. Therefore, an $a$ labeled edge originates from $I_1$, which thus has a reduce-shift conflict, a contradiction.
\par
Last, suppose by contradiction that $I_{1, \, 2}$ contains a new initial reduce-reduce conflict between $0_{A}$ and $0_{B}$ for some terminal character $a$ that is in both look-ahead sets. Clearly in one of the merged m-states, say $I_1$, there are (in the closure) the candidates:
\[
\langle 0_{A}, \, \pi_1 \rangle \ \text{with} \ a \in \pi_1 \quad \text{and} \quad \langle 0_{B}, \, \rho_1 \rangle \ \text{with} \ a \not \in \rho_1
\]
and in $I_2$ there are (in the closure) the candidates:
\[
\langle 0_{A}, \, \pi_2 \rangle \ \text{with} \ a \not \in \pi_2 \quad \text{and} \quad \langle 0_{B}, \, \rho_2 \rangle \ \text{with} \ a \in \rho_2
\]
To show the contradiction, we recall how look-aheads are computed. Let the bases be respectively $\langle q_{C}, \, \sigma_1 \rangle$ for $I_1$ and $\langle q_{C}, \, \sigma_2 \rangle$ for $I_2$.
Then any character in, say, $\pi_1$ comes in two possible ways: 1) it is present in $\sigma_1$, or 2) it is a character that follows some state that is in the closure of $I_1$. Focusing on character $a$, the second possibility is excluded because $a \not \in \pi_2$, whereas the look-ahead elements brought by case 2) are necessarily the same for m-states $I_1$ and $I_2$.
It remains that the presence of  $a$ in $\pi_1$ and in $\rho_2$ comes from its presence in $\sigma_1$ and in $\sigma_2$. Hence $a$ must be also in $\pi_2$, a contradiction.
\par
The parser algorithms differ only in their controllers, $\mathcal{P}$ and $\mathcal{C}$. First, take a string accepted by the parser controlled by $\mathcal{P}$. Since any m-state created by \emph{merge} encodes exactly the same cases for reduction and for shift as the original merged m-states, the parsers will perform exactly the same moves for both pilots. Moreover, the chains of candidate identifiers are clearly identical since the candidate offset are not affected by \emph{merge}. Therefore, at any time the parser stacks store the same elements, up to the merge relation, and the compact parser recognizes the same strings and constructs the same tree.
\par
Second, suppose by contradiction that an illegal string is recognized using the compact parser. For this string consider the first parsing time where $\mathcal{P}$ stops in error in m-state $I_1$, whereas $\mathcal{C}$ is able to move from $I_{1, \, 2}$.
\par
If the move is a terminal shift, the same shift is necessarily present in $I_1$. If it is a reduction, since the candidate identifier chains are identical, the same string is reduced to the same nonterminal. Therefore also the following nonterminal shift operated by $\mathcal{C}$ is legal also for $\mathcal{P}$, which is a contradiction. \qed
\end{proof}
We have thus established that the parser controlled by the compact pilot is equivalent to the original one.
\subsubsection{Candidate identifiers or pointers unnecessary}
Thanks to the $STP$ property, the parser can be simplified to remove the  need for  $cid$'s (or stack pointers). We recall that a $cid$ was needed to find the reach of a non-empty reduction move into the stack: elements were popped until the $cid$ chain reached an initial state, the end-of-list sentinel. Under $STP$ hypothesis, that test is now replaced by a simpler device, to be later incorporated in the final $ELL \, (1)$ parser (Alg. \ref{AlgorTopdownPDA}).
\par
With reference to Alg. \ref{AlgELR1parser}, only shift and reduction moves are modified. First, let us focus on the situation when the old parser, with $J$ on top of stack and $J_{ \vert base} = \langle q_A, \, \pi \rangle$, performs the shift of $X$ (terminal or non-) from state $q_A$, which is necessarily non-initial; the shift would require to compute and record a non-null $cid$ into the $sms$ $J'$ to be pushed on stack. In the same situation, the pointerless parser cancels from the top-of-stack element $J$ all candidates other than $q_A$, since they correspond to discarded parsing alternatives. Notice that the canceled candidates are necessarily in $J_{ \vert closure}$, hence they contain only initial states. Their elimination from the stack allows the parser to uniquely identify the reduction to be made when a final state of machine $M_A$ is entered, by a simple rule: keep popping the stack until the first occurrence of initial state $0_A$ is found.
\par
Second, consider a shift from an initial state $\langle 0_A, \, \pi\rangle$, which is necessarily in $J_{ \vert closure}$. In this case the pointerless parser leaves $J$ unchanged and pushes $\vartheta \, (J, \, X)$ on the stack. The other candidates present in $J$ cannot be canceled because they may be the origin of future nonterminal shifts.
\par
Since $cid$'s are not used by this parser, a stack element is identical to an m-state (of the compact pilot). Thus a shift move first updates the top of stack element, then it pushes the input token and the next m-state. We specify only the moves that differ from Alg. \ref{AlgELR1parser}.
\begin{algorithm}\label{algPointerless}Pointerless parser $\mathcal{A}_{PL}$.
\par
Let the pilot be compacted; m-states are denoted $K_i$ and stack symbols $H_i$; the set of candidates of $H_i$ is weakly included in $K_i$.
\begin{description}
\item[\textit{Shift move}] Let the current character be $a$, $H$ be the top of stack element, containing candidate
$\langle q_A, \, \pi \rangle$. Let ${q_A} \stackrel a \to {q_A}'$ and $\vartheta \, (K, \, a) = K'$ be respectively
the state transition and m-state transition, to be applied. The shift move does:
\begin{enumerate}
\item if $q_A\in K_{ \vert base}$ (i.e., $q_A$ is not initial), eliminate all other candidates from $H$,
i.e., set $H$ equal to $K_{ \vert base}$
\item push $a$ on stack and get next token
\item push $H' = K'$ on the stack
\end{enumerate}
\item[\textit{Reduction move} (non-initial state)] Let the stack be $H[0] \, a_1 \, H[1] \, a_2 \, \ldots \, a_k \, H[k]$.
\par
Assume that the pilot chooses the reduction candidate $ \langle q_A, \, \pi \rangle \in K[k]$, where $q_A$ is a final but non-initial state. Let $H[h]$ be the topmost stack element such that $0_A \in K[h]_{ \vert kernel}$. The move does:
\begin{enumerate}
\item grow the syntax forest by applying the reduction $a_{h+1} \, a_{h+2} \, \ldots \, a_k \leadsto  A$ and pop the stack
symbols $H[k]$, $a_k$, $H[k-1]$, $a_{k-1}$, \ldots, $H[h+1]$, $a_{h+1}$
\item execute the nonterminal shift move  $\vartheta \, (K[h], \, A)$
\end{enumerate}
\item[\textit{Reduction move} (initial state)] It differs from the preceding case in that, for the chosen reduction
candidate $\langle 0_A, \, \pi \rangle$, the state is initial and final. Reduction $\varepsilon \leadsto A$ is applied
to grow the syntax forest. Then the parser performs the nonterminal shift move $\vartheta \, (K[k], \, A)$.
\item[\textit{Nonterminal shift move}] It is the same as a shift move, except that the shifted symbol is a nonterminal.
The only difference is that the parser does not read the next input token at line (2) of \textit{Shift move}.
\end{description}
\end{algorithm}
Clearly this reorganization removes the need of $cid$'s or pointers while preserving correctness.
\begin{property}\label{propPointerlessParser}
If the $ELR \, (1)$ pilot of an $EBNF$ grammar or machine net satisfies the $STP$ condition, the pointerless parser $\mathcal{A}_{PL}$ of Alg. \ref{algPointerless} is equivalent to the $ELR \, (1)$ parser $\mathcal{A}$ of Alg. \ref{AlgELR1parser}.
\end{property}
\begin{proof}
There are two parts to the proof, both straightforward to check. First, after parsing the same string, the stacks of  parsers $\mathcal{A}$ and  $\mathcal{A}_{PL}$ contain the same number $k$ of stack elements, respectively $J[0] \, \ldots \, J[k]$ and $K[0] \, \ldots \, K[k]$, and for every pair of corresponding elements, the set of states included in $K[i]$ is a subset of the set of states included in $J[i]$
because Alg. \ref{algPointerless} may have discarded a few candidates. Furthermore, we claim the next relation for any pair of stack elements at position $i$ and $i-1$.
\begin{gather*}
\text{in} \ J[i] \ \text{candidate} \ \langle q_A, \, \pi, \, \sharp j \rangle \ \text{points to} \ \langle p_A, \, \pi, \, \sharp l \rangle \ \text{in} \ J[i-1]_{ \vert base}
\\
\iff
\\
\langle q_A, \, \pi \rangle \in K[i] \ \text{and the only candidate in} \ K[i-1] \ \text{is} \ \langle p_A, \, \pi \rangle
\end{gather*}
\begin{gather*}
\text{in} \ J[i] \ \text{candidate} \ \langle q_A, \, \pi, \, \sharp j \rangle \ \text{points to} \ \langle 0_A, \, \pi, \, \bot \rangle \ \text{in} \ J[i-1]_{ \vert closure}
\\
\iff
\\
\langle q_A, \, \pi \rangle \in K[i] \ \text{and} \ K[i-1] \ \text{equals the projection of} \ J[i-1] \ \text{on} \ \langle \text{state}, \, \text{look-ahead} \rangle
\end{gather*}
Then by the specification of reduction  moves, $\mathcal{A}$ performs reduction:
\[
a_{h+1} \, a_{h+2} \, \ldots \, a_k \leadsto A \ \iff \ \mathcal{A}_{PL} \ \text{performs the same reduction}
\]
\end{proof}
\begin{example}Pointerless parser trace.
\par
Such a parser is characterized by having in each stack element only one non-initial machine state, plus possibly a few initial ones. Therefore the parser explores only one possible
reduction at a time and candidate pointers are not needed.
\par
Given the input string $( \; ( \; ) \; a \; ) \dashv$, Figure \ref{figExPointerlessParserTrace} shows the execution trace of a pointerless parser for  the same input as in Figure \ref{figExRunningParseTrace} (parser using $cid$'s). The graphical conventions are unchanged: the m-state (of the compact pilot) in each cell is framed, e.g., $K_{0_S}$ is denoted \framebox{$0_S$}, etc.; the final candidates are encircled, e.g., \ovalnode{\relax}{$3_T$}, etc.; and the look-aheads are omitted to avoid clogging. So a candidate appears as a pure machine state. Of course, here there are no pointers;  instead, the initial candidates canceled by Alg. \ref{algPointerless} from a m-state, are striked out. We observe that upon starting a reduction, the initial state of the active machine might in principle show up in a few stack elements to be popped. For instance the first (from the Figure top) reduction $( \; E \; ) \leadsto T$ of machine $M_T$, pops three stack elements, namely $K_{3_T}$, $K_{2_T}$ and $K_{1_T}$, and the one popped last, i.e., $K_{1_T}$, contains a striked out candidate $0_T$ that would be initial for machine $M_T$. But the real initial state for the reduction remains instead unstriked below in the stack in  element $K_{1_T}$, which in fact is not popped and thus is the origin of the shift on $T$ soon after executed.
\begin{figure}[h!]
\begin{center}
\scalebox{0.725}{\def\arraystretch{1.4}
\begin{tabular}{||c||c|c|c|c|c|c||m{67pt}||} \hline\hline
\multirow{2}{*}{\textbf{stack base}} & \multicolumn{6}{|c||}{\textbf{string to be parsed {\rm (with end-marker)} and stack contents}} & \multirow{2}{*}{\textbf{effect after}} \\ \cline{2-7}
& 1 & 2 & 3 & 4 &	5 & &  \\ \hline\hline
\multirow{44}{*}{$\framebox{$0_E$} \ \begin{array}{c} \ovalnode{\relax}{0_E} \\ 0_T \end{array}$} & $($ & $($ & $)$ & $a$ & $)$ & $\dashv$ & \\ \cline{2-6}
& \multicolumn{6}{|c||}{\cellcolor[gray]{0.8} $\begin{array}{c} \\ \\ \\ \end{array}$} & initialisation \par of the stack \\ \cline{2-8}
& $($ & $($ & $)$ & $a$ &	$)$ & $\dashv$ & \\ \cline{2-6}
& $\framebox{$1_T$} \ \begin{array}{c} 1_T \\ \hline \ovalnode{\relax}{0_E} \\ 0_T \end{array}$
& \multicolumn{5}{|c||}{\cellcolor[gray]{0.8}} & shift on $($ \\ \cline{2-8}
& $($ &	$($ & $)$ & $a$ & $)$ & $\dashv$ & \\ \cline{2-6}
& $\framebox{$1_T$} \ \begin{array}{c} 1_T \\ \hline \ovalnode{\relax}{0_E} \\ 0_T \end{array}$
& $\framebox{$1_T$} \ \begin{array}{c} 1_T \\ \hline \ovalnode{\relax}{0_E} \\ 0_T \end{array}$
& \multicolumn{4}{|c||}{\cellcolor[gray]{0.8}} & shift on $($ \\ \cline{2-8}
& $($ &	$($ \qquad \qquad $E$ & $)$ & $a$ & $)$ & $\dashv$ & \\ \cline{2-6}
& $\framebox{$1_T$} \ \begin{array}{c} 1_T \\ \hline \ovalnode{\relax}{0_E} \\ 0_T \end{array}$	
& $\framebox{$1_T$} \ \begin{array}{c} 1_T \\ \hline \ovalnode{\relax}{0_E} \\ 0_T \end{array}$ & \multicolumn{4}{|c||}{\cellcolor[gray]{0.8}} & reduction $\varepsilon \leadsto E$ \\ \cline{2-8}
& $($ &	$($ \qquad \qquad $E$ & $)$ & $a$ & $)$ & $\dashv$ & \\ \cline{2-6}
& $\framebox{$1_T$} \ \begin{array}{c} 1_T \\ \hline \ovalnode{\relax}{0_E} \\ 0_T \end{array}$ & $\framebox{$1_T$} \ \begin{array}{c} 1_T \\ \hline \msout{0_E} \\ \msout{0_T} \end{array}$ \ \vline \ \ \ $\framebox{$2_T$} \ \begin{array}{c} 2_T \end{array}$ & $\framebox{$3_T$} \ \begin{array}{c} \ovalnode{\relax}{3_T} \end{array}$ & \multicolumn{3}{|c||}{\cellcolor[gray]{0.8}} & shifts on $E$ and $)$ \\ \cline{2-8}
& $($ &	\multicolumn{2}{|c|}{$T$} & $a$ & $)$ & $\dashv$ & \\ \cline{2-6}
& $\framebox{$1_T$} \ \begin{array}{c} 1_T \\ \hline \ovalnode{\relax}{0_E} \\ 0_T \end{array}$ & \multicolumn{2}{|c|}{$\framebox{$1_E$} \ \begin{array}{c} \ovalnode{\relax}{1_E} \\ \hline 0_T \\ \end{array}$} & \multicolumn{3}{|c||}{\cellcolor[gray]{0.8}} & reduction $(E) \leadsto T$ \par and shift on $T$ \\ \cline{2-8}
& $($ &	\multicolumn{2}{|c|}{$T$} & $a$ & $)$ & $\dashv$ & \\ \cline{2-6}
& $\framebox{$1_T$} \ \begin{array}{c} 1_T \\ \hline \ovalnode{\relax}{0_E} \\ 0_T \end{array}$ & \multicolumn{2}{|c|}{$\framebox{$1_E$} \ \begin{array}{c} \ovalnode{\relax}{1_E} \\ \hline 0_T \\ \end{array}$} & $\framebox{$3_T$} \ \begin{array}{c} \ovalnode{\relax}{3_T} \end{array}$ & \multicolumn{2}{|c||}{\cellcolor[gray]{0.8}} & shift on $a$ \\ \cline{2-8}
& $($ &	\multicolumn{2}{|c|}{$T$} & $T$ & $)$ & $\dashv$ & \\ \cline{2-6}
& $\framebox{$1_T$} \ \begin{array}{c} 1_T \\ \hline \ovalnode{\relax}{0_E} \\ 0_T \end{array}$ & \multicolumn{2}{|c|}{$\framebox{$1_E$} \ \begin{array}{c} \ovalnode{\relax}{1_E} \\ \hline \msout{0_T} \\ \end{array}$} & $\framebox{$1_E$} \ \begin{array}{c} \ovalnode{\relax}{1_E} \\ \hline 0_T \\ \end{array}$ & \multicolumn{2}{|c||}{\cellcolor[gray]{0.8}} & reduction $a \leadsto T$ \par and shift on $T$ \\ \cline{2-8}
& $($ &	\multicolumn{3}{|c|}{$E$} & $)$ & $\dashv$ & \\ \cline{2-6}
& $\framebox{$1_T$} \ \ \begin{array}{c} 1_T \\ \hline \msout{0_E} \\ \msout{0_T} \end{array}$ & \multicolumn{3}{|c|}{$\framebox{$2_7$} \ \begin{array}{c} 2_T \\ \end{array}$} & $\framebox{$3_7$} \ \begin{array}{c} \ovalnode{\relax}{3_T} \\ \end{array}$ & \cellcolor[gray]{0.8} & reduction $T T \leadsto E$ \par and shifts on $E$ and $)$ \\ \cline{2-8}
& \multicolumn{5}{|c|}{$T$} & $\dashv$ & \\ \cline{2-6}
& \multicolumn{5}{|c|}{$\framebox{$1_E$} \ \def\arraystretch{2.0} \begin{array}{c} \ovalnode{\relax}{1_E} \\ \hline 0_T \end{array}$} & \cellcolor[gray]{0.8} & reduction $( E ) \leadsto T$ \par and shift on $T$ \\ \cline{2-8}
& \multicolumn{5}{|c|}{$E$} & $\dashv$ & \\ \cline{2-6}
& \multicolumn{6}{|c||}{\cellcolor[gray]{0.8} $\begin{array}{c} \\ \\ \\ \end{array}$} & reduction $T \leadsto E$ \par and accept without \par shitting on $E$ \\ \hline\hline
\end{tabular}
}
\end{center}
\caption{Steps of the pointerless parsing algorithm $\mathcal{A_{PL}}$; the candidates are canceled by the shift moves as explained in Algorithm \ref{algPointerless}.}
\label{figExPointerlessParserTrace}
\end{figure}
\par
 We also point out that Alg. \ref{algPointerless} \emph{does not} cancel any initial candidates from a stack element, if a shift move is executed from
one such candidate (see the \emph{Shift move} case of Alg. \ref{algPointerless}). The motivation for not canceling is twofold. First, these candidates will not cause any early stop of the series of pop moves in the reductions that may come later, or said differently, they will not break any reduction handle. For instance the first (from the Figure top) shift on $($ of machine $M_T$, keeps both initial candidates $0_E$ and $0_T$ in the stack m-state $K_{1_T}$, as the shift originates from the initial candidate $0_T$. This candidate will instead be canceled (i.e., it will show striked) when a shift on $E$ is executed (soon after reduction $T \; T \leadsto E$), as the shift originates from the non-initial candidate $1_T$. Second, some of such initial candidates may be needed for a subsequent nonterminal shift move.
\end{example}
To sum up, we have shown that condition  $STP$ permits to construct a simpler shift-reduce parser, which has a reduced stack alphabet and does not need pointers to manage reductions. The same parser can be further simplified, if we make another hypothesis: that the grammar is not left-recursive.
\subsection{$\emph{ELL} \, (1)$ condition}\label{subsectELL(1)condition}
Our definition of top-down deterministically parsable grammar or network comes next.
\begin{definition}\label{defELL(1)condition}$ELL \, (1)$ condition.\footnote{The historical acronym ``$ELL \, (1)$'' has been introduced over and over in the past by several authors with slight differences. We hope that reusing again the acronym will not be considered an abuse.}
\par
A machine net $\mathcal{M}$ meets the $ELL(1)$ condition if the following three clauses are satisfied:
\begin{enumerate}
\item there are no left-recursive derivations
\item the net meets the $ELR(1)$ condition
\item  the net has the  single transition property ($STP$)
\end{enumerate}
\end{definition}
Condition (1) can be easily checked by drawing a graph, denoted by $\mathcal{G}$, that has as nodes the initial states of the net and has an edge $0_A \longrightarrow 0_B$ if in machine $M_A$ there exists an edge $0_A \stackrel B \longrightarrow r_A$, or more generally a path:
\[
0_A \stackrel {A_1} \longrightarrow q_1 \stackrel {A_2} \longrightarrow \ldots \stackrel {A_k} \longrightarrow q_k \stackrel B \longrightarrow r_A \qquad k \geq 0
\]
where all the nonterminals $A_1$, $A_2$, $\ldots$, $A_k$ are nullable. The net is left-recursive if, and only if, graph $\mathcal{G}$ contains a circuit.
\par
Actually, left recursive derivations cause, in all but one situations, a violation of clause (2) or (3) of Def. \ref{defELL(1)condition}. The only case that would remain undetected is a left-recursive derivation involving the axiom $0_S$ and caused by state-transitions of the form:
\begin{equation}\label{eqSignificantLeftRecurs}
0_S \stackrel {S} \longrightarrow 1_S \qquad \text{or} \qquad 0_S \stackrel {A_1} \longrightarrow1_S \quad 0_{A_1} \stackrel {A_2} \longrightarrow 1_{A_1} \quad \ldots \quad 0_{A_k} \stackrel {S} \longrightarrow 1_{A_k} \qquad k \geq 1
\end{equation}
Three different types of left-recursive derivations are illustrated in Figure \ref{FigLeftRecDerivNoELL1}. The first two cause violations of clause (2) or (3), so that only the third needs to be checked on graph  $\mathcal{G}$: the graph contains a self-loop on node $0_E$.
\begin{figure}[h!]
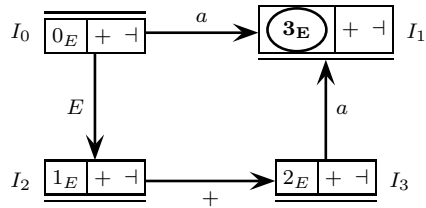

\vspace{0.25cm}
\begin{center}
\begin{tabular}{cp{1.0cm}c}
$E \to X \; E \; + \; a \; \mid \; a$ && $X \to \varepsilon \; \mid \; b$ \\

\pspar\psset{border=0pt,nodesep=0pt,labelsep=5pt,colsep=1.125cm,rowsep=1.5cm}
\scalebox{0.8}{
\begin{psmatrix}

\\ \ovalnode{0E}{$0_E$} & \ovalnode{4E}{$4_E$} & \ovalnode{1E}{$1_E$} & \ovalnode{2E}{$2_E$} & \ovalnode{3E}{$3_E$}

\nput[labelsep=0pt]{180}{0E}{$M_E \ \to$}

\nput[labelsep=0pt]{0}{3E}{$\to$}

\ncarc{0E}{3E} \naput{$a$}

\ncline{0E}{4E} \nbput{$X$}

\ncline{4E}{1E} \nbput{$E$}

\ncline{1E}{2E} \nbput{$+$}

\ncline{2E}{3E} \nbput{$a$}

\end{psmatrix}}

&&

\pspar\psset{border=0pt,nodesep=0pt,labelsep=5pt,colsep=1.125cm}
\scalebox{0.8}{
\begin{psmatrix}

\ovalnode{0X}{$0_X$} & \ovalnode{1X}{$1_X$}

\nput[labelsep=0pt]{180}{0X}{$M_X \ \to$}

\nput[labelsep=0pt]{0}{1X}{$\to$}

\nput[labelsep=0pt]{-90}{0X}{$\downarrow$}

\ncline{0X}{1X} \naput{$b$}

\end{psmatrix}}
\end{tabular}
\vspace{0.25cm}
\end{center}
\begin{center}
a shift-reduce conflict in m-state $I_0$ caused by a left-recursive derivation that makes use of $\varepsilon$-rules
\end{center}
\begin{center}
\vspace{0.25cm}
\scalebox{1.0}{
\pspar\psset{border=0pt,nodesep=0pt,rowsep=0.25cm}
\begin{psmatrix}

\rnode{I0}{$\def\arraystretch{1.25}\begin{array}{|c|c|} \hline\hline
0_E & \dashv \\ \cline{2-2}
\ovalnode{\relax}{\mathbf{0_X}} & a \; b \\ \hline
\end{array}$}

&

\rnode{I1}{$\def\arraystretch{1.25}\begin{array}{|c|c|} \hline
\ovalnode{\relax}{\mathbf{3_E}} & \dashv \\ \hline\hline
\end{array}$} \\

\nput[labelsep=5pt]{-180}{I0}{$I_0$}

\nput[labelsep=5pt]{0}{I1}{$I_1$}

\ncline{I0}{I1} \naput{$a$}

\end{psmatrix}}
\end{center}
\hrule
\begin{center}
\begin{tabular}{cp{1.0cm}c}
$S \to a \; A$ && $A \to A \; b \; \mid \; b$ \\

\pspar\psset{border=0pt,nodesep=0pt,labelsep=5pt,rowsep=1.25cm,colsep=1cm}
\scalebox{0.8}{
\begin{psmatrix}

\\ \ovalnode{0S}{$0_S$} & \ovalnode{1S}{$1_S$} & \ovalnode{2S}{$2_S$}

\nput[labelsep=0pt]{180}{0S}{$M_S \ \to$}

\nput[labelsep=0pt]{0}{2S}{$\to$}

\ncline{0S}{1S} \naput{$a$}

\ncline{1S}{2S} \naput{$A$}

\end{psmatrix}}

&&

\pspar\psset{border=0pt,nodesep=0pt,labelsep=5pt,rowsep=1.25cm,colsep=1cm}
\scalebox{0.8}{
\begin{psmatrix}

\\ \ovalnode{0A}{$0_A$} & \ovalnode{1A}{$1_A$}  & \ovalnode{2A}{$2_A$}

\nput[labelsep=0pt]{180}{0A}{$M_A \ \to$}

\nput[labelsep=0pt]{0}{2A}{$\to$}

\ncarc[arcangle=45]{0A}{2A} \naput{$b$}

\ncline{0A}{1A} \nbput{$A$}

\ncline{1A}{2A} \nbput{$b$}

\end{psmatrix}}
\end{tabular}
\vspace{0.25cm}
\end{center}
\begin{center}
a left-recursive derivation through a nonterminal other than the axiom has the effect \par of creating two candidates in the base of m-state $I_2$ and thus violates clause (2)
\end{center}
\begin{center}
\vspace{0.5cm}
\pspar\psset{border=0pt,nodesep=0pt,rowsep=1.0cm}
\scalebox{1.0}{
\begin{psmatrix}

\rnode{I0}{$\def\arraystretch{1.25}\begin{array}{|c|c|} \hline\hline
0_S& \dashv \\ \hline
\end{array}$}

&

\rnode{I1}{$\def\arraystretch{1.25}
\begin{array}{|c|c|} \hline
1_S & \dashv \\ \hline\hline
0_A & \dashv \\ \hline
\end{array}$}

&

\rnode{I2}{$\def\arraystretch{1.25}
\begin{array}{|c|c|} \hline
\ovalnode{\relax}{\mathbf{2_S}} & \multirow{2}{*}{$\dashv$} \\
1_A & \\ \hline\hline
\end{array}$}

\nput[labelsep=5pt]{90}{I0}{$I_0$}

\nput[labelsep=5pt]{90}{I1}{$I_1$}

\nput[labelsep=5pt]{90}{I2}{$I_2$}

\ncline{I0}{I1}\naput{a}

\ncline{I1}{I2}\naput{A}

\end{psmatrix}}
\end{center}
\hrule
\begin{center}
\vspace{0.25cm}
\begin{tabular}{cp{1.0cm}c}
$E \to E \; + \; a \; \mid \; a$

&&

\pspar\psset{border=0pt,nodesep=0pt,colsep=1.0cm,labelsep=5pt,rowsep=1.0cm}
\scalebox{0.8}{
\begin{psmatrix}

\\ \ovalnode{0E}{$0_E$} & \ovalnode{1E}{$1_E$} & \ovalnode{2E}{$2_E$} & \ovalnode{3E}{$3_E$}

\nput[labelsep=0pt]{180}{0E}{$M_E \ \to$}

\nput[labelsep=0pt]{0}{3E}{$\to$}

\ncarc{0E}{3E} \naput{$a$}

\ncline{0E}{1E} \nbput{$E$}

\ncline{1E}{2E} \nbput{$+$}

\ncline{2E}{3E} \nbput{$a$}

\end{psmatrix}}
\end{tabular}
\vspace{0.25cm}
\end{center}
\begin{center}
clauses (1) and (2) are met and the left-recursive derivation $0_E \Rightarrow 0_E \, 1_E$ is entirely contained in m-state $I_0$
\end{center}
\begin{center}
\vspace{0.25cm}
\pspar\psset{border=0pt,nodesep=0pt,rowsep=1.3cm}
\scalebox{1.0}{
\begin{psmatrix}
\rnode{I0}{$\def\arraystretch{1.25}\begin{array}{|c|c|} \hline\hline
0_E & + \; \dashv \\ \hline
\end{array}$} &

\rnode{I1}{$\def\arraystretch{1.25}\begin{array}{|c|c|} \hline
\ovalnode{\relax}{\mathbf{3_E}} & + \; \dashv \\
\hline\hline
\end{array}$} \\

\rnode{I2}{$\def\arraystretch{1.25}\begin{array}{|c|c|} \hline
1_E & + \; \dashv \\ \hline\hline
\end{array}$}

&

\rnode{I3}{$\def\arraystretch{1.25}\begin{array}{|c|c|} \hline
2_E & + \; \dashv \\ \hline \hline
\end{array}$}

\nput[labelsep=5pt]{180}{I0}{$I_0$}

\nput[labelsep=5pt]{0}{I1}{$I_1$}

\nput[labelsep=5pt]{180}{I2}{$I_2$}

\nput[labelsep=5pt]{0}{I3}{$I_3$}

\ncline{I0}{I1} \naput{$a$}

\ncline{I0}{I2} \nbput{$E$}

\ncline{I2}{I3} \nbput{$+$}

\ncline{I3}{I1} \nbput{$a$}

\end{psmatrix}}
\end{center}
\caption{Left-recursive nets violating $ELL \, (1)$ conditions. Top: left-recursion with $\varepsilon$-rules violates clause (2). Middle: left-recursion on $A \neq \text{axiom}$ violates clause (3). Bottom: left-recursive axiom $E$ is undetected by clauses (2) and (3).}\label{FigLeftRecDerivNoELL1}
\end{figure}
\par
It would not be difficult to formalize the properties illustrated by the examples and to restate clause (1) of Definition \ref{defELL(1)condition} as follows: the net has no left-recursive derivation of the form \eqref{eqSignificantLeftRecurs}, i.e., involving  the axiom and not using $\varepsilon$-rules. This apparently weaker but indeed equivalent condition is stated by \cite{Beatty82}, in his definition of (non-extended) $LL \, (1)$ grammars.
\subsubsection{Discussion}\label{emptyNonTerm}
To sum up, $ELL \, (1)$ grammars, as defined here, are $ELR \, (1)$ grammars that do not allow left-recursive derivations and satisfy the single-transition property (no multiple candidates in m-state bases and hence no convergent transitions). This is a more precise reformulation of  the definitions of ``$LL \, (1)$'' and ``$ELL \, (1)$'' grammars, which have accumulated in half a century. To be fair to the perhaps most popular definition of $LL(1)$ grammar, we contrast it with ours. There are marginal contrived examples where a violation of $STP$ caused by the presence of multiple candidates in the base of a m-state, does not hinder a top-down parser from working deterministically \cite{Beatty82}. A typical case is the $LR \, (1)$ grammar $\set{ \; S \to A \, a \, \mid \, B \, b, \; A \to C, \; B \to C, \; C \to \varepsilon \, }$. One of the m-states has two candidates in the base: $\big\{ \, \langle A \to C \; \bullet, \; a \rangle, \, \langle B \to C \; \bullet, \, b \rangle \; \big\}$. The choice between the alternatives of $S$ is determined by the following character, $a$ or $b$, yet the grammar violates $STP$. It easy to see that a necessary condition for such a situation to occur is that the language derived from some nonterminal of the grammar in question consists only of the empty string: $\exists \, A \in V$  such that $L_A \, (G) = \set{ \; \epsilon \; }$. However this case can be usually removed without any penalty, nor a loss of generality, in all grammar applications.
\par
In fact, the grammar above can be simplified and the equivalent grammar
\[
\set{ \; S \to A \, \mid \, B, \; A \to a, \; B \to b \; }
\]
is obtained, which complies with $STP$.
\subsection{Stack contraction and predictive parser}\label{sectParserCFG}
The last development, to be next presented, transforms the already compacted pilot graph into the Control Flow Graph of a \emph{predictive} parser. The way the latter parser uses the stack differs from the previous models, the shift-reduce and pointerless parsers. More precisely, now a terminal shift move, which always executes a push operation, is sometimes implemented without a push and sometimes with multiple pushes. The former case happens when the shift remains inside the same machine: the predictive parser does not \emph{push} an element upon performing a terminal shift, but it updates the top of stack element to record the new state. Multiple pushes happen when the shift determines one or more transfers from the current machine to others; the predictive parser performs a push for each transfer.
\par
The essential information to be kept on stack, is the sequence of machines that have been activated and have not reached a final state (where a reduction occurs). At each parsing time, the current or \emph{active} machine is the one that is doing the analysis, and the \emph{current} state is kept in the top of stack element. Previous non-terminated activations of the same or other machines are in the \emph{suspended} state. For each suspended machine $M_A$, a stack entry is needed to store the state $q_A$, from where the machine will resume the computation when control is returned after performing the relevant reductions.
\par
The main advantage of predictive parsing is that the construction of the syntax tree can be anticipated: the parser can generate on-line the left derivation of the input.
\subsubsection{Parser control-flow graph}
Moving from the above considerations, first we slightly transform the compact pilot graph $\mathcal{C}$ and make it isomorphic to the original machine net $\mathcal{M}$. The new graph is named \emph{parser control-flow graph} ($PCFG$) because it represents the blueprint of parser code. The first step of the transformation splits every m-node of $\mathcal{C}$ that contains multiple candidates, into a few nodes that contain only one candidate. Second,  the kernel-equivalent nodes are coalesced and the original look-ahead sets are combined into one. The third step creates new edges, named \emph{call edges}, whenever a machine transfers control to another machine. At last, each call edge is labeled with a set of characters, named \emph{guide set}, which is a summary of the information needed for the parsing decision to transfer control to another machine.
\begin{definition}\label{defPilotCFG}Parser Control-Flow Graph.
\par
Every \emph{node} of the $PCFG$, denoted by $\mathcal{F}$, is identified by a state $q$ of machine net $\mathcal{M}$ and denoted, without ambiguity, by $q$. Moreover, every node $q_A$, where $q_A \in F_A$ is  final, consists of a pair $\langle q_A, \, \pi \rangle$, where set $\pi$, named \emph{prospect}\footnote{Although traditionally the same word ``look-ahead''  has been used for both shift-reduce and top-down parsers, the set definitions differ and we prefer to differentiate their names.} set, is the union of the look-ahead sets $\pi_i$ of every candidate $\langle q_A, \, \pi_i \rangle$ existing in the compact pilot graph $\mathcal{C}$:
\[
\pi = \bigcup_{\forall \; \langle q_A, \, \pi_i \rangle \, \in \, \mathcal{C}} \pi_i
\]
The \emph{edges} of $\mathcal{F}$ are of two types, named \emph{shift} and \emph{call}:
\begin{enumerate}
\item There exists in $\mathcal{F}$ a shift edge $q_A \stackrel X \to r_A$ with $X$ terminal or non-, if the same edge is in machine $M_A$.
\item
There exists in $\mathcal{F}$ a call edge $q_A \stackrel {\gamma_1} \dashrightarrow 0_{A_1}$, where $A_1$ is a nonterminal
possibly different from $A$, if $q_A \stackrel {A_1} \to r_A$ is in $M_A$, hence necessarily in some m-state $K$ of $\mathcal{C}$ there exist candidates $\langle q_A, \, \pi \rangle$ and $\langle 0_{A_1}, \, \rho \rangle$; and the m-state $\vartheta \, (K, \, A_1)$ contains candidate $\langle r_A, \, \pi_{r_A} \rangle $. The call edge label $\gamma_1 \subseteq \Sigma \; \cup \; \{ \; \dashv \; \}$, named \emph{guide set},\footnote{It is the same as a ``predictive parsing table element'' in \cite{AhoLamSethiUl2006}.} is recursively defined as follows:
 \begin{itemize}
\item
it holds $b \in \gamma_1$ if, and only if, any conditions below hold:
{\def\arraystretch{1.5}\begin{align}
\label{eqGuideSet1}
& b \in Ini \, \big(L \, (0_{A_1}) \big) \\
\label{eqGuideSet2}
& A_1 \ \text{is nullable} \ \text{and} \ b \in Ini \, \big( L \, (r_A) \big) \\
\label{eqGuideSet3}
& A_1 \, and \, L \, (r_A) \ \text{are both nullable} \ \text{and} \ b \in \pi_{r_A} \\
\label{eqGuideSet4}
& \exists \ \text{in} \ \mathcal{F} \ \text{a call edge} \ 0_{A_1} \stackrel {\gamma_2} \dashrightarrow 0_{A_2} \ \text{and} \ b \in \gamma_2
\end{align}}
\end{itemize}
\end{enumerate}
\end{definition}
Relations	\eqref{eqGuideSet1}, \eqref{eqGuideSet2}, \eqref{eqGuideSet3} are not recursive and respectively consider that $b$ is generated by $M_{A_1}$ called by $M_A$; or by $M_A$ but starting from state $r_A$; or that $b$ follows $M_A$. Rel. \eqref{eqGuideSet4} is recursive and traverses the net as far as the chain of call sites activated. We observe that Rel.	 \eqref{eqGuideSet4} determines an inclusion relation $\gamma_1 \supseteq \gamma_2$ between any two concatenated call edges $q_A \stackrel {\gamma_1} \dashrightarrow 0_{A_1} \stackrel {\gamma_2} \dashrightarrow 0_{A_2}$. We also write $\gamma_1 = Gui \, (q_A \dashrightarrow 0_{A_1})$ instead of $q_A \stackrel {\gamma_1} \dashrightarrow 0_{A_1}$.
\par
We next extend the definition of guide set to the terminal shift edges and to the dangling darts that tag the final nodes. For each terminal shift edge $p \stackrel a \to q$ labeled with $a \in \Sigma$, we set $Gui \, \left( p \stackrel a \to q \right) := \set{ \, a \, }$. For each  dart that tags a final node containing a candidate $\langle f_A, \, \pi \rangle$ with $f_A \in F_A$, we set $Gui \, (f_A \to ) := \pi$. Then in the $PCFG$ all edges (except the nonterminal shifts) can be interpreted as conditional instructions, enabled if the current character $cc$ belongs to the associated guide set: a terminal shift edge labeled with $\set{ \, a \, }$ is enabled by a predicate $cc = a$ (or $cc \in \set{ \, a \, }$ for uniformity); a call edge labeled with $\gamma$ represents a conditional procedure invocation, where the enabling predicate is $cc \in \gamma$; a final node dart labeled with $\pi$ is interpreted as a conditional return-from-procedure instruction to be executed if $cc \in \pi$. The remaining $PCFG$ edges are nonterminal shifts, which are interpreted as unconditional return-from-procedure instructions.
\par
We show that for the same state such predicates never conflict with one another.
\begin{property}\label{PropDisjointGuideSets}
For every node $q$ of the $PCFG$ of a grammar satisfying the $ELL \, (1)$ condition, the guide sets of any two edges originating from $q$ are disjoint.
\end{property}
\begin{proof}
Since every machine is deterministic, identically labeled shift edges cannot originate from the same node, and it remains to consider the cases of
shift-call and call-call edge pairs.
\par
Consider a shift edge $q_A \stackrel a \to r_A$ with $a \in \Sigma$, and two call edges $q_A \stackrel {\gamma_B} \dashrightarrow 0_B$ and $q_A \stackrel {\gamma_C} \dashrightarrow 0_C$. First, assume by contradiction $a \in \gamma_B$. If $a \in \gamma_B$ comes from Rel. \eqref{eqGuideSet1}, then in the pilot m-state $\vartheta \, (q_A, \, a)$ there are two base candidates, a condition that is ruled out by $STP$. If it comes from Rel. \eqref{eqGuideSet2}, in the m-state with base $q_A$, i.e., m-state $K$, there is a conflict between the shift of $a$ and the reduction $0_B$, which owes its look-ahead $a$ to a path from $r_A$ in machine $M_A$. If it comes from Rel. \eqref{eqGuideSet3}, in m-state $K$ there is the same shift-reduce conflict as before, though now reduction $0_B$ owes its-look ahead $a$ to some other machine that previously invoked machine $M_A$. Finally, it may come from Rel. \eqref{eqGuideSet4}, which is the recursive case and defers the three cases before to some other machine that is immediately invoked by machine $M_B$: there is either a violation of $STP$ or a shift-reduce conflict.
\par
Second, assume by contradiction $a \in \gamma_B$ and $a \in \gamma_C$. Since $a$ comes from one of four relations, twelve combinations should be examined. But since they are similar to the previous argumentation, we deal only with one, namely case \eqref{eqGuideSet1}-\eqref{eqGuideSet1}: clearly, m-state $\vartheta \, (q_A, \, a)$ violates $STP$.
\end{proof}
\par
The converse of Property \ref{PropDisjointGuideSets} also holds, which makes the condition of having disjoint guide sets
a characteristic property of $ELL \, (1)$ grammars. We will see that this condition can be checked easily on the $PCFG$,
with no need to build the $ELR \, (1)$ pilot automaton.
\begin{property}\label{PropDisjointGuideSetsConv}
If the guide sets of a $PCFG$ are disjoint, then the net satisfies the $ELL \, (1)$ condition of Definition \ref{defELL(1)condition}.
\end{property}
The proof is in the Appendix.
\begin{example}\label{esControlFlowGraphLL(1)}
For the running example, the $PCFG$ is represented in Figure \ref{FigELL1CFG} with the same layout as the machine net for comparability.
\begin{figure}[h!]
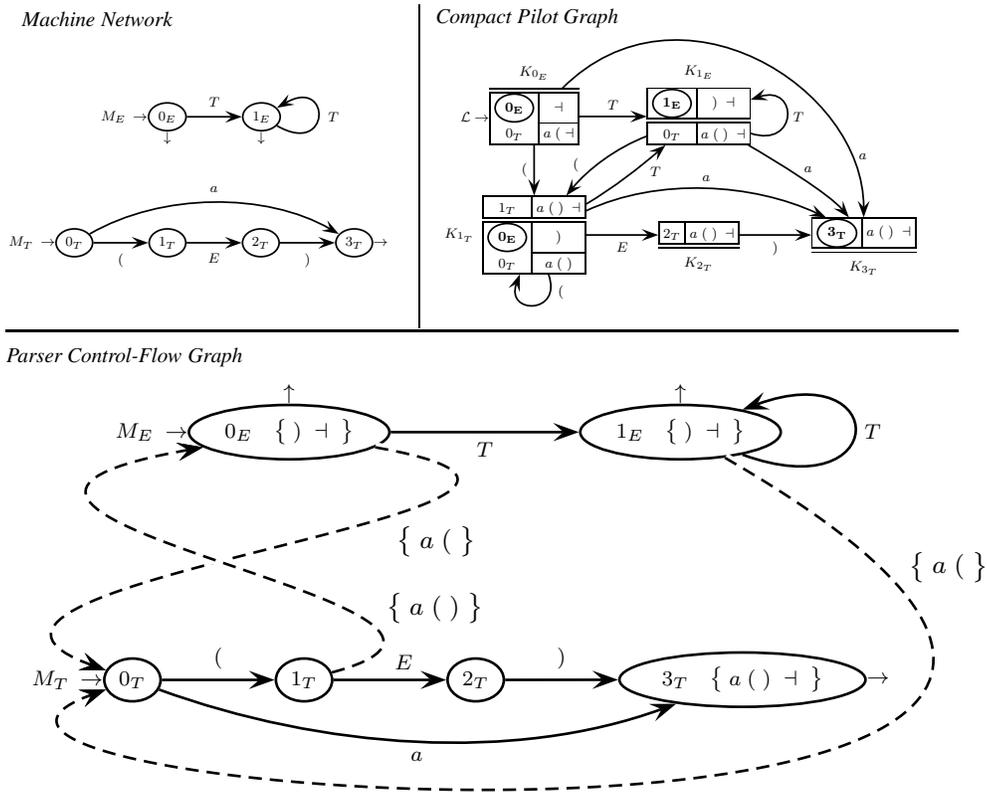

\begin{center}
\begin{tabular}{c|c}
\begin{minipage}{0.40\textwidth}
\begin{flushleft}
\emph{Machine Network}
\end{flushleft}
\end{minipage}

&

\begin{minipage}{0.55\textwidth}
\begin{flushleft}
\emph{Compact Pilot Graph}
\end{flushleft}
\end{minipage}

\\

\begin{minipage}{0.40\textwidth}
\begin{center}

\scalebox{0.65}{
\pspar\psset{border=0pt,nodesep=0pt,labelsep=5pt,colsep=1.125cm}
\begin{psmatrix}
\ovalnode{0E}{$0_E$} & \ovalnode{1E}{$1_E$}

\nput[labelsep=0pt]{180}{0E}{$M_E \ \to$}

\nput[labelsep=0pt]{-90}{1E}{$\downarrow$}

\nput[labelsep=0pt]{-90}{0E}{$\downarrow$}

\ncline{0E}{1E} \naput{$T$}

\nccurve[angleA=-30,angleB=30,ncurvA=7,ncurvB=7]{1E}{1E} \nbput{$T$}
\end{psmatrix}}
\par
\vspace{1.25cm}
\par
\scalebox{0.65}{
\pspar\psset{border=0pt,nodesep=0pt,labelsep=5pt,,colsep=1.125cm}
\begin{psmatrix}
\ovalnode{0T}{$0_T$} & \ovalnode{1T}{$1_T$} & \ovalnode{2T}{$2_T$} & \ovalnode{3T}{$3_T$}

\nput[labelsep=0pt]{180}{0T}{$M_T \ \to$}

\nput[labelsep=0pt]{0}{3T}{$\to$}

\ncline{0T}{1T} \nbput{$($}

\ncline{1T}{2T} \nbput{$E$}

\ncline{2T}{3T} \nbput{$)$}

\ncarc{0T}{3T} \naput{$a$}
\end{psmatrix}}
\end{center}
\end{minipage}

&

\begin{minipage}{0.55\textwidth}
\vspace{0.75cm}
\begin{center}
\pspar\psset{arrows=->,border=0pt,nodesep=0pt,rowsep=1.0cm,colsep=1.25cm}
\scalebox{0.65}{
\begin{psmatrix}
\rnode{I0}{$\def\arraystretch{1.25}\begin{array}{|c|c|} \hline\hline
\ovalnode{o}{\mathbf{0_E}} & \dashv \\ \cline{2-2}
0_T & a \; ( \; \dashv \\ \hline
\end{array}$}

&

\rnode{I14}{$\def\arraystretch{1.25}\begin{array}{|c|c|} \hline
\ovalnode{o}{\mathbf{1_E}} & ) \;\dashv \\ \hline\hline
0_T & a \; ( \; ) \; \dashv \\ \hline
\end{array}$} \\

\rnode{I36}{$\def\arraystretch{1.25}\begin{array}{|c|c|} \hline
1_T & a \; ( \; ) \;\dashv \\ \hline\hline
\ovalnode{o}{\mathbf{0_E}} & ) \\ \cline{2-2}
0_T & a \; ( \; ) \\ \hline
\end{array}$}

&

\rnode{I78}{$\def\arraystretch{1.25}\begin{array}{|c|c|} \hline
2_T & a \; ( \; ) \; \dashv \\ \hline\hline
\end{array}$}

&

\rnode{I25}{$\def\arraystretch{1.25}\begin{array}{|c|c|} \hline
\ovalnode{o}{\mathbf{3_T}} & a \; ( \; ) \; \dashv \\ \hline\hline
\end{array}$}

\nput[labelsep=5pt]{90}{I0}{$K_{0_E}$}

\nput[labelsep=5pt]{90}{I14}{$K_{1_E}$}

\nput[labelsep=5pt]{-90}{I25}{$K_{3_T}$}

\nput[labelsep=5pt]{180}{I36}{$K_{1_T}$}

\nput[labelsep=5pt]{-90}{I78}{$K_{2_T}$}

\nput[labelsep=0pt]{180}{I0}{$\mathcal{L} \to$}

\ncline{I0}{I14} \naput{$T$}

\nccurve[angleA=45,angleB=90,ncurvA=1,ncurvB=0.8]{I0}{I25} \aput(0.85){$a$}

\nccurve[angleA=-30,angleB=130,ncurvA=0.5,ncurvB=0.5]{I14}{I25} \naput{$a$}

\nccurve[angleA=-18,angleB=18,ncurvA=3,ncurvB=3]{I14}{I14} \nbput{$T$}

\ncline{I0}{I36} \nbput{$($}

\ncarc[arcangle=-5]{I36}{I14} \bput(0.75){$T$}

\ncarc[arcangle=-15]{I14}{I36} \bput(0.75){$($}

\ncline{I36}{I78} \nbput{$E$}

\ncline{I78}{I25} \nbput{$)$}

\nccurve[angleA=-70,angleB=-110,ncurvA=3,ncurvB=3]{I36}{I36} \aput(0.15){$($}

\ncarc[arcangle=25]{I36}{I25} \naput{$a$}

\end{psmatrix}}
\end{center}
\vspace{0.5cm}
\end{minipage}
\end{tabular}
\vspace{0.0cm}
\hrule
\begin{flushleft}
\emph{Parser Control-Flow Graph}
\end{flushleft}
\vspace{0.25cm}
\pspar\psset{border=0pt,nodesep=0pt,rowsep=2.0cm,colsep=2.5cm}
\scalebox{1.0}{
\begin{psmatrix}

\ovalnode{0E}{$0_E \quad \big\{ \; ) \; \dashv \; \big\}$} & \ovalnode{1E}{$1_E \quad \big\{ \; ) \; \dashv \; \big\}$}

\nput[labelsep=0pt]{180}{0E}{$M_E \ \to$}

\nput[labelsep=0pt]{90}{0E}{$\uparrow$}

\nput[labelsep=0pt]{90}{1E}{$\uparrow$}

\ncline{0E}{1E} \nbput{$T$}

\nccurve[angleA=-20,angleB=20,ncurvA=7,ncurvB=7]{1E}{1E} \nbput{$T$}

\end{psmatrix}}
\par
\vspace{2.5cm}
\par
\pspar\psset{border=0pt,nodesep=0pt,rowsep=2.0cm,colsep=1.5cm}
\scalebox{1.0}{
\begin{psmatrix}

\ovalnode{0T}{$0_T$} & \ovalnode{1T}{$1_T$} & \ovalnode{2T}{$2_T$} & \ovalnode{3T}{$3_T \quad \big\{ \; a \; ( \; ) \; \dashv \; \big\}$}

\nput[labelsep=0pt]{180}{0T}{$M_T \ \to$}

\nput[labelsep=0pt]{0}{3T}{$\to$}

\ncline{0T}{1T} \naput{$($}

\ncline{1T}{2T} \aput(0.625){$E$}

\ncline{2T}{3T} \naput{$)$}

\ncarc[arcangle=-20]{0T}{3T} \nbput{$a$}

\end{psmatrix}}

\psset{border=0.05cm,linestyle=dashed}

\nccurve[angleA=-10,angleB=160,ncurvA=2,ncurvB=1.5]{0E}{0T} \aput(0.315){\scalebox{1.2}{$\big\{ \; a \; ( \; \big\}$}}

\nccurve[angleA=-30,angleB=-160,ncurvA=2.75,ncurvB=1.0]{1E}{0T}\aput(0.11){\scalebox{1.2}{$\big\{ \; a \; ( \; \big\}$}}

\nccurve[angleA=15,angleB=-170,ncurvA=2,ncurvB=3]{1T}{0E} \bput(0.2){\scalebox{1.2}{$\big\{ \; a \; ( \; ) \; \big\}$}}

\vspace{0.5cm}
\end{center}
\caption{The Parser Control-Flow Graph $\mathcal{F}$ of the running example, from the net and the compact pilot.}\label{FigELL1CFG}
\end{figure}
In the $PCFG$ there are new nodes (only node $0_T$ in this example), which  derive from the initial candidates (excluding those containing the axiom $0_S$) extracted from the closure part of the m-states of $\mathcal{C}$. Having added such nodes, the closure part of the $\mathcal{C}$ nodes (except for node $I_0$) becomes redundant and has been eliminated from $PCFG$ node contents.
\par
As said, prospect sets are needed only in final states; they have the following properties:
\begin{itemize}
\item For final states that are not initial, the prospect set coincides with the corresponding look-ahead set of the compact
pilot $\mathcal{C}$. This is the case of nodes $1_E$ and $3_T$.
\item For a final-initial state, such as $0_E$, the prospect set is the union of the look-ahead sets of every candidate
$\langle 0_E, \, \pi \rangle$ that occurs in $\mathcal{C}$. For instance, $\langle 0_E, \, \set{ \; ), \, \dashv \; } \rangle$ takes $\dashv$ from m-state $K_{0_E}$ and $)$ from $K_{1_T}$.
\end{itemize}
Solid edges represent the shift ones, already present in the machine net and pilot graph. Dashed edges represent the call ones, labeled by guide sets, and how to compute them is next illustrated:
\begin{itemize}
\item the guide set of call edge $0_E \dashrightarrow 0_T$ (and of $1_E \dashrightarrow 0_T$) is $\set{ \; a, \, ) \; }$, since from state $0_T$
both characters can be shifted
\item the guide set of edge $1_T \dashrightarrow 0_E$ includes the terminals:
\par
\begin{tabular}{lp{11cm}}
$)$ & since $1_T \stackrel E \to 2_T$ is in $M_T$, language $L \, (0_E)$ is nullable and $) \in Ini \, (2_T)$ \\
$a$, $($ & since from $0_E$ a call edge goes out to $0_T$ with prospect set $\set{ \; a, \, ( \; }$
\end{tabular}
\end{itemize}
In accordance with Prop. \ref{PropDisjointGuideSets}, the terminal labels of all the edges that originate from the same node, do not overlap.
\end{example}
\subsubsection{Predictive parser}
It is straightforward to derive the  parser  from the $PCFG$; its nodes are the pushdown stack elements; the top of stack element identifies the state of the active machine, while  inner stack elements refer to states of suspended machines, in the correct order of suspension. There are four sorts of moves. A \emph{scan} move, associated with a terminal shift edge, reads the current character, $cc$, as the corresponding machine would do. A \emph{call} move, associated with a call edge, checks the enabling predicate, saves on stack the return state, and switches to the invoked machine without consuming $cc$. A \emph{return} move is triggered when the active machine enters a final state whose prospect set includes $cc$: the active state is set to the return state of the most recently suspended machine. A \emph{recognizing} move terminates parsing.
\begin{algorithm} \label{AlgorTopdownPDA}Predictive recognizer, $\mathcal{A}$.
\begin{itemize}
\item
The stack elements are the states of $PCFG$ $\mathcal{F}$. At the beginning the stack contains the initial candidate $\langle 0_S \rangle$.
\item
Let $\langle q_A \rangle$ be the top element, meaning that the active machine $M_A$ is in state $q_A$. Moves are next specified, this way:
\begin{itemize}
\item
\emph{scan move}
\par
if the shift edge $q_A \stackrel{cc} \longrightarrow r_A$ exists, then scan the next character and
replace the stack top by $\langle r_A \rangle$ (the active machine does not change)
\item
\emph{call move}
\par
if there exists a call edge $q_A \stackrel {\gamma} \dashrightarrow  0_{B}$ such that $cc \in \gamma$, let $q_A \stackrel B \to r_{A}$ be the corresponding nonterminal shift edge; then pop, push element $\langle r_A \rangle$ and push element $\langle 0_B \rangle$
\item
\emph{return move}
\par
if $q_A$ is a final state and $cc$ is in the prospect set associated with $q_A$ in the $PCFG$ $\mathcal{F}$, then pop
\item
\emph{recognition move}
\par
if $M_A$ is the axiom machine, $q_A$ is a final state and $cc = \dashv$, then accept and halt
\item
in any other case, reject the string and halt
\end{itemize}
\end{itemize}
\end{algorithm}
\par
From Prop. \ref{PropDisjointGuideSets} it follows that for every parsing configuration at most one move is possible, i.e., the algorithm is deterministic.
\subsubsection{Computing left derivations} To construct the syntax tree, the algorithm is next extended with an output function, thus turning the $DPDA$ into a pushdown transducer that computes the left derivation of the input string, using the right-linearized grammar $\hat{G}$. The output actions are specified in Table \ref{tabOutputPredictParser} and are so straightforward that we do not need to prove their correctness. Moreover, we recall that the syntax tree for grammar $\hat{G}$ is essentially an encoding of the syntax tree for the original $EBNF$ grammar $G$, such that each node has at most two child nodes.
\begin{table}[htb]
\caption{Derivation steps computed by predictive parser. Current character is $cc = b$.}\label{tabOutputPredictParser}
\begin{center}
\vspace{0.25cm}
\def\arraystretch{1.5}\begin{tabular}{p{8cm}|l}
\textit{Parser move} & \textit{Output derivation step} \\ \hline
Scan move for transition $q_A \stackrel{b}\longrightarrow r_A$ & $q_A \underset {\hat{G}} \Longrightarrow b \, r_A$ \\
Call move for call edge $q_A \stackrel {\gamma} \dashrightarrow  0_{B}$ and transition $q_A \stackrel {0_{B}} \to r_{A}$ & $q_A \underset {\hat{G}} \Longrightarrow 0_B \, r_A$ \\
Return move for state $q_A \in F_A$ & $q_A \underset {\hat{G}} \Longrightarrow \varepsilon$ \\
\end{tabular}
\end{center}
\end{table}
\begin{example}Running example: trace of predictive parser for input $x = ( \, a \, )$.
\begin{center}
\def\arraystretch{1.5}\begin{tabular}{lr|c|l}
\emph{stack} & $x$ & \emph{predicate} & \emph{left derivation} \\ \hline
$\vline \left\langle 0_E \right\rangle \vline$ & $( \, a \, ) \, \dashv$ & $( \in \gamma = \set{ \, a \, ( \, }$ & $0_E \Rightarrow 0_T \, 1_E$ \\ \hline
$\vline \left\langle 1_E \right\rangle \vline \langle 0_T \rangle \vline$ & $( \, a \, ) \, \dashv$ & scan & $0_E \stackrel + \Rightarrow ( \, 1_T \, 1_E$ \\ \hline
$\vline \left\langle 1_E \right\rangle \vline \langle 1_T \rangle \vline$ & $a \, ) \, \dashv$ & $a \in \gamma = \set{ \, a \, ( \, }$ & $0_E \stackrel + \Rightarrow ( \, 0_E \, 2_T \, 1_E$ \\ \hline
$\vline \left\langle 1_E \right\rangle \vline \langle 2_T \rangle \vline \langle 0_E \rangle \vline$ & $a \, ) \, \dashv$ & $a \in \gamma = \set{ \, a \, ( \, }$ & $0_E \stackrel + \Rightarrow ( \, 0_T \, 1_E \, 2_T \, 1_E$ \\ \hline
$\vline \left\langle 1_E \right\rangle \vline \langle 2_T \rangle \vline \langle 1_E \rangle \vline \langle 0_T \rangle\vline$ & $a \, ) \, \dashv$ & scan & $0_E \stackrel + \Rightarrow ( \, a \, 3_T \, 1_E \, 2_T \, 1_E$ \\ \hline
$\vline \left\langle 1_E \right\rangle \vline \langle 2_T \rangle \vline \langle 1_E \rangle \vline \langle 3_T \rangle\vline$ & $) \, \dashv$ & $) \in \pi = \set{ \, a \, ( \, ) \, \dashv \, }$ & $0_E \stackrel + \Rightarrow ( \, a \, \varepsilon \, 1_E \, 2_T \, 1_E$ \\ \hline
$\vline \left\langle 1_E \right\rangle \vline \langle 2_T \rangle \vline \langle 1_E \rangle \vline$ & $) \, \dashv$ & $) \in \pi = \set{ \, ) \, \dashv \, }$ & $0_E \stackrel + \Rightarrow ( \, a \, \varepsilon \, 2_T \, 1_E$ \\ \hline
$\vline \left\langle 1_E \right\rangle \vline \langle 2_T \rangle \vline$ & $) \, \dashv$ & scan & $0_E \stackrel + \Rightarrow ( \, a \, ) \, 3_T \, 1_E$ \\ \hline
$\vline \left\langle 1_E \right\rangle \vline \langle 3_T \rangle \vline$ & $\dashv$ & $\dashv \, \in \pi = \set{ \, a \, ( \, ) \, \dashv \, }$ & $0_E \stackrel + \Rightarrow ( \, a \, ) \, \varepsilon \, 1_E$ \\ \hline
$\vline \left\langle 1_E \right\rangle \vline$ & $\dashv$ & $\dashv \, \in \pi = \set{ \, ) \, \dashv \, }$ \ accept & $0_E \stackrel + \Rightarrow ( \, a \, ) \, \varepsilon$ \\ \hline
\end{tabular}
\end{center}
For the original grammar, the corresponding derivation is:
\[
 E \Rightarrow T \Rightarrow ( \, E \, ) \Rightarrow ( \, T \, ) \Rightarrow ( \, a \, )
\]
\end{example}
\subsubsection{Parser implementation by recursive procedures}\label{subsect:recDescParser}
Predictive parsers are often implemented using recursive procedures. Each machine is transformed into a parameter-less so-called \emph{syntactic procedure}, having a $CFG$ matching the corresponding $PCFG$ subgraph, so that the current state of a machine is encoded at runtime by the program counter. Parsing starts in the axiom procedure and successfully terminates when the input has been exhausted, unless an error has occurred before.
The standard runtime mechanism of procedure invocation and return automatically implements the call and return moves. An example should suffice to show that the procedure pseudo-code is mechanically obtained from the $PCFG$.
\begin{example}\label{esProcRicorsLL(1)}Recursive descent parser.
\par
For the $PCFG$ of Figure \ref{FigELL1CFG}, the syntactic procedures are shown in Figure \ref{figuraProcedureSintatEsDiscesaRic}. The pseudo-code can be optimized in several ways.
\par
\begin{figure}[h!]
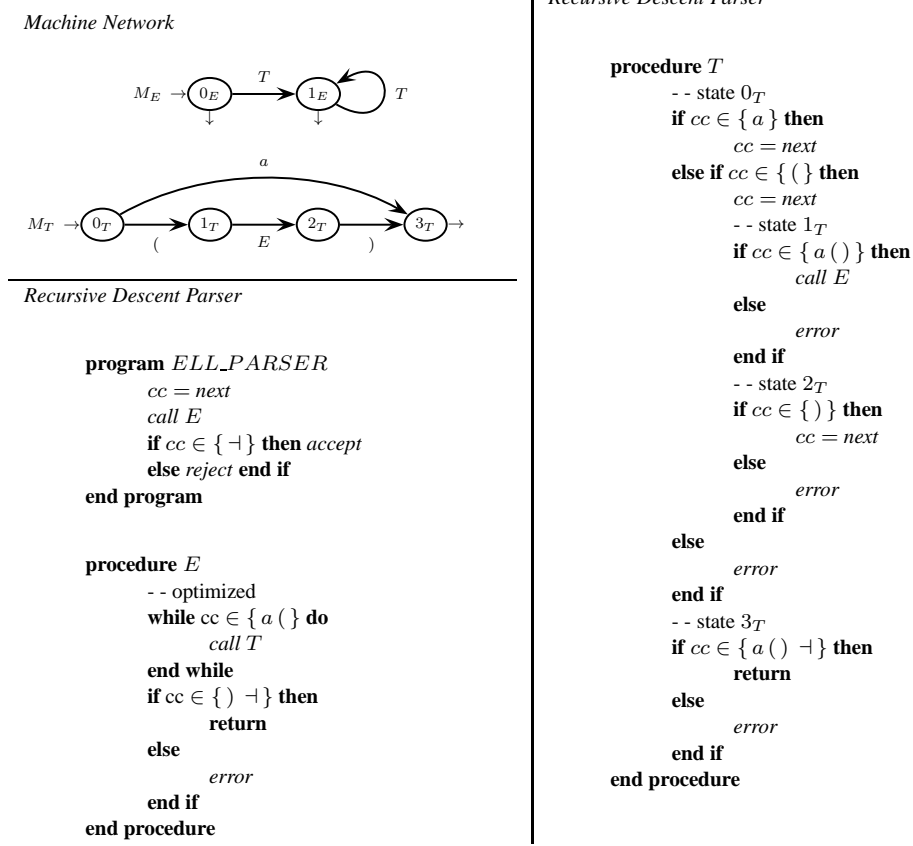

\begin{center}
\begin{tabular}{c|c}
\def\arraystretch{10}\begin{tabular}{c}

\begin{minipage}{0.50\textwidth}
\begin{center}
\begin{flushleft}
\emph{Machine Network}
\end{flushleft}
\par
\vspace{0.5cm}
\par
\scalebox{0.75}{
\pspar\psset{border=0pt,nodesep=0pt,labelsep=5pt,colsep=1.125cm}
\begin{psmatrix}
\ovalnode{0E}{$0_E$} & \ovalnode{1E}{$1_E$}

\nput[labelsep=0pt]{180}{0E}{$M_E \ \to$}

\nput[labelsep=0pt]{-90}{1E}{$\downarrow$}

\nput[labelsep=0pt]{-90}{0E}{$\downarrow$}

\ncline{0E}{1E} \naput{$T$}

\nccurve[angleA=-30,angleB=30,ncurvA=7,ncurvB=7]{1E}{1E} \nbput{$T$}
\end{psmatrix}}
\par
\vspace{1.25cm}
\par
\scalebox{0.75}{
\pspar\psset{border=0pt,nodesep=0pt,labelsep=5pt,,colsep=1.125cm}
\begin{psmatrix}
\ovalnode{0T}{$0_T$} & \ovalnode{1T}{$1_T$} & \ovalnode{2T}{$2_T$} & \ovalnode{3T}{$3_T$}

\nput[labelsep=0pt]{180}{0T}{$M_T \ \to$}

\nput[labelsep=0pt]{0}{3T}{$\to$}

\ncline{0T}{1T} \nbput{$($}

\ncline{1T}{2T} \nbput{$E$}

\ncline{2T}{3T} \nbput{$)$}

\ncarc{0T}{3T} \naput{$a$}
\end{psmatrix}}
\par
\vspace{0.5cm}
\end{center}
\end{minipage}

\\ \hline

\begin{minipage}{0.50\textwidth}
\vspace{0.1cm}
\begin{flushleft}
\emph{Recursive Descent Parser}
\end{flushleft}
\par
\begin{tabbing}
\hspace{0.75cm} \= \hspace{0.75cm} \= \hspace{0.75cm} \= \hspace{0.75cm} \= \kill \\
\>\textbf{program} $ELL\_PARSER$ \\
\>\> $cc =$ \emph{next} \\
\>\> \emph{call} $E$ \\
\>\> \textbf{if} $cc \in \set{ \, \dashv \, }$ \textbf{then} \emph{accept} \\
\>\> \textbf{else} \emph{reject} \textbf{end if} \\
\>\textbf{end program}
\end{tabbing}
\par
\begin{tabbing}
\hspace{0.75cm} \= \hspace{0.75cm} \= \hspace{0.75cm} \= \hspace{0.75cm} \= \kill \\
\>\textbf{procedure} $E$ \\
\>\> - - optimized \\
\>\> \textbf{while} cc $\in \set{ \, a \, ( \, }$ \textbf{do} \\
\>\>\> \emph{call} $T$ \\
\>\> \textbf{end while} \\
\>\> \textbf{if} cc $\in \set{ \, ) \, \dashv \, }$ \textbf{then} \\
\>\>\> \textbf{return} \\
\>\> \textbf{else} \\
\>\>\> \emph{error} \\
\>\> \textbf{end if} \\
\>\textbf{end procedure}
\end{tabbing}
\end{minipage}

\end{tabular}

&

\begin{minipage}{0.40\textwidth}
\begin{flushleft}
\emph{Recursive Descent Parser}
\end{flushleft}
\par
\begin{tabbing}
\hspace{0.75cm} \= \hspace{0.75cm} \= \hspace{0.75cm} \= \hspace{0.75cm} \= \hspace{0.75cm} \= \kill \\
\>\textbf{procedure} $T$ \\
\>\> - - state $0_T$ \\
\>\> \textbf{if} $cc \in \set{ \, a \, }$ \textbf{then} \\
\>\>\> $cc =$ \emph{next} \\
\>\> \textbf{else if} $cc \in \set{ \, ( \, }$ \textbf{then} \\
\>\>\> $cc =$ \emph{next} \\
\>\>\> - - state $1_T$ \\
\>\>\> \textbf{if} $cc \in \set{ \, a \, ( \, ) \, }$ \textbf{then} \\
\>\>\>\> \emph{call} $E$ \\
\>\>\> \textbf{else} \\
\>\>\>\> \emph{error} \\
\>\>\> \textbf{end if} \\
\>\>\> - - state $2_T$ \\
\>\>\> \textbf{if} $cc \in \set{ \, ) \, }$ \textbf{then} \\
\>\>\>\> $cc =$ \emph{next} \\
\>\>\> \textbf{else} \\
\>\>\>\> \emph{error} \\
\>\>\> \textbf{end if} \\
\>\> \textbf{else} \\
\>\>\> \emph{error} \\
\>\> \textbf{end if} \\
\>\> - - state $3_T$ \\
\>\> \textbf{if} $cc \in \set{ \, a \, ( \, ) \, \dashv \, }$ \textbf{then} \\
\>\>\> \textbf{return} \\
\>\> \textbf{else} \\
\>\>\> \emph{error} \\
\>\> \textbf{end if} \\
\>\textbf{end procedure} \\
\end{tabbing}
\end{minipage}

\end{tabular}
\end{center}
\caption{Main program and syntactic procedures of a recursive descent parser (Ex. \ref{esProcRicorsLL(1)} and Fig. \ref{FigELL1CFG}); function \emph{next} is the programming interface to the lexical analyzer or scanner; function \emph{error} is the messaging interface.} \label{figuraProcedureSintatEsDiscesaRic}
\end{figure}
\end{example}
\subsection{Direct construction of parser control-flow graph}\label{sectSimplifELL1parserConstr}
We have presented and justified a series of rigorous steps that lead from an $ELR \, (1)$ pilot to the compact pointer-less parser, and finally, to the parser control-flow graph. However, for a human wishing to design a predictive parser it would be tedious to perform all those steps.
We therefore provide a simpler procedure for checking that an $EBNF$ grammar satisfies the $ELL \, (1)$ condition. The procedure
operates directly on the grammar Parser Control Flow Graph and does not require the construction of the $ELR \, (1)$ pilot.
It uses a set of recursive equations defining the prospect and guide sets for all the states and edges of the $PCFG$; the equations
are interpreted as instructions to compute iteratively the guide sets, after which the $ELL \, (1)$ check simply verifies, according
to Property \ref{PropDisjointGuideSetsConv}, that the guide sets are disjoint.
\subsubsection{Equations defining the prospect sets}\label{prospSetEqtn}
\begin{enumerate}
\item If the net includes shift edges of the kind $p_i \stackrel {X_i} \to q$ then the prospect set $\pi_q$ of state $q$ is:	
\[
\pi_q := \bigcup_{p_i \stackrel {X_i} \to q} \pi_{p_i}
\]
\item If the net includes nonterminal shift edge $q_i \stackrel {A} \to r_i$ and the corresponding call edge in the PCFG is $q_i \dashrightarrow 0_A$, then
the prospect set for the initial state $0_A$ of machine $M_A$ is:
 \[
\pi_{0_A} := \pi_{0_A} \cup \bigcup_{q_i \stackrel {A} \to r_i} \Big( \, Ini \, \big( \, L \, (r_i) \, \big) \ \cup \ \mathbf{if} \ \ Nullable \, \big( \, L \, (r_i) \, \big) \ \ \mathbf{then} \ \ \pi_{q_i} \ \ \mathbf{else} \ \ \emptyset \, \Big)
\]
\end{enumerate}
Notice that the two sets of rules apply in an exclusive way to disjoints sets of nodes, because the normalization of the machines disallows re-entrance into initial states.
\subsubsection{Equations defining the guide sets}\label{guideSetEqtn}
\begin{enumerate}
\item For each call edge $q_A \dashrightarrow 0_{A_1}$ associated with a nonterminal shift edge $q_A \stackrel {A_1} \to r_A$,
such that possibly other call edges $0_{A_1} \dashrightarrow 0_{B_i}$ depart from state $0_{A_1}$,
the guide set $Gui \, (q_A \dashrightarrow 0_{A_1})$ of the call edge is defined as follows, see also conditions (\ref{eqGuideSet1}-\ref{eqGuideSet4}):
\[
Gui \, (q_A \dashrightarrow 0_{A_1}) := \bigcup \ \left\{
\def\arraystretch{1.5}\begin{array}{l}
Ini \, \big( \, L \, (A_1 \, \big) \\ \hline
\mathbf{if} \ Nullable \, (A_1) \ \mathbf{then} \ \ Ini \, \big( L \, (r_A) \big) \ \mathbf{else} \ \emptyset \\ \hline
\mathbf{if} \ Nullable \, (A_1) \wedge \, Nullable \, \big( L \, (r_A) \, \big) \ \mathbf{then} \ \pi_{r_A} \ \mathbf{else} \ \emptyset \\ \hline
\underset {0_{A_1} \dashrightarrow 0_{B_i}} \bigcup Gui \, ( 0_{A_1} \dashrightarrow 0_{B_i} )
\end{array}
\right.
\]
\item For a final state $f_A \in F_A$, the guide set of the  tagging dart equals the prospect set:
\[
Gui \, (f_A \to) := \pi_{f_A}
\]
\item For a terminal shift edge $q_A \stackrel a \to r_A$ with $a \in \Sigma$, the guide set is simply the shifted terminal:
\[
Gui \, \left(q_A \stackrel a \to r_A\right) := \set{ \, a \, }
\] 	
\end{enumerate}
\par
A computation starts by assigning to the prospect set of $0_S$ (initial state of axiom machine) the end-marker: $\pi_{0_S} := \set{ \, \dashv \, }$.
All other sets are initialized to empty. Then the above rules are repeatedly applied until a fixpoint is reached.
\par
Notice that the rules for computing the prospect sets are consistent with the definition of look-ahead set
given in Section \ref{callSitesActivationLookahead}; furthermore, the rules for computing the guide sets
are consistent with the definition of Parser Control Flow Graph provided in Section \ref{defPilotCFG}.
\begin{example}Running example: computing the prospect and guide sets.
\par
The following table shows the computation of  most prospect and guide sets for the $PCFG$ of Figure \ref{FigELL1CFG}
(p. \pageref{FigELL1CFG}). The computation is completed at the third step.
\begin{center}
\vspace{0.25cm}
$\def\arraystretch{1.5}\begin{array}{|c|c|c|c|c|c||c|c|c|}
\hline
\multicolumn{6}{|c||}{\emph{Prospect sets of}} & \multicolumn{3}{c|}{\emph{Guide sets of}} \\ \hline\hline
0_E & 1_E & { 0_T } & {1_T} & {2_T} & {3_T} & 0_E \dashrightarrow 0_T & 1_E \dashrightarrow 0_T & 1_T \dashrightarrow 0_E \\ \hline\hline
\dashv & \emptyset & \emptyset & \emptyset & \emptyset & \emptyset & \emptyset & \emptyset & \emptyset \\ \hline
 \, ) \, \dashv  \, &  \, ) \, \dashv  \, &  \, a \, ( \, ) \, \dashv  \, &  \, a \, ( \, ) \, \dashv  \, &  \, a \, ( \, ) \, \dashv  \, &  \, a \, ( \, ) \, \dashv  \, & a \; ( & a \; ( & a \; ( \; ) \\ \hline
 \, ) \, \dashv  \, &  \, ) \, \dashv  \, &  \, a \, ( \, ) \, \dashv  \, &  \, a \, ( \, ) \, \dashv  \, &  \, a \, ( \, ) \, \dashv  \, &  \, a \, ( \, ) \, \dashv  \, & a \; ( & a \; ( & a \; ( \; ) \\ \hline
\end{array}$
\vspace{0.25cm}
\end{center}
\end{example}
\begin{example} Guide sets in a non-$ELL(1)$ grammar.
\par
The grammar $\set{ \; S \to a^\ast \; N, \; N \to a \; N \; b \; \mid \; \varepsilon \; }$ of example \ref{exELRnonsingletonBase} (p.\pageref{exELRnonsingletonBase}) violates the single transition property, as shown in Figure \ref{figMultiCandidatesInBase},
hence it is not $ELL \, (1)$. This can be verified also by computing the guide sets on the $PCFG$. We have
$Gui \, (0_S \dashrightarrow 0_N) \, \cap \, Gui \, \left(0_S \stackrel a \to 1_S \right) = \set{ \, a \, } \neq \emptyset$ and
$Gui \, (1_S \dashrightarrow 0_N) \, \cap \, Gui \, \left(1_S \stackrel a \to 1_S \right) = \set{ \, a \, } \neq \emptyset$: the guide sets on the edges departing from states $0_S$ and $1_S$ are not disjoint.
\end{example}
\section{Tabular parsing}\label{sectionTabular}\label{SectGeneralParsingMeth}
The $LR$ and $LL$ parsing methods are inadequate for dealing with nondeterministic and ambiguous grammars. The seminal work \cite{Earley70} introduced a cubic-time algorithm for recognizing strings of any context-free grammar, even of an ambiguous one, though it did not explicitly present a parsing method, i.e., a means for constructing the (potentially numerous) parsing trees of the accepted string. Later, efficient representations of parse forests have been invented, which do not duplicate the common subtrees and thus achieve a polynomial-time complexity for the tree construction algorithm. Until recently \cite{journals/acta/AycockB09}, Earley parsers performed a complete recognition of the input before constructing any syntax tree \cite{GruneJacobs2008}.
\par
Concerning the possibility of directly using extended $BNF$ grammars, Earley himself already gave some hints and later a few parser generators have been implemented, but no authoritative work exists, to the best of our knowledge.
\par
Another long-standing discussion concerns the pros and cons of using look-ahead. Since this issue is out of scope here, and a few experimental studies  (see \cite{AycockHorspool2002}) have indicated that look-ahead-less algorithms can be faster at least for programming languages, we present the simpler version that does not use look-ahead. This is in line with the classical theoretical presentations of the Earley parsers for $BNF$ grammars, such as the one in \cite{Revesz1991}.
\par
Actually, our focus in the present work is on programming languages that are non-ambiguous formal notations (unlike natural languages). Therefore our variant of the Earley algorithm is well suited for non-ambiguous $EBNF$ grammars, though possibly nondeterministic, and the related procedure for building parse trees does not deal with multiple trees or forests.
\subsection{String recognition}
Our algorithm is straightforward to understand, if one comes from the Vector Stack implementation of the $ELR \, (1)$ parser of Section \ref{VectorStackSection}. When analyzing a string $x = x_1 \, \ldots \, x_n$ (with $\vert \, x \, \vert = n \geq 1$) or $x = \varepsilon$ (with $\vert \, x \, \vert = n = 0$), the algorithm uses a vector $E \, [0 \, \ldots \, n]$ or $E \, [0]$, respectively, of $n + 1 \geq 1$ elements, called Earley vector.
\par
Every vector element $E \, [i]$ contains a set of pairs $\left\langle \, p_X, \, j \, \right\rangle$ that consist of a state $p_X$ of the machine $M_X$ for some nonterminal $X$, and of an integer index $j$ (with $0 \leq j \leq i \leq n$) that points back to the element $E \, [j]$ that contains a corresponding pair with the initial state $0_X$ of the same machine $M_X$. This index marks the position in the input string, from where the currently assumed derivation from nonterminal $X$ may have started.
\par
Before introducing our variant of the Earley Algorithm for $EBNF$ grammars, we preliminarily define the operations \textit{Completion} and \textit{TerminalShift}.
\begin{tabbing}
\hspace{0.25cm} \= \hspace{0.75cm} \= \hspace{0.75cm} \= \hspace{0.75cm} \=\hspace{0.75cm} \= \hspace{0.75cm} \= \hspace{0.75cm} \= \kill
\> $\emph{Completion} \, (E, \, i)$ \> \> \> \> \> - - with index $0 \leq i \leq n$ \\[0.2cm]
\> \textbf{do} \\[0.2cm]
\> \> - - loop that computes the \emph{closure} operation \\
\> \> - - for each pair that launches machine $M_X$ \\[0.2cm]
\> \> \textbf{for} $\left( \text{each pair $\left\langle \; p, \; j \; \right\rangle \in E \, [i]$ and $X, \, q \in V, \, Q$ s.t. $p \stackrel X \to q$} \right)$ \textbf{do} \\[0.2cm]
\> \> \> add pair $\left\langle \; 0_X, \; i \; \right\rangle$ to element $E \, [i]$ \\[0.2cm]
\> \> \textbf{end for} \\[0.2cm]
\> \> - - nested loops that compute the \emph{nonterminal shift} operation \\
\> \> - - for each final pair that enables a shift on nonterminal $X$ \\[0.2cm]
\> \> \textbf{for} $\big( \text{each pair $\left\langle \; f, \; j \; \right\rangle \in E \, [i]$ and $X \in V$ such that $f \in F_X$} \big)$ \textbf{do} \\[0.2cm]
\> \> \> - - for each pair that shifts on nonterminal $X$ \\[0.2cm]
\> \> \> \textbf{for} $\left( \text{each pair $\left\langle \; p, \; l \; \right\rangle \in E \, [j]$ and $q \in Q$ s.t. $p \stackrel X \to q$} \right)$ \textbf{do} \\[0.2cm]
\> \> \> \> add pair $\left\langle \; q, \; l \; \right\rangle$ to element $E \, [i]$ \\[0.1cm]
\> \> \> \textbf{end for} \\[0.2cm]
\> \> \textbf{end for} \\[0.2cm]
\> \textbf{while} \; $\big( \, \text{some pair has been added} \, \big)$
\end{tabbing}
Notice that in the \textit{Completion} operation the nullable nonterminals are dealt with by a combination of closure and nonterminal shift operations.
\begin{tabbing}
\hspace{0.25cm} \= \hspace{0.75cm} \= \hspace{0.75cm} \= \hspace{0.75cm} \=\hspace{0.75cm} \= \hspace{0.75cm} \= \hspace{0.75cm} \= \kill
\> $\emph{TerminalShift} \, (E, \, i)$ \> \> \> \> \> - - with index $1 \leq i \leq n$ \\[0.2cm]
\> - - loop that computes the \emph{terminal shift} operation \\
\> - - for each preceding pair that shifts on terminal $x_i$ \\[0.2cm]
\> \textbf{for} $\left( \text{each pair $\left\langle \; p, \; j \; \right\rangle \in E \, [i - 1]$ and $q \in Q$ s.t. $p \stackrel {x_i} \to q$} \right)$ \textbf{do} \\[0.2cm]
\> \> add pair $\left\langle \; q, \; j \; \right\rangle$ to element $E \, [i]$ \\[0.2cm]
\> \textbf{end for}
\end{tabbing}
The algorithm below for Earley syntactic analysis uses $Completion$ and $TerminalShift$.
\begin{algorithm}Earley syntactic analysis.
\begin{tabbing}
\hspace{0.25cm} \= \hspace{0.75cm} \= \hspace{0.75cm} \= \hspace{0.75cm} \=\hspace{0.75cm} \= \hspace{0.75cm} \= \hspace{0.75cm} \= \kill
\> - - analyze the terminal string $x$ for possible acceptance \\[0.0cm]
\> - - define the Earley vector $E \, [0 \dots n]$ with $\vert x \vert = n \geq 0$ \\[0.2cm]
\> $E \, [0] := \set{ \; \left\langle \; 0_S, \; 0 \; \right\rangle \; }$ \> \> \> \> \> \> - - initialize the first elem. $E \, [0]$ \\[0.2cm]
\> \textbf{for $i := 1$ to $n$ do} \> \> \> \> \> \> - - initialize all elem.s $E \, [1 \dots n]$  \\[0.2cm]
\> \> $E \, [i] := \emptyset$ \\[0.2cm]
\> \textbf{end for} \\[0.2cm]
\> $\emph{Completion} \, (E, \, 0)$ \> \> \> \> \> \> - - complete the first elem. $E \, [0]$ \\[0.2cm]
\> $i := 1$ \\[0.2cm]
\> - - while the vector is not finished and the previous elem. is not empty \\[0.2cm]
\> \textbf{while} $\big( \; i \leq n \ \land \ E \, [i - 1] \not = \emptyset \; \big)$ \textbf{do} \\[0.2cm]
\> \> $\emph{TerminalShift} \, (E, \, i)$ \> \> \> \> \> - - put into the current elem. $E \, [i]$ \\[0.2cm]
\> \> $\emph{Completion} \, (E, \, i)$ \> \> \> \> \> - - complete the current elem. $E \, [i]$ \\[0.2cm]
\> \> $i++$ \\[0.2cm]
\> \textbf{end while} \qed
\end{tabbing}
\end{algorithm}
\begin{example}\label{exEarley1}Non-deterministic $EBNF$ grammar with nullable nonterminal.
\par
The Earley acceptance condition if the following: $\left\langle \; f, \; 0 \; \right\rangle \in E \, [n]$ with $f \in F_S$. A string $x$ belongs to language $L \, (G)$ if and only of the Earley acceptance condition is true. Figure \ref{figGrammAndNet1For Earley} lists a non-ambiguous, non-deterministic $EBNF$ grammar and shows the corresponding machine net. The string $a\,a\,b\,b\,a\,a$ is analyzed: its syntax tree and analysis trace are in Figures \ref{FigEarleySyntaxTree} and \ref{figTraceEarley}, respectively; the edges in the latter figure ought to be ignored as they are related with the syntax tree construction discussed later.
\par
\begin{figure}[h!]
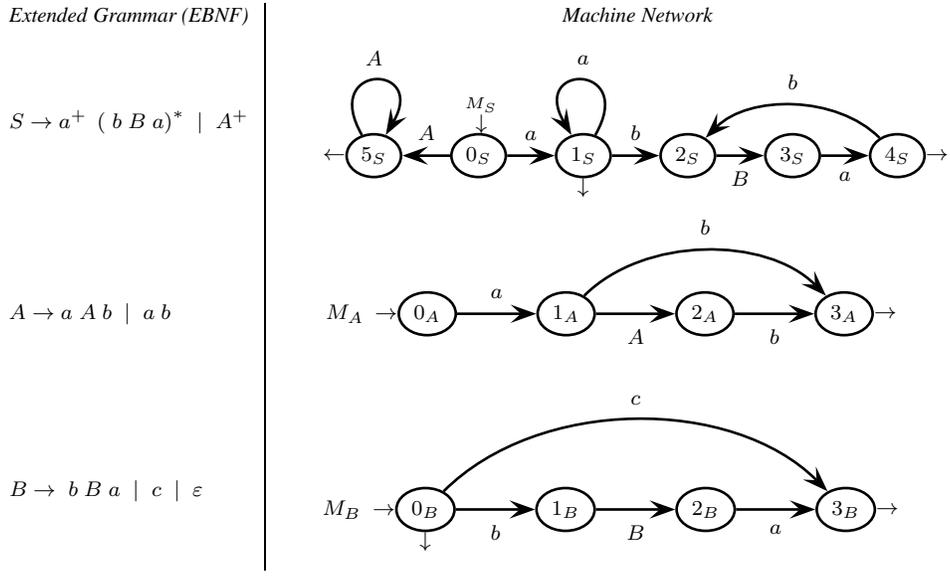


\begin{center}
\pspar\psset{border=0pt,nodesep=0pt,labelsep=5pt,colsep=30pt,rowsep=30pt}
\begin{tabular}{l|c}

\emph{Extended Grammar (EBNF)}

&

\emph{Machine Network} \\

$S \to a^+ \; \left( \; b \; B \; a \right)^\ast \; \mid \; A^+$

&

\begin{minipage}{0.75\textwidth}
\begin{center}
\vspace{1.375cm}
\pspar\psset{border=0pt,nodesep=0pt,labelsep=5pt,colsep=18pt,rowsep=30pt}
\begin{psmatrix}
\ovalnode{5S}{$5_S$} & \ovalnode{0S}{$0_S$} & \ovalnode{1S}{$1_S$} & \ovalnode{2S}{$2_S$} & \ovalnode{3S}{$3_S$} & \ovalnode{4S}{$4_S$}

\nput[labelsep=0pt]{90}{0S}{ $\overset {M_S} \downarrow$}

\nput[labelsep=0pt]{-90}{1S}{$\downarrow$}

\nput[labelsep=0pt]{0}{4S}{$\to$}

\nput[labelsep=0pt]{180}{5S}{$\leftarrow$}

\ncline{0S}{1S} \naput{$a$}

\ncline{0S}{5S} \nbput{$A$}

\nccurve[angleA=60,angleB=120,ncurvA=7,ncurvB=7]{1S}{1S} \nbput{$a$}

\ncline{1S}{2S} \naput{$b$}

\ncline{2S}{3S} \nbput{$B$}

\ncline{3S}{4S} \nbput{$a$}

\nccurve[angleA=135,angleB=45,ncurv=0.7]{4S}{2S} \nbput{$b$}

\nccurve[angleA=120,angleB=60,ncurvA=7,ncurvB=7]{5S}{5S} \naput{$A$}

\end{psmatrix}
\vspace{0.5cm}
\end{center}
\end{minipage} \\

$A \to a \; A \; b \; \mid \; a \; b$

&

\begin{minipage}{0.75\textwidth}
\begin{center}
\vspace{1cm}
\begin{psmatrix}

\ovalnode{0A}{$0_A$} & \ovalnode{1A}{$1_A$} & \ovalnode{2A}{$2_A$} & \ovalnode{3A}{$3_A$}

\nput[labelsep=0pt]{180}{0A}{$M_A \ \to$}

\nput[labelsep=0pt]{0}{3A}{$\to$}

\ncline{0A}{1A} \naput{$a$}

\ncline{1A}{2A} \nbput{$A$}

\ncline{2A}{3A} \nbput{$b$}

\nccurve[angleA=45,angleB=135,ncurv=0.7]{1A}{3A} \naput{$b$}

\end{psmatrix}
\vspace{1cm}
\end{center}
\end{minipage} \\

$B \to \; b \; B \; a \; \mid \; c \; \mid \; \varepsilon$

&

\begin{minipage}{0.75\textwidth}
\begin{center}
\vspace{1cm}
\begin{psmatrix}

\ovalnode{0B}{$0_B$} & \ovalnode{1B}{$1_B$} & \ovalnode{2B}{$2_B$} & \ovalnode{3B}{$3_B$}

\nput[labelsep=0pt]{180}{0B}{$M_B \ \to$}

\nput[labelsep=0pt]{-90}{0B}{$\downarrow$}

\nput[labelsep=0pt]{0}{3B}{$\to$}

\ncline{0B}{1B} \nbput{$b$}

\ncline{1B}{2B} \nbput{$B$}

\ncline{2B}{3B} \nbput{$a$}

\nccurve[angleA=45,angleB=135,ncurv=0.7]{0B}{3B} \naput{$c$}

\end{psmatrix}
\vspace{0.5cm}
\end{center}
\end{minipage}
\end{tabular}
\end{center}
\caption{Non-deterministic $EBNF$ grammar $G$ and network of Example \ref{exEarley1}.}\label{figGrammAndNet1For Earley}
\end{figure}
\par
\begin{figure}[h!]
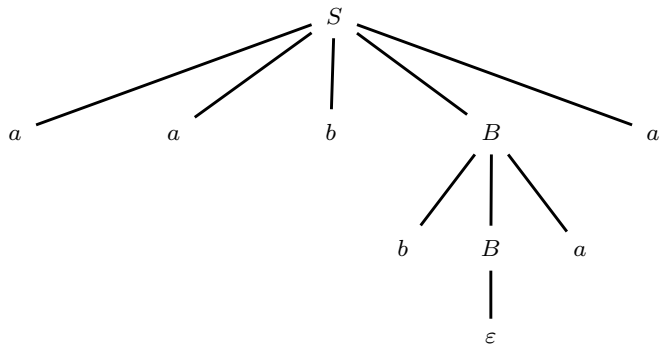

\begin{center}
\scalebox{1.1}{
\pspar\psset{arrows=-,linestyle=solid,nodesep=5pt,treefit=tight}
\pstree[levelsep=40pt,treesep=50pt]{\TR{$S$}}{
    \TR{$a$}
    \TR{$a$}
    \TR{$b$}  			
    \pstree[levelsep=40pt, treesep=25pt]{\TR{$B$}}{
        \TR{$b$}
        \pstree[levelsep=30pt]{\TR{$B$}}{
            \TR{$\varepsilon$}									  	
		}
  		\TR{$a$} 				  	 	
  	}  	   			
  	\TR{$a$}  	
}}
\caption{Syntax tree for the string $a\,a\,b\,b\,a\,a$ of Example \ref{exEarley1}.} \label{FigEarleySyntaxTree}
\end{center}
\end{figure}
\begin{figure}[h!]
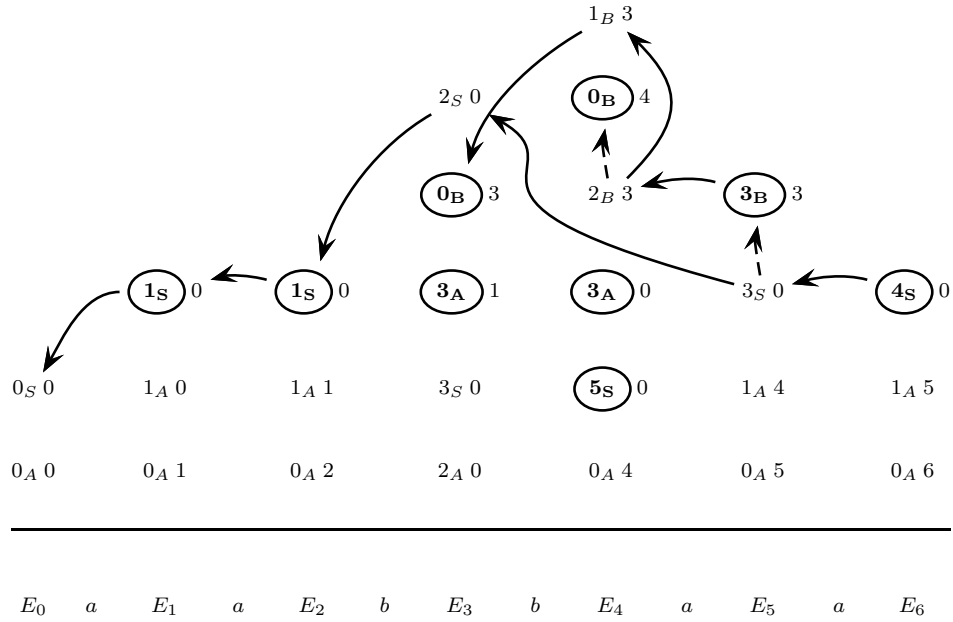

\begin{center}
\scalebox{1.0}{
\pspar\psset{nodesep=3pt,labelsep=5pt,rowsep=0.65cm,colsep=0.4cm}
\begin{psmatrix}

&&&&&&&& \rnode{41}{$1_B \ 3$} \\

&&&&&& \rnode{31}{$2_S \ 0$} && \rnode{42}{$\ovalnode{\relax}{\mathbf{0_B}} \ 4$} \\

&&&&&& \rnode{32}{$\ovalnode{\relax}{\mathbf{0_B}} \ 3$} && \rnode{43}{$2_B \ 3$} && \rnode{53}{$\ovalnode{\relax}{\mathbf{3_B}} \ 3$} \\

&& \rnode{11}{$\ovalnode{\relax}{\mathbf{1_S}} \ 0$} && \rnode{21}{$\ovalnode{\relax}{\mathbf{1_S}} \ 0$} && \rnode{33}{$\ovalnode{\relax}{\mathbf{3_A}} \ 1$} && \rnode{44}{$\ovalnode{\relax}{\mathbf{3_A}} \ 0$} && \rnode{54}{$3_S \ 0$} && \rnode{61}{$\ovalnode{\relax}{\mathbf{4_S}} \ 0$} \\

\rnode{01}{$0_S \ 0$} && \rnode{12}{$1_A \ 0$} && \rnode{22}{$1_A \ 1$} && \rnode{34}{$3_S \ 0$}  && \rnode{45}{$\ovalnode{\relax}{\mathbf{5_S}} \ 0$} && \rnode{55}{$1_A \ 4$} && \rnode{62}{$1_A \ 5$} \\

\rnode{02}{$0_A \ 0$} && \rnode{13}{$0_A \ 1$} && \rnode{23}{$0_A \ 2$} && \rnode{35}{$2_A \ 0$} && \rnode{46}{$0_A \ 4$} && \rnode{56}{$0_A \ 5$} && \rnode{63}{$0_A \ 6$} \\ \hline \\

$E_0$ & $a$ & $E_1$ & $a$ & $E_2$ & $b$ & $E_3$ & $b$  & $E_4$ & $a$& $E_5$ & $a$& $E_6$

\nccurve[angleA=-180,angleB=60]{11}{01}

\nccurve[angleA=165,angleB=15]{21}{11}

\nccurve[angleA=-150,angleB=75]{31}{21}

\nccurve[angleA=-150,angleB=75]{41}{32}

\nccurve[angleA=165,angleB=15]{53}{43}

\nccurve[angleA=165,angleB=15]{61}{54}

\nccurve[angleA=45,angleB=-45,ncurvA=1.1,ncurvB=1.1]{43}{41}

\nccurve[angleA=165,angleB=-30,ncurvA=2.25,ncurvB=0.75]{54}{31}

\nccurve[angleA=100,angleB=-100,linestyle=dashed]{43}{42}

\nccurve[angleA=100,angleB=-100,linestyle=dashed]{54}{53}

\end{psmatrix}}
\end{center}
\caption{Tabular parsing trace of string $a\,a\,b\,b\,a\,a$ with the machine net in Figure \ref{figGrammAndNet1For Earley}.}\label{figTraceEarley}
\end{figure}
\end{example}
\par
The following lemma correlates the presence of certain pairs in the Earley vector elements, with the existence of a leftmost derivation for the string prefix analyzed up to that point and, together with the associated corollary, provides a proof of the \textit{correctness} of the algorithm.
\par
\begin{lemma}\label{EarleyCorrect}  If it holds $\left\langle \; q_A, \; j \; \right\rangle \in E \, [i]$, which implies inequality $j \leq i$, with $q_A \in Q_A$, i.e., state $q_A$ belongs to the machine $M_A$ of nonterminal $A$, then it holds $\left\langle \; 0_A, \; j \; \right\rangle \in E \, [j]$ and the right-linearized grammar $\hat{G}$ admits a leftmost derivation $0_A \stackrel \ast \Rightarrow x_{j+1} \, \ldots \, x_i \, q_A$ if $j < i$ or $0_A \stackrel \ast \Rightarrow q_A$ if $j = i$. \qed
\end{lemma}
\par
The proof is in the Appendix.
\begin{corollary} If the Earley acceptance condition is satisfied, i.e., if $\left\langle \; f_S, \; 0 \; \right\rangle \in E \, [n]$ with $f_S \in F_S$, then the \emph{EBNF} grammar $G$ admits a derivation $S \stackrel + \Rightarrow x$, i.e., $x \in L \, (S)$, and string $x$ belongs to language $L \, (G)$. \qed
\end{corollary}
\par
The following lemma, which is the converse of Lemma \ref{EarleyCorrect}, states the \textit{completeness} of the Earley Algorithm.
\par
\begin{lemma}\label{EarleyComplete} Take an \emph{EBNF} grammar $G$ and a string $x = x_1 \, \dots \, x_n$ of length $n$ that belongs to language $L \, (G)$. In the right-linearized grammar $\hat{G}$, consider any leftmost derivation $d$ of a prefix $x_1 \, \dots \, x_i$ ($i \leq n$) of $x$, that is:
\[
d \colon 0_S \stackrel + \Rightarrow x_1 \, \ldots \, x_i \, q_A \, W
\]
with $q_A \in Q_A$ and $W \in Q_A^\ast$. The two points below apply:
\begin{enumerate}
\item \label{Earleycase1}
if it holds $W \neq \varepsilon$, i.e., $W = r_B \, Z$ for some $r_B \in Q_B$, then it holds $\exists \, j \; 0 \leq j \leq i$ and $\exists \, p_B \in Q_B$ such that the machine net has an arc $p_B \stackrel A \to r_B$ and grammar $\hat{G}$ admits two leftmost derivations $d_1 \colon 0_S \stackrel + \Rightarrow x_1 \, \ldots \, x_j \, p_B \, Z$ and $d_2 \colon 0_A \stackrel + \Rightarrow x_{j+1} \, \ldots \, x_i \, q_A$, so that derivation $d$ decomposes as follows:
\[\def\arraystretch{1.5}
\begin{array}{rclcl}
d \colon 0_S & \stackrel {d_1} \Rightarrow & x_1 \, \ldots \, x_j \, p_B \, Z & \stackrel {p_B \to 0_A \, r_B} \Longrightarrow & x_1 \, \ldots \, x_j \, 0_A \, r_B \, Z \\
& \stackrel {d_2} \Rightarrow & x_1 \, \ldots \, x_j  \, x_{j+1} \, \ldots \, x_i \, q_A \, r_B \, Z & = & x_1 \, \ldots \, x_i \, q_A \, W
\end{array}
\]
as an arc $p_B \stackrel A \to r_B$ in the net maps to a rule $p_B \to 0_A \, r_B$ in grammar $\hat{G}$
\item
this point is split into two steps, the second being the crucial one: \label{Earleycase2}
\begin{enumerate}
\item
if it holds $W = \varepsilon$, then it holds $A = S$, i.e., nonterminal $A$ is the axiom, $q_A \in Q_S$ and $\left\langle \; q_A, \; 0 \, \right\rangle
\in E \, [i]$ \label{Earleycase2a}
\item if it also holds $x_1 \ldots x_i \in L \, (G)$, i.e., the prefix also belongs to language $L \, (G)$,
then it holds $q_A = f_S \in F_S$, i.e., state $q_A = f_S$ is final for the axiomatic machine $M_S$, and the prefix is accepted by the Earley algorithm \label{Earleycase2b}
\end{enumerate}
\end{enumerate}
Limit cases: if it holds $i = 0$ then it holds $x_1 \, \ldots \, x_i = \varepsilon$; if it holds $j = i$ then it holds $x_{j+1} \, \ldots \, x_i = \varepsilon$; and if it holds $x = \varepsilon$ (so $n = 0$) then both cases hold, i.e., $j = i = 0$.
\par
If the prefix coincides with the whole string $x$, i.e., $i = n$, then step \eqref{Earleycase2bapp} implies that string $x$, which by hypothesis belongs to language $L \, (G)$, is accepted by the Earley algorithm, which therefore is complete. \qed
\end{lemma}
The proof is in the Appendix.
\subsection{Syntax tree construction}
The next procedure \emph{BuildTree} (\emph{BT}) builds the parse tree of a recognized string through processing the vector $E$ constructed by the Earley algorithm. Function \emph{BT} is recursive and has four formal parameters: nonterminal $X \in V$, state $f$, and two nonnegative indices  $j$ and $i$. Nonterminal $X$ is the root  of the  (sub)tree to be built. State $f$ is final for machine $M_X$; it is the end of the computation path in $M_X$ that corresponds to analyzing the substring generated by $X$. Indices $j$ and $i$  satisfy the inequality $0 \leq j \leq i \leq n$;  they respectively specify the left and right ends of the substring generated by  $X$:
\[
X \underset G {\stackrel + \Longrightarrow} x_{j + 1} \; \dots \; x_i \quad \text{if $j < i$} \hspace{2cm}
X \underset G {\stackrel + \Longrightarrow} \varepsilon \quad \text{if $j = i$}
\]
Grammar $G$ admits derivation $S \stackrel + \Rightarrow x_1 \; \dots \; x_n$ or $S \stackrel + \Rightarrow \varepsilon$, and the Earley algorithm accepts string $x$. Thus, element $E \, [n]$ contains the final axiomatic pair $\left\langle \; f, \; 0 \, \right\rangle$. To build the tree of string $x$ with root node $S$, function \emph{BT} is called with parameters $\emph{BT} \, \left( \; S, \; f, \; 0, \; n \; \right)$; then the function will recursively build all the subtrees and will assemble them in the final tree. Function \emph{BT} returns the syntax tree in the form of a parenthesized string, with brackets labeled by the root nonterminal of each (sub)tree. The commented code follows.
\par
\begin{tabbing}
\hspace{0.25cm} \= \hspace{0.75cm} \= \hspace{0.75cm} \= \hspace{0.75cm} \=\hspace{0.75cm} \= \hspace{0.75cm} \= \hspace{0.75cm} \= \hspace{0.75cm} \= \hspace{0.75cm} \= \hspace{0.1875cm} \= \kill
\> $\emph{BuildTree} \; \left( \; X, \; f, \; j, \; i \; \right)$  \\[0.2cm]
\> - - $X$ is a nonterminal, $f$ is a final state of $M_X$ and $0 \leq j \leq i \leq n$ \\[0.0cm]
\> - - return as parenthesized string the syntax tree rooted at node $X$ \\[0.2cm]
\> - - node $X$ will have a list $\mathcal{C}$ of terminal and nonterminal child nodes \\[0.0cm]
\> - - either list $\mathcal{C}$ will remain empty or it will be filled from right to left \\[0.2cm]
\> $\mathcal{C} := \varepsilon$  \> \> \> \> \> - - set to $\varepsilon$ the list $\mathcal{C}$ of child nodes of $X$ \\[0.2cm]
\> $q := f$            \> \> \> \> \> - - set to $f$ the state $q$ in machine $M_X$ \\[0.2cm]
\> $k := i$            \> \> \> \> \> - - set to $i$ the index $k$ of vector $E$ \\[0.2cm]
\> - -  walk back the sequence of term. \& nonterm. shift oper.s in $M_X$ \\[0.2cm]
\> \textbf{while} $\left( \; q \not = 0_X \; \right)$ \textbf{do} \> \> \> \> \> - - while current state $q$ is not initial \\[0.2cm]
\> \> - - try to backwards recover a terminal shift move $p \stackrel {x_k} \to q$, i.e., \\[0.0cm]
\> \> - - check if node $X$ has terminal $x_k$ as its current child leaf \\[0.2cm]
(a) \> \> \textbf{if} $\left( \def\arraystretch{1.25}\begin{array}{l} \exists \, h = k - 1 \ \ \exists \, p \in Q_X \ \ \text{such that} \\
\left\langle \; p, \; j \; \right\rangle \in E \, [h] \ \land \ \text{net has $p \stackrel {x_k} \to q$} \end{array} \right)$ \textbf{then} \\[0.2cm]
\> \> \> $\mathcal{C} := x_k \cdot \mathcal{C}$ \> \> \> - - concatenate leaf $x_k$ to list $\mathcal{C}$ \\[0.2cm]
\> \> \textbf{end if} \\[0.2cm]
\> \> - - try to backwards recover a nonterm. shift oper. $p \stackrel Y \to q$, i.e., \\[0.0cm]
\> \> - - check if node $X$ has nonterm. $Y$ as its current child node \\[0.2cm]
(b) \> \> \textbf{if} $\left( \def\arraystretch{1.25}\begin{array}{l} \exists \, Y \in V \ \ \exists \, e \in F_Y \ \ \exists \, h \; j \leq h \leq k \leq i \ \ \exists \, p \in Q_X \ \ \text{s.t.} \\ \left\langle \; e, \; h \; \right\rangle \in E \, [k] \ \land \ \left\langle \; p, \; j \; \right\rangle \in E \, [h] \ \land \ \text{net has $p \stackrel Y \to q$} \end{array} \right)$ \textbf{then} \\[0.2cm]
\> \> \> - - recursively build the subtree of the derivation: \\[0.0cm]
\> \> \> - - $Y \underset G {\stackrel + \Rightarrow} x_{h + 1} \; \dots \; x_k$ if $h < k$ \ or \ $Y \underset G {\stackrel + \Rightarrow} \varepsilon$ if $h = k$ \\[0.0cm]
\> \> \> - - and concatenate to list $\mathcal{C}$ the subtree of node $Y$ \\[0.2cm]
\> \> \> $\mathcal{C} := \emph{BuildTree} \; \left( \; Y, \; e, \; h, \; k \; \right) \cdot \mathcal{C}$ \\[0.2cm]
\> \> \textbf{end if} \\[0.2cm]
\> \> $q := p$        \> \> \> \> - - shift the current state $q$ back to $p$ \\[0.2cm]
\> \> $k := h$        \> \> \> \> - - drag the current index $k$ back to $h$ \\[0.2cm]
\> \textbf{end while} \\[0.2cm]
\> \textbf{return} $\left( \; \mathcal{C} \; \right)_X$ \> \> \> \> \> - - return the tree rooted at node $X$ \qed
\end{tabbing}
\par
Figure \ref{figTraceEarley} reports the analysis trace of Example \ref{exEarley1} and also shows solid edges that correspond to iterations of the \textit{while} loop in the procedure, and dashed edges that match the recursive calls. Notice that calling function \emph{BT} with equal indices $i$ and $j$, means building a subtree of the empty string, which may be made of one or more nullable nonterminals. This happens in the Figure \ref{FigBuildTree2}, witch shows the tree of the \emph{BT} calls and of the returned subtrees for Example \ref{exEarley1}, with the call $BT \, (B, \, 0_B, \, 4, \, 4 \,)$.
\par
\begin{figure}[htbp!]
\begin{center}
\scalebox{1.1}{\pspar\psset{arrows=-,levelsep=2.0cm,nodesep=5pt,treesep=1.5cm}
\pstree{\TR[name=root]{$\def\arraystretch{1.5}\begin{array}{c}  \emph{BT} \, \left( \, S, \, 4_S, \, 0, \, 6 \, \right) \\ = \Big( \; a \; a \; b \; \big( \; b \; \left( \; \varepsilon \; \right)_B \; a \; \big)_B \; a \; \Big)_S \end{array}$}}{
    \TR{$a_1$}
    \TR{$a_2$}
    \TR{$b_3$}
    \pstree[treesep=1.25cm]{\TR[name=B1]{$\def\arraystretch{1.5}\begin{array}{c} \emph{BT} \, \left( \, B, \, 3_B, \, 3, \, 5 \, \right) \\
    = \big( \; b \; \left( \; \varepsilon \; \right)_B \; a \; \big)_B \end{array}$}}{
        \TR{$b_4$}
        \pstree[treesep=1.0cm,levelsep=1.25cm]{\TR[name=B2]{$\def\arraystretch{1.5}\begin{array}{c}  \emph{BT} \, \left( \, B, \, 0_B, \, 4, \, 4 \, \right) \\ = \left( \; \varepsilon \; \right)_B \end{array}$}}{
            \TR{$\varepsilon$}
		}
  		\TR{$a_5$}
  	}
  	\TR{$a_6$}
}}
\end{center}
\caption{\label{FigBuildTree2} Calls and return values of \emph{BuildTree} for Example \ref{exEarley1}.}
\end{figure}
\par
Notice that since a leftmost derivation uniquely identifies a syntax tree, the conditions in the two mutually exclusive \emph{if-then}
conditionals inside the while loop, are always satisfied in only one way: otherwise the analyzed string would admit several
distinct left derivations and therefore it would be ambiguous.
\par
\addvspace{10pt}
\par
As previously remarked, nullable terms are dealt with by the Earley algorithm through a chain of \emph{Closure} and \emph{Nonterminal Shift} operations. An optimized version of the algorithm can be defined to perform the analysis of nullable terminals in a single step, along the same lines as defined by \cite{AycockHorspool2002}; further work by the same authors also defines optimized procedures for building the parse tree in the presence of nullable nonterminals, which can also be adjusted and applied to our version of the Earley algorithm.
\section{Conclusion}
We hope that this  extension and conceptual compaction of classical parser construction methods, will be appreciated by compiler and language designers as well as by instructors. Starting from syntax  diagrams, which are the most readable representation of grammars in language reference manuals, our method directly constructs deterministic shift-reduce parsers, which are more general and accurate than all preceding proposals. Then we have extended to the $EBNF$ case Beatty's old  theoretical comparisons of $LL \, (1)$ versus $LR \, (1)$ grammars, and exploited it to  derive general deterministic top-down parsers through   step-wise simplifications of shift-reduce parsers. We have evidenced that such simplifications are correct if multiple-transitions and left-recursive derivations are excluded. For completeness, to address the needs of non-deterministic $EBNF$ grammars, we have included an accurate presentation of the tabular Earley parsers, including syntax tree generatio. Our goal of coming up with a minimalist comprehensive presentation of parsing methods for Extended $BNF$ grammars, has thus been attained.
\par
To finish we mention a practical development. There are circumstances that suggest or impose to use separate parsers, for different language parts identified by a grammar partition, i.e., by the sublanguages generated by certain subgrammars. The idea of \emph{grammar partition} dates back to \cite{Korenjak69}, who wanted to reduce the size of an $LR \, (1)$ pilot by decomposing the original parser into a family of subgrammar parsers, and to thus reduce the number of candidates and macro-states. Parser size reduction remains a goal of parser partitioning in the domain of natural language processing, e.g., see \cite{MengEtAl2002}. In  such projects the component parsers are all homogeneous, whereas we are more interested in heterogeneous partitions, which use different parsing algorithms. Why should one want to diversify the algorithms used for different language parts? First, a language may contain parts that are harder to parse than others; thus a simpler $ELL \, (1)$ parser should be used whenever possible, limiting the use of an $ELR \, (1)$ parser - or even of an Earley one - to the sublanguages that warrant a more powerful method. Second, there are well-known examples of language embedding, such as SQL inside C, where the two languages may be biased towards different parsing methods.
\par
Although heterogeneous parsers can be built on top of legacy parsing programs, past experimentation of mixed-mode parsers (e.g., \cite{CrespiPsaila98}) has met with practical rather than conceptual difficulties, caused by the need to interface different parser systems. Our unifying approach looks promising for building seamless heterogeneous parsers that switch from an algorithm to another as they proceed. This is due to the homogeneous representation of the parsing stacks and tables, and to the exact formulation of the conditions that enable to switch from a more to a less general algorithm.
Within the current approach, mixed mode parsing should have little or no implementation overhead, and can be viewed as a pragmatic technique for achieving greater parsing power without committing to a more general but less efficient non-deterministic algorithm, such as the so-called generalized $LR$ parsers initiated by \cite{Tomita86}.
\par
\paragraph*{Aknowledgement} We are grateful to the students of Formal Languages and Compiler courses at Politecnico di Milano, who bravely accepted to study this theory while in progress. We also thank Giorgio Satta for helpful discussions of Earley algorithms.
\bibliographystyle{acmtrans}
\bibliography{LFC_bib}
\newpage
\section{Appendix}
\subsection*{Proof of theorem \ref{theorELR1}}
Let $G$ be an $EBNF$ grammar represented by machine net $\mathcal{M}$ and let $\hat{G}$ be the equivalent right-linearized grammar. Then net $\mathcal{M}$ meets the $ELR \, (1)$ condition if, and only if, grammar $\hat{G}$ meets the $LR \, (1)$ condition.
\begin{proof}
Let $Q$ be the set of the states of $\mathcal{M}$. Clearly, for any non-empty rule $X \to Y \; Z$ of $\hat{G}$, it holds $X \in Q$, $Y \in  \Sigma \cup \set{ \; 0_A \; \mid \quad 0_A \in Q \; }$ and $Z \in  Q \setminus \set{ \; 0_A \; \mid \quad 0_A \in Q \; }$.
\par
Let $\mathcal{P}$ and $\mathcal{\hat{P}}$ be the $ELR$ and $LR$ pilots of $G$ and $\hat{G}$, respectively. Preliminarily we study the correspondence between their transition functions, $\vartheta$ and $\hat{\vartheta}$, and their m-states, denoted by $I$ and $\hat{I}$. It helps us compare the pilot graphs for the running example in Figures \ref{figRunningExELRpilot} and \ref{figPilotRightLinearized}, and for Example \ref{exConvergenceProperty} in Figure \ref{figMultiCandidatesInBaseRightLin}. Notice that in the m-states of the $LR$ pilot $\mathcal{\hat{P}}$, the candidates are denoted by marked rules (with a $\bullet$).
\par
We observe that since grammar $\hat{G}$ is $BNF$, the graph of $\mathcal{\hat{P}}$ has the well-known property that all the edges entering the same m-state, have identical labels. In contrast, the identical-label property (which is a form of locality) does not hold for $\mathcal{P}$ and therefore a m-state of $\mathcal{P}$ is possibly split into several m-states of $\mathcal{\hat{P}}$.
\par
Due to the very special right-linearized grammar form, the following mutually exclusive classification of the m-states of $\mathcal{\hat{P}}$, is exhaustive:
\begin{itemize}
\item the m-state $\hat{I_0}$ is \emph{initial}
\item a m-state $\hat{I}$ is \emph{intermediate} if every candidate in $\hat{I}_{\vert base}$ has the form $p_A \to Y \; \bullet \; q_A$
\item a m-state $\hat{I}$ is a \emph{sink reduction} if every candidate has the form $p_A \to Y \; q_A \; \bullet$
\end{itemize}
\par
\begin{figure}[hbt!]
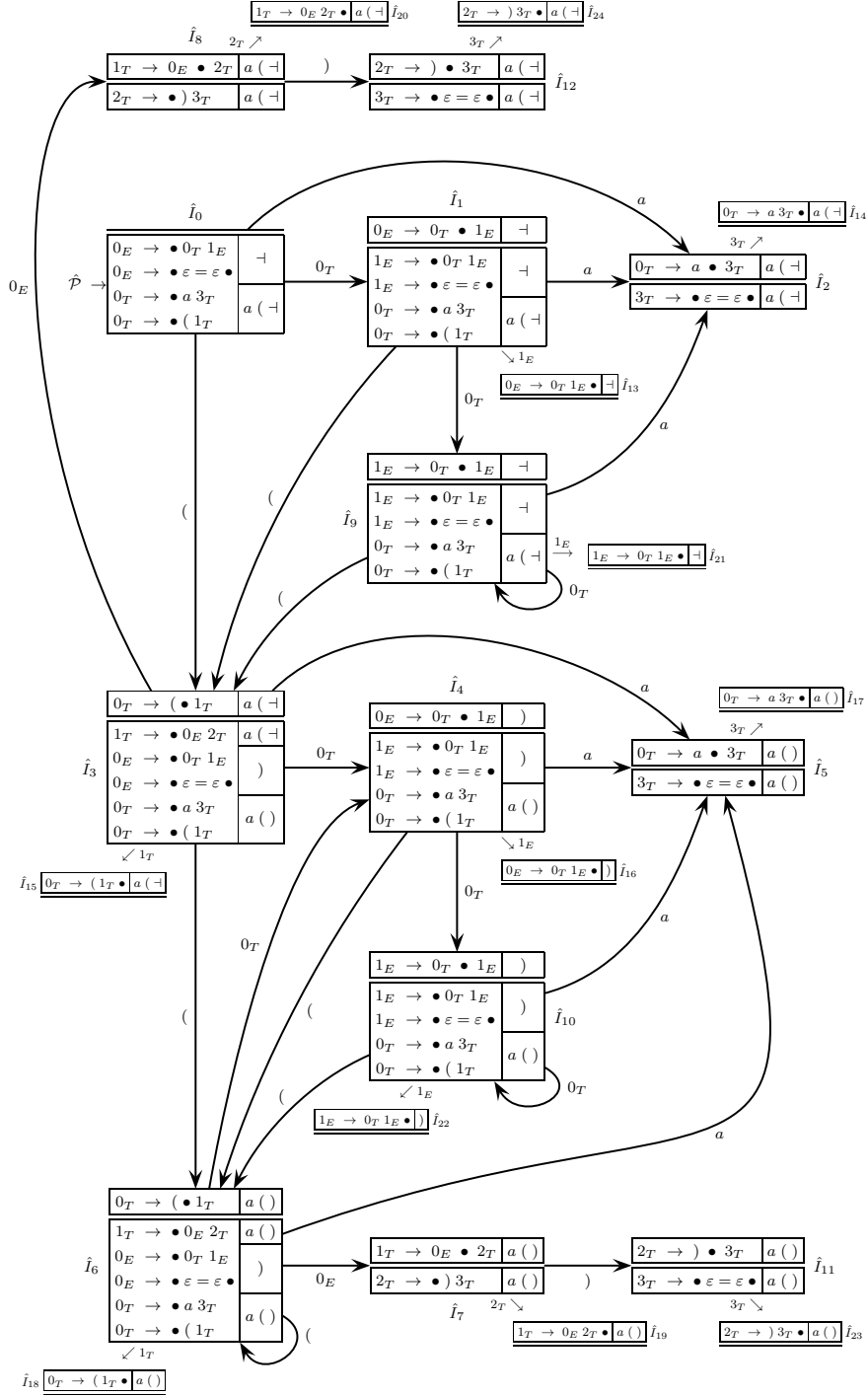

\begin{center}
\vspace{0.25cm}
\scalebox{0.75}{
\pspar\psset{border=0pt,nodesep=0pt,rowsep=1.9cm}
\begin{psmatrix}

\rnode{J8}{$\def\arraystretch{1.25}\begin{array}{|rcl|c|} \hline
1_T & \to & 0_E \; \bullet \; 2_T & a \; ( \; \dashv \\
\hline \hline
2_T & \to & \bullet \; ) \; 3_T & a \; ( \; \dashv \\
\hline
\end{array}$} &

\rnode{J12}{$\def\arraystretch{1.25}\begin{array}{|rcl|c|} \hline
2_T & \to & ) \; \bullet \; 3_T & a \; ( \; \dashv \\
\hline \hline
3_T & \to & \bullet \; \varepsilon = \varepsilon \; \bullet & a \; ( \; \dashv \\
\hline
\end{array}$} & \\

\rnode{J0}{$\def\arraystretch{1.25}\begin{array}{|rcl|c|}
\hline\hline 0_E & \to & \bullet \; 0_T \; 1_E & \multirow{2}{*}{$\dashv$} \\
0_E & \to & \bullet \; \varepsilon = \varepsilon \; \bullet & \\
\cline{4-4} 0_T & \to & \bullet \; a \; 3_T & \multirow{2}{*}{$a \; ( \; \dashv$} \\
0_T & \to & \bullet \; ( \; 1_T & \\
\hline
\end{array}$} &

\rnode{J1}{$\def\arraystretch{1.25}\begin{array}{|rcl|c|} \hline
0_E & \to & 0_T \; \bullet \; 1_E & \dashv \\
\hline \hline
1_E & \to & \bullet \; 0_T \; 1_E & \multirow{2}{*}{$\dashv$} \\
1_E & \to & \bullet \; \varepsilon = \varepsilon \; \bullet &  \\
\cline{4-4} 0_T & \to & \bullet \; a \; 3_T & \multirow{2}{*}{$a \; ( \; \dashv$} \\
0_T & \to & \bullet \; ( \; 1_T & \\
\hline
\end{array}$} &

\rnode{J2}{$\def\arraystretch{1.25}\begin{array}{|rcl|c|} \hline
0_T & \to & a \; \bullet \; 3_T & a \; ( \; \dashv  \\
\hline \hline
3_T & \to & \bullet \; \varepsilon = \varepsilon \; \bullet & a \; ( \; \dashv  \\
\hline
\end{array}$} \\ &

\rnode{J9}{$\def\arraystretch{1.25}\begin{array}{|rcl|c|} \hline
1_E & \to & 0_T \; \bullet \; 1_E & \dashv \\
\hline \hline
1_E & \to & \bullet \; 0_T \; 1_E & \multirow{2}{*}{$\dashv$} \\
1_E & \to & \bullet \; \varepsilon = \varepsilon \; \bullet &  \\
\cline{4-4} 0_T & \to & \bullet \; a \; 3_T & \multirow{2}{*}{$a \; ( \; \dashv$} \\
0_T & \to & \bullet \; ( \; 1_T & \\
\hline
\end{array}$} & \\

\rnode{J3}{$\def\arraystretch{1.25}\begin{array}{|rcl|c|} \hline
0_T & \to & ( \; \bullet \; 1_T & a \; ( \; \dashv \\
\hline \hline 1_T & \to & \bullet \; 0_E \; 2_T & a \; ( \; \dashv \\
\cline{4-4} 0_E & \to & \bullet \; 0_T \; 1_E & \multirow{2}{*}{$)$} \\
0_E & \to & \bullet \; \varepsilon = \varepsilon \; \bullet & \\
\cline{4-4} 0_T & \to & \bullet \; a \; 3_T & \multirow{2}{*}{$a \; ( \; )$} \\
0_T & \to & \bullet \; ( \; 1_T & \\
\hline
\end{array}$} &

\rnode{J4}{$\def\arraystretch{1.25}\begin{array}{|rcl|c|} \hline
0_E & \to & 0_T \; \bullet \; 1_E & ) \\
\hline \hline
1_E & \to & \bullet \; 0_T \; 1_E & \multirow{2}{*}{$)$} \\
1_E & \to & \bullet \; \varepsilon = \varepsilon \; \bullet &  \\
\cline{4-4} 0_T & \to & \bullet \; a \; 3_T & \multirow{2}{*}{$a \; ( \; )$} \\
0_T & \to & \bullet \; ( \; 1_T & \\
\hline
\end{array}$} &

\rnode{J5}{$\def\arraystretch{1.25}\begin{array}{|rcl|c|} \hline
0_T & \to & a \; \bullet \; 3_T & a \; ( \; )  \\
\hline \hline
3_T & \to & \bullet \; \varepsilon = \varepsilon \; \bullet & a \; ( \; )  \\
\hline
\end{array}$} \\ &

\rnode{J10}{$\def\arraystretch{1.25}\begin{array}{|rcl|c|} \hline
1_E & \to & 0_T \; \bullet \; 1_E & ) \\
\hline \hline
1_E & \to & \bullet \; 0_T \; 1_E & \multirow{2}{*}{$)$} \\
1_E & \to & \bullet \; \varepsilon = \varepsilon \; \bullet &  \\
\cline{4-4} 0_T & \to & \bullet \; a \; 3_T & \multirow{2}{*}{$a \; ( \; )$} \\
0_T & \to & \bullet \; ( \; 1_T & \\
\hline
\end{array}$} & \\

\rnode{J6}{$\def\arraystretch{1.25}\begin{array}{|rcl|c|} \hline
0_T & \to & ( \; \bullet \; 1_T & a \; ( \; ) \\
\hline \hline 1_T & \to & \bullet \; 0_E \; 2_T & a \; ( \; ) \\
\cline{4-4} 0_E & \to & \bullet \; 0_T \; 1_E & \multirow{2}{*}{$)$} \\
0_E & \to & \bullet \; \varepsilon = \varepsilon \; \bullet & \\
\cline{4-4} 0_T & \to & \bullet \; a \; 3_T & \multirow{2}{*}{$a \; ( \; )$} \\
0_T & \to & \bullet \; ( \; 1_T & \\
\hline
\end{array}$} &

\rnode{J7}{$\def\arraystretch{1.25}\begin{array}{|rcl|c|} \hline
1_T & \to & 0_E \; \bullet \; 2_T & a \; ( \; ) \\
\hline \hline
2_T & \to & \bullet \; ) \; 3_T & a \; ( \; ) \\
\hline
\end{array}$} &

\rnode{J11}{$\def\arraystretch{1.25}\begin{array}{|rcl|c|} \hline
2_T & \to & ) \; \bullet \; 3_T & a \; ( \; ) \\
\hline \hline
3_T & \to & \bullet \; \varepsilon = \varepsilon \; \bullet & a \; ( \; ) \\
\hline
\end{array}$}

\nput[labelsep=0.5cm]{-55}{J1}{\scalebox{0.8}{$\def\arraystretch{1.25}\begin{array}[c]{|rcl|c|} \hline 0_E & \to & 0_T \; 1_E \; \bullet & \dashv \\ \hline\hline \end{array}$\;$\hat{I}_{13}$}}
\nput[labelsep=0.1cm]{-55}{J1}{\scalebox{0.8}{$\searrow 1_E$}}

\nput[labelsep=0.75cm]{-15}{J9}{\scalebox{0.8}{$\def\arraystretch{1.25}\begin{array}[c]{|rcl|c|} \hline 1_E & \to & 0_T \; 1_E \; \bullet &\dashv \\ \hline\hline \end{array}$\;$\hat{I}_{21}$}}
\nput[labelsep=0.1cm]{-20}{J9}{\scalebox{0.8}{$\longrightarrow$}}
\nput[labelsep=0.15cm]{-13}{J9}{\scalebox{0.8}{$1_E$}}

\nput[labelsep=0.5cm]{-55}{J4}{\scalebox{0.8}{$\def\arraystretch{1.25}\begin{array}[c]{|rcl|c|} \hline 0_E & \to & 0_T \; 1_E \; \bullet & ) \\ \hline\hline \end{array}$\;$\hat{I}_{16}$}}
\nput[labelsep=0.1cm]{-55}{J4}{\scalebox{0.8}{$\searrow 1_E$}}

\nput[labelsep=0.5cm]{-115}{J10}{\scalebox{0.8}{$\def\arraystretch{1.25}\begin{array}[c]{|rcl|c|} \hline 1_E & \to & 0_T \; 1_E \; \bullet & ) \\ \hline\hline \end{array}$\;$\hat{I}_{22}$}}
\nput[labelsep=0.1cm]{-115}{J10}{\scalebox{0.8}{$\swarrow 1_E$}}

\nput[labelsep=0.5cm]{-120}{J3}{\scalebox{0.8}{$\hat{I}_{15}$\;$\def\arraystretch{1.25}\begin{array}[c]{|rcl|c|} \hline 0_T & \to & ( \; 1_T \; \bullet & a \; ( \; \dashv \\ \hline\hline \end{array}$}}
\nput[labelsep=0.1cm]{-120}{J3}{\scalebox{0.8}{$\swarrow 1_T$}}

\nput[labelsep=0.5cm]{-120}{J6}{\scalebox{0.8}{$\hat{I}_{18}$\;$\def\arraystretch{1.25}\begin{array}[c]{|rcl|c|} \hline 0_T & \to & ( \; 1_T \; \bullet & a \; ( \; ) \\ \hline\hline \end{array}$}}
\nput[labelsep=0.1cm]{-120}{J6}{\scalebox{0.8}{$\swarrow 1_T$}}

\nput[labelsep=0.5cm]{60}{J5}{\scalebox{0.8}{$\def\arraystretch{1.25}\begin{array}[c]{|rcl|c|} \hline 0_T & \to & a \; 3_T \; \bullet & a \; ( \; ) \\ \hline\hline \end{array}$\;$\hat{I}_{17}$}}
\nput[labelsep=0.1cm]{60}{J5}{\scalebox{0.8}{$3_T\nearrow$}}

\nput[labelsep=0.5cm]{60}{J2}{\scalebox{0.8}{$\def\arraystretch{1.25}\begin{array}[c]{|rcl|c|} \hline 0_T & \to & a \; 3_T \; \bullet & a \; ( \; \dashv \\ \hline\hline \end{array}$\;$\hat{I}_{14}$}}
\nput[labelsep=0.1cm]{60}{J2}{\scalebox{0.8}{$3_T\nearrow$}}

\nput[labelsep=0.5cm]{-45}{J7}{\scalebox{0.8}{$\def\arraystretch{1.25}\begin{array}[c]{|rcl|c|} \hline 1_T & \to & 0_E \; 2_T \; \bullet & a \; ( \; ) \\ \hline\hline \end{array}$\;$\hat{I}_{19}$}}
\nput[labelsep=0.1cm]{-45}{J7}{\scalebox{0.8}{$2_T\searrow$}}

\nput[labelsep=0.5cm]{-60}{J11}{\scalebox{0.8}{$\def\arraystretch{1.25}\begin{array}[c]{|rcl|c|} \hline 2_T & \to & ) \; 3_T \; \bullet & a \; ( \; ) \\ \hline\hline \end{array}$\;$\hat{I}_{23}$}}
\nput[labelsep=0.1cm]{-60}{J11}{\scalebox{0.8}{$3_T\searrow$}}

\nput[labelsep=0.5cm]{45}{J8}{\scalebox{0.8}{$\def\arraystretch{1.25}\begin{array}[c]{|rcl|c|} \hline 1_T & \to & 0_E \; 2_T \; \bullet & a \; ( \; \dashv \\ \hline\hline \end{array}$\;$\hat{I}_{20}$}}
\nput[labelsep=0.1cm]{45}{J8}{\scalebox{0.8}{$2_T\nearrow$}}

\nput[labelsep=0.5cm]{60}{J12}{\scalebox{0.8}{$\def\arraystretch{1.25}\begin{array}[c]{|rcl|c|} \hline 2_T & \to & ) \; 3_T \; \bullet & a \; ( \; \dashv \\ \hline\hline \end{array}$\;$\hat{I}_{24}$}}
\nput[labelsep=0.1cm]{60}{J12}{\scalebox{0.8}{$3_T\nearrow$}}

\nput[labelsep=5pt]{90}{J0}{$\hat{I}_0$}

\nput[labelsep=5pt]{90}{J1}{$\hat{I}_1$}

\nput[labelsep=5pt]{180}{J9}{$\hat{I}_9$}

\nput[labelsep=5pt]{0}{J2}{$\hat{I}_2$}

\nput[labelsep=5pt]{180}{J3}{$\hat{I}_3$}

\nput[labelsep=5pt]{90}{J4}{$\hat{I}_4$}

\nput[labelsep=5pt]{0}{J10}{$\hat{I}_{10}$}

\nput[labelsep=5pt]{0}{J5}{$\hat{I}_5$}

\nput[labelsep=5pt]{180}{J6}{$\hat{I}_6$}

\nput[labelsep=5pt]{-90}{J7}{$\hat{I}_7$}

\nput[labelsep=5pt]{0}{J11}{$\hat{I}_{11}$}

\nput[labelsep=5pt]{90}{J8}{$\hat{I}_8$}

\nput[labelsep=5pt]{0}{J12}{$\hat{I}_{12}$}

\nput[labelsep=0pt]{180}{J0}{$\hat{\mathcal{P}} \; \to$}

\ncline{J0}{J1} \naput{$0_T$}

\ncarc[arcangle=45]{J0}{J2} \aput(0.85){$a$}

\ncline{J1}{J2} \naput{$a$}

\ncline{J1}{J9} \naput{$0_T$}

\nccurve[angleA=15,angleB=-110,ncurvA=0.8,ncurvB=0.6]{J9}{J2} \nbput{$a$}

\ncarc[arcangle=-20]{J9}{J3} \nbput{$($}

\nccurve[angleA=-30,angleB=-60,ncurvA=2,ncurvB=2]{J9}{J9} \aput(0.25){$0_T$}

\ncline{J0}{J3} \nbput{$($}

\ncarc[arcangle=-15]{J1}{J3} \nbput{$($}

\nccurve[angleA=120,angleB=-180,ncurvA=1,ncurvB=0.3]{J3}{J8} \naput{$0_E$}

\ncline{J8}{J12} \naput{$)$}

\ncarc[arcangle=45]{J3}{J5} \aput(0.85){$a$}

\ncline{J4}{J5} \naput{$a$}

\ncline{J3}{J4} \naput{$0_T$}

\ncline{J3}{J6} \nbput{$($}

\nccurve[angleA=-30,angleB=-60,ncurvA=2,ncurvB=2]{J10}{J10} \aput(0.25){$0_T$}

\ncline{J4}{J10} \naput{$0_T$}

\nccurve[angleA=15,angleB=-110,ncurvA=0.8,ncurvB=0.6]{J10}{J5} \nbput{$a$}

\ncarc[arcangle=-20]{J10}{J6} \nbput{$($}

\ncarc[arcangle=-10]{J4}{J6} \naput{$($}

\nccurve[angleB=-160,angleA=80,ncurvB=0.4,ncurvA=0.7]{J6}{J4} \naput{$0_T$}

\nccurve[angleA=-30,angleB=-60,ncurvA=2,ncurvB=2]{J6}{J6} \aput(0.25){$($}

\ncline{J6}{J7} \nbput{$0_E$}

\nccurve[angleA=20,angleB=-75,ncurvA=1.5,ncurvB=1.5]{J6}{J5} \nbput{$a$}

\ncline{J7}{J11} \nbput{$)$}

\end{psmatrix}}
\vspace{0.5cm}
\end{center}
\caption{Pilot graph of the right-linearized grammar $\hat{\mathcal{G}}$ of the running example (see Ex. \ref{reteEsprAritmRightLin}, p.\pageref{reteEsprAritmRightLin}).}
\label{figPilotRightLinearized}
\end{figure}
For instance, in Figure \ref{figPilotRightLinearized} the intermediate m-states are numbered from $1$ to $12$ and the sink reduction m-states from $13$ to $24$.
\par
We say that a candidate $\langle q_X, \, \lambda \rangle$ of $\mathcal{P}$ \emph{corresponds} to a candidate of $\hat{\mathcal{P}}$ of the form $\langle p_X \to s \; \bullet \; q_X, \, \rho \rangle$, if the look-ahead sets are identical, i.e., if $\lambda = \rho$. Then two m-states $I$ and $\hat{I}$ of $\mathcal{P}$ and $\hat{\mathcal{P}}$, respectively, are called \emph{correspondent} if the candidates in $I_{ \vert base}$ and in $\hat{I}_{ \vert base}$ correspond to each other. Moreover, we arbitrarily define as correspondent the initial m-states $I_0$ and $\hat{I}_0$. To illustrate, in the running example a few pairs of correspondent m-states are: $(\hat{I}_1, \, I_1)$, $(\hat{I}_9, \, I_1)$, $(\hat{I}_4, \, I_4)$ and $(\hat{I}_{10}, \, I_4)$.
\par
The following straightforward properties of correspondent m-states will be needed.
\begin{lemma}\label{lemmaRightLinPilot}
The mapping defined by the correspondence relation from the set containing m-state $\hat{I}_0$ and the intermediate m-states of $\hat{\mathcal{P}}$, to the set of the m-states of $\mathcal{P}$, is total, many-to-one and onto (surjective).
\begin{enumerate}
\item\label{lemmaItem1} For any terminal or nonterminal symbol $s$ and for any correspondent m-states $I$ and $\hat{I}$,
transition $\vartheta \, (I, \, s) = I'$ is defined $\iff$ transition $\vartheta \, (\hat{I}, \, s) = \hat{I'}$ is defined and m-state $\hat{I'}$ is intermediate.   moreover $I', \, \hat{I'}$ are correspondent.
\item\label{lemmaItem2} Let state $f_A$ be final non-initial. Candidate $\langle f_A, \, \lambda \rangle$ is in m-state $I$ (actually
in its base $I_{ \vert base}$) $\iff$ a correspondent m-state $\hat{I}$ contains candidates $\langle p_A \to s \; \bullet \; f_A, \, \lambda \rangle$ and $\langle f_A \to \varepsilon \; \bullet, \, \lambda \rangle$.
\item\label{lemmaItem3} Let state $0_A$ be final initial with $A \neq S$. Candidate $\langle 0_A, \, \pi \rangle$ is in m-state $I$ $\iff$
a correspondent m-state $\hat{I}$ contains candidate $\langle 0_A \to \varepsilon \; \bullet, \, \pi \rangle$.
\item\label{lemmaItem4} Let state $0_S$ be final. Candidate $\langle 0_S, \, \pi \rangle$ is in m-state $I_0$ $\iff$ candidate
$\langle 0_S \to \varepsilon \; \bullet, \, \pi \rangle$ is in m-state $\hat{I}_0$.
\item\label{lemmaItem5} For any pair of correspondent m-states $I$ and $\hat{I}$, and for any initial state $0_A$, it holds
$0_A \in I_{ \vert closure}$ $\iff$ $0_A \to \bullet \; \alpha \in \hat{I}_{ \vert closure}$ for every alternative $\alpha$ of $0_A$.
\end{enumerate}
\end{lemma}
\begin{proof} (of Lemma). First, to prove that the mapping is onto, assume by contradiction that a m-state $I$ such that $\langle q_A, \, \lambda \rangle \in I$, has no correspondent intermediate m-state in $\hat{\mathcal{P}}$. Clearly, there exists a $\hat{I}$ such that the kernels of $\hat{I}_{ \vert base}$ and $I_{ \vert base}$ are identical because the right-linearized grammar has the same derivations as $\mathcal{M}$ (vs Sect. \ref{subSectDerivation}), and it remains that the look-aheads of a correspondent candidate in $\hat{I}_{ \vert base}$ and $I_{ \vert base}$ differ. The definition of look-ahead set for a $BNF$ grammar (not included) is the traditional one; it is easy to check that our Def. \ref{defLookaheadSet} of closure function computes exactly the same sets, a contradiction.
\par
\emph{Item} \ref{lemmaItem1}. We observe that the bases of $I$ and $\hat{I}$ are identical, therefore any candidate $\langle 0_B, \, \pi\rangle$ computed by  Def. \ref{defLookaheadSet}, matches a candidate in $\hat{I}$, with $\langle 0_B \to \bullet \; s \; q_B, \, \pi \rangle$ computed by the traditional closure function for $\hat{\mathcal{P}}$. Therefore if $\vartheta \, (I, \, s)$ is defined, also $\vartheta \, (\hat{I}, \, s)$ is defined and yields an intermediate m-state.
Moreover the next m-states are clearly correspondent since their bases are identical.
On the other hand, if  $\hat{\vartheta} \, (\hat{I}, \, s)$ is a sink reduction m-state then $\vartheta \, (I, \, s)$ is undefined.
\par
\emph{Item} \ref{lemmaItem2}. Consider an edge $p_A \stackrel s \to f_A$ entering state $f_A$. By item \ref{lemmaItem1}, consider $J$ and $\hat{J}$ as correspondent m-states that include the predecessor state $p_A$, resp. within the  candidate $\langle p_A, \, \lambda \rangle$ and $\langle q_A \to t \; \bullet \; p_A, \, \lambda \rangle$. Then $\vartheta \, (J, \, s)$ includes $\langle f_A, \, \lambda \rangle$ and $\hat{\vartheta} \, (\hat{J}, \, s)$ includes $\langle p_A \to s \; \bullet \; f_A, \, \lambda \rangle$, hence also $\hat{\vartheta} \, (\hat{J}, \, s)_{ \vert closure}$ includes $\langle f_A \to \varepsilon \; \bullet, \, \lambda \rangle$.
\par
\emph{Item} \ref{lemmaItem3}. If candidate $\langle 0_A, \, \pi \rangle$ is in $I$, it is in $I_{ \vert closure}$. Thus $I_{ \vert base}$ contains a candidate $\langle p_B, \, \lambda \rangle$ such that $closure \, \left(\langle p_B, \, \lambda \rangle \right) \ni \langle 0_A, \, \pi \rangle$. Therefore, some intermediate m-state contains a candidate $\langle q_B \to s \; \bullet \; p_B, \, \lambda \rangle$ in the base and candidate
$\langle 0_A \to \varepsilon \; \bullet, \, \pi \rangle$ in its closure. The converse reasoning is analogous.
\par
\emph{Item} \ref{lemmaItem4}. This case is obvious.
\par
\emph{Item} \ref{lemmaItem5}. We consider the case where the m-state $I$ is not initial and the state $0_A \in I_{ \vert closure}$ results from the closure operation applied to a candidate of $I_{ \vert base}$ (the other cases, where $I$ is the initial m-state $I_0$ or the state $0_A$ results from the closure operation applied to a candidate of $I_{ \vert closure}$, can be similarly dealt with): $0_A \in I_{ \vert closure}$ $\iff$ some machine $M_B$ includes the edge $p_B \stackrel A \to q_B$ $\iff$ $p_B \in I_{ \vert base}$ $\iff$ for every correspondent state $\hat{I}$ there are some suitable $X$ and $r_B$ such that $r_B \to X \bullet \; p_B \in \hat{I}_{ \vert base}$ $\iff$ $p_B \to \bullet \; 0_A \; q_B \in \hat{I}_{ \vert closure}$ $\iff$
$0_A \to \bullet \; \alpha \in \hat{I}_{ \vert closure}$ for every alternative $\alpha$ of nonterminal $0_A$.
\par
This concludes the proof of the Lemma.
\end{proof}
\vspace{0.25cm}
\par
\textbf{Part ``if''} (of Theorem). We argue that a violation of the $ELR \, (1)$ condition in $\mathcal{P}$ implies an $LR \, (1)$ conflict in $\mathcal{\hat{P}}$. Three cases need to be examined.
\begin{description}
\item[\textbf{Shift - Reduce conflict}]
\par
Consider a conflict in m-state:
\[
I \ni \langle f_B, \, \set{ \; a \; } \rangle \qquad \text{where} \ f_B \ \text{is final non-initial and} \ \vartheta \, (I, \, a) \ \text{is defined}
\]
By Lemma \ref{lemmaRightLinPilot}, items (1) and (2), there exists a correspondent m-state $\hat{I}$ such that $\vartheta \, (\hat{I}, \, a)$ is defined and $\langle f_B \to \varepsilon \; \bullet, \, \set{ \; a \; } \rangle \in \hat{I}$, thus proving that the same conflict is in $\mathcal{\hat{P}}$.
\par
Similarly, consider a conflict in m-state:
\[
I \ni \langle 0_B, \, \set{ \; a \; } \rangle \qquad \text{where} \ 0_B \ \text{is final and initial and } \ \vartheta \, (I, \, a) \ \text{is defined}
\]
By Lemma \ref{lemmaRightLinPilot}, items (1) and (3), there exists a correspondent  m-state $\hat{I}$ such that $\vartheta \, (\hat{I}, \, a)$ is defined and $\langle 0_B \to \varepsilon \; \bullet, \, \set { \; a \; } \rangle \in \hat{I}$, thus proving that the same conflict is in $\mathcal{\hat{P}}$.
\par
\item[\textbf{Reduce - Reduce conflict}]
\par
Consider a conflict in m-state:
\[
I \supseteq \big\{ \, \langle f_A , \, \set{ \; a \; } \rangle, \ \langle f_B , \, \set{ \; a \; } \rangle \, \big\} \qquad \text{where} \ f_A \ \text{and} \ f_B \ \text{are final non-initial}
\]
By item (2) of the Lemma, the same conflict exists in a m-state:
\[
\hat{I} \supseteq \big\{ \, \langle f_A \to \varepsilon \; \bullet, \, \set{ \; a \; } \rangle, \ \langle f_B \to \varepsilon \; \bullet, \ \set{ \; a \; } \rangle \, \big\}
\]
Similarly, a conflict in m-state:
\[
I \supseteq \big\{ \, \langle 0_A , \, \set{ \; a \; } \rangle, \ \langle f_B, \, \set{ \; a \; } \rangle \, \big\} \qquad\text{where} \ 0_A \ \text{and} \ f_B \ \text{are final}
\]
corresponds, by items (2) and (3), to a conflict in m-state:
\[
\hat{I} \supseteq \big\{ \, \langle 0_A \to \varepsilon \; \bullet, \, \set{ \; a \; } \rangle, \ \langle f_B \to \varepsilon \; \bullet, \set{ \; a \; } \rangle \, \big\}
\]
By item (4) the same holds true for the special case $0_A = 0_S$.
\par
\item[\textbf{Convergence conflict}]
\par
Consider a convergence conflict $I \stackrel X \to I'$, where:
\[
I \supseteq \big\{ \, \langle p_A, \, \set{ \; a \; } \rangle, \langle q_A, \, \set{ \; a \; } \rangle \, \big\}
\]
\[
\delta \, (p_A, \, X) = \delta \, (q_A, \, X) = r_a \qquad \ I'_{ \vert base} \ni \langle r_A, \, \set{ \; a \; } \rangle
\]
First, if neither $p_A$ nor $q_A$ are the initial state, both candidates are in the base of $I$. By item (1) there are correspondent intermediate m-states and transition $\hat{I} \stackrel X \to \hat{I'}$ with $\hat{I'} \supseteq \big\{ \, \langle p_A \to X \; \bullet \; r_A, \; \set{ \; a \; } \rangle, \ \langle p_A \to X \; \bullet \; r_A, \, \set{ \; a \; } \rangle \, \big\}$. Therefore m-state $\hat{\vartheta} \, (\hat{I'}, \, r_A)$ contains two reduction candidates with identical look-ahead, which is a conflict.
\par
Quite similarly, if (arbitrarily) $q_A \equiv 0_A$, then candidate $\langle p_A, \, \set{ \; a \; } \rangle$ is in the base of $I$ and $I_{ \vert base}$ necessarily contains a candidate $C = \langle s, \, \rho \rangle$ such that $closure \, (C) = \langle 0_A, \, \set{ \; a \; } \rangle$. Therefore for some $t$ and $Y$, there exists a m-state $\hat{I}$ correspondent of $I$ such that $\langle t \to Y \; \bullet \; s, \, \rho \rangle \in \hat{I}_{ \vert base}$ and $\langle 0_A \to \bullet \; X \; r_A, \, \set{ \; a \; } \rangle \in \hat{I}_{ \vert closure}$, hence it holds
$\hat{\vartheta} \, (\hat{I}, \, X) = \hat{I}'$, and $\langle 0_A \to X \; \bullet \; r_A, \, \set{ \; a \; } \rangle \in \hat{I}'$.
\par
\end{description}
\vspace{0.25cm}
\par
\textbf{Part ``only if''} (of Theorem). We argue  that every $LR \, (1)$ conflict in $\mathcal{\hat{P}}$ entails a violation of the $ELR \, (1)$ condition in $\mathcal{P}$.
\begin{description}
\item[\textbf{Shift - Reduce conflict}]
The  conflict  occurs in a m-state $\hat{I}$ such that $\langle f_B \to \varepsilon \; \bullet, \, \set{ \; a \; } \rangle \in \hat{I}$ and $\vartheta \, (\hat{I}, \, a)$ is defined. By items (1) and (2) (or (3)) of the Lemma, the correspondent m-state $I$ contains $\langle f_B, \, \set{ \; a \; } \rangle$ and the move $\vartheta \, (I, \, a)$ is defined, thus resulting in the same conflict.
\par
\item[\textbf{Reduce - Reduce conflict}]
\par
First, consider a m-state having the form:
\[
\hat{I}_{ \vert closure} \supseteq \big\{ \, \langle f_A \to \varepsilon \; \bullet, \, \set{ \; a \; } \rangle, \
\langle f_B \to \varepsilon \; \bullet, \, \set{ \; a \;} \rangle \, \big\}
\]
where $f_A$ and $f_B$ are final non-initial. By item (2) of the Lemma, the correspondent m-state $I$ contains the candidates
$I \supseteq \big\{ \, \langle f_A, \, \set{ \; a \; } \rangle, \ \langle f_B, \, \set{ \; a \; } \rangle \, \big\}$ and has the same conflict.
\par
Second, consider a m-state having the form:
\[
\hat{I}_{ \vert closure} \supseteq \big\{ \, \langle f_A \to \varepsilon \; \bullet, \, \set{ \, a \; } \rangle, \ \langle 0_B \to \varepsilon \; \bullet, \, \set{ \; a \; } \rangle \, \big\}
\]
where $f_A$ is final non-initial. By items (2) and (3) the same conflict is in the correspondent m-state $I$.
\par
At last, consider  a reduce-reduce conflict in a sink reduction m-state:
\[
\hat{I} \supseteq \big\{ \, \langle p_A \to X \; r_A \; \bullet, \, \set{ \; a \; } \rangle, \ \langle q_A \to X \; r_A \; \bullet, \ \set{ \; a \; } \rangle \, \big\}
\]
Then there exist a m-state and a transition $\hat{I}' \stackrel {r_A} \to \hat{I}$ such that $\hat{I}'$ contains candidates $\langle p_A \to X \; \bullet \; r_A, \, \set{ \; a \; } \rangle$ and $\langle q_A \to X \; \bullet \; r_A, \, \set{ \; a \; } \rangle$. Therefore the correspondent m-state $I'$ contains candidate $\langle r_A, \, \set{ \; a \; } \rangle$, and there are a m-state $\hat{I}''$ and a transition $\hat{I}'' \stackrel X \to \hat{I}'$ such that it holds $\big\{ \, \langle p_A \to \bullet \; X \; r_A, \set{ \; a \; } \rangle, \ \langle q_A \to \bullet \; X \; r_A, \, \set{ \; a \; } \rangle \, \big\} \subseteq \hat{I}''_{ \vert closure} $. Since $\hat{I}''_{ \vert closure} \neq \emptyset$, $\hat{I}''$ is not
a sink reduction state; let us call $I''$ its correspondent state. Then:
\begin{itemize}
\item If $p_A$ is initial then $\langle p_A, \set{ \; a \; } \rangle \in I''_{ \vert closure}$ by virtue of Item \ref{lemmaItem5}.
\item If $p_A$ is not initial then there exists a candidate $\langle t_A \to Z \bullet \; p_A, \, \set{ \; a \; } \rangle \in
\hat{I}''_{ \vert base}$ (notice that the look-ahead is the same because we are still in the same machine $M_A$)
and $\langle p_A, \, \set{ \; a \; } \rangle \in I''$. A similar reasoning applies to state $q_A$. Therefore
$\langle p_A, \, \set{ \; a \; } \rangle \in I''$ and $I'' \stackrel X \to I'$ has a convergence conflict.
\end{itemize}
\end{description}
This concludes the `if'' and ``only if'' parts, and the Theorem itself. \qed
\end{proof}
\vspace{0.25cm}
\par
As a second example to illustrate convergence conflicts, the pilot graph $\hat{\mathcal{P}}$ equivalent to $\mathcal{P}$ of Figure \ref{figPilotWithConvergence}, p.\pageref{figPilotWithConvergence}, is shown in Figure \ref{figMultiCandidatesInBaseRightLin}.
\par
\begin{figure}[htb!]
\vspace{10pt}
\begin{center}
\scalebox{1.0}{
\pspar\psset{border=0cm,nodesep=0pt,colsep=30pt,rowsep=30pt}
\begin{psmatrix}

\rnode{I0}{$\def\arraystretch{1.25}\begin{array}{|rcl|c|} \hline\hline
0_S & \to & \bullet \; a \; 1_S & \multirow{3}{*}{$\dashv$} \\
0_S & \to & \bullet \; b \; 4_S & \\
0_S & \to & \bullet \; 0_A \; 5_S & \\ \cline{4-4}
0_A & \to & \bullet \; a \; 1_A & e \\ \hline
\end{array}$}

&

\rnode{I1}{$\def\arraystretch{1.25}\begin{array}{|rcl|c|} \hline
0_S & \to & a \; \bullet \; 1_S & \dashv \\ \cline{4-4}
0_A & \to & a \; \bullet \; 1_A & e \\ \hline\hline
1_S & \to & \bullet \; b \; 2_S & \dashv \\ \cline{4-4}
1_A & \to & \bullet \; 0_S \; 2_A & \multirow{5}{*}{$e$} \\
0_S & \to & \bullet \; a \; 1_S & \\
0_S & \to & \bullet \; b \; 4_S & \\
0_S & \to & \bullet \; 0_A \; 5_S & \\
0_A & \to & \bullet \; a \; 1_A & \\ \hline
\end{array}$}

&

\rnode{I4}{$\def\arraystretch{1.25}\begin{array}{|rcl|c|} \hline
0_S & \to & a \; \bullet \; 1_S & \multirow{2}{*}{$e$} \\
0_A & \to & a \; \bullet \; 1_A & \\ \hline\hline
1_S & \to & \bullet \; b \; 2_S & \multirow{6}{*}{$e$}  \\
1_A & \to & \bullet \; 0_S \; 2_A & \\
0_S & \to & \bullet \; a \; 1_S & \\
0_S & \to & \bullet \; b \; 4_S & \\
0_S & \to & \bullet \; 0_A \; 5_S & \\
0_A & \to & \bullet \; a \; 1_A & \\ \hline
\end{array}$} \\

\rnode{gram}{\multirow{3}{*}{
$\left\{\def\arraystretch{1.25}\begin{array}{rcl}
0_S & \to & a \; 1_S \\
0_S & \to & b \; 4_S \\
0_S & \to & 0_A \; 5_S \\
1_S & \to & b \; 2_S \\
2_S & \to & c \; 3_S \\
2_S & \to & d \; 3_S \\
3_S & \to & \varepsilon \\
4_S & \to & c \; 3_S \\
5_S & \to & e \; 3_S \\
0_A & \to & a \; 1_A \\
1_A & \to & 0_S \; 2_A \\
2_A & \to & \varepsilon \\
\end{array}\right. \qquad$}}

&

\rnode{I5}{$\def\arraystretch{1.25}\begin{array}{|rcl|c|} \hline
1_S & \to & b \; \bullet \; 2_S & \dashv \\ \cline{4-4}
0_S & \to & b \; \bullet \; 4_S & e \\ \hline\hline
2_S & \to & \bullet \; c \; 3_S &  \multirow{2}{*}{$\dashv$} \\
2_S & \to & \bullet \; d \; 3_S & \\ \cline{4-4}
4_S & \to & \bullet \; c \; 3_S & e \\ \hline
\end{array}$}

&

\rnode{I8}{$\def\arraystretch{1.25}\begin{array}{|rcl|c|} \hline
1_S & \to & b \; \bullet \; 2_S & \multirow{2}{*}{$e$} \\
0_S & \to & b \; \bullet \; 4_S & \\ \hline\hline
2_S & \to & \bullet \; c \; 3_S & \multirow{3}{*}{$e$} \\
2_S & \to & \bullet \; d \; 3_S & \\
4_S & \to & \bullet \; c \; 3_S & \\ \hline
\end{array}$} \\

&

\rnode{I10}{$\def\arraystretch{1.25}\begin{array}{|rcl|c|} \hline
2_S & \to & c \; \bullet \; 3_S & \dashv \\ \cline{4-4}
4_S & \to & c \; \bullet \; 3_S & e \\ \hline\hline
3_S & \to & \bullet \; \varepsilon = \varepsilon \; \bullet & e \; \dashv \\ \hline
\end{array}$}

&

\rnode{I11}{$\def\arraystretch{1.25}\begin{array}{|rcl|c|} \hline
2_S & \to & c \; \bullet \; 3_S & \multirow{2}{*}{$e$} \\
4_S & \to & c \; \bullet \; 3_S & \\ \hline\hline
3_S & \to & \bullet \; \varepsilon = \varepsilon \; \bullet & e \\ \hline
\end{array}$} \\

&

\rnode{I10sink}{$\def\arraystretch{1.25}\begin{array}{|rcl|c|} \hline
2_S & \to & c \; 3_S \; \bullet & \dashv \\ \cline{4-4}
4_S & \to & c \; 3_S \; \bullet & e \\ \hline\hline
\end{array}$}

&
\rnode{I11sink}{$\def\arraystretch{1.25}\begin{array}{|rcl|c|} \hline
2_S & \to & c \; 3_S \; \bullet & \multirow{2}{*}{$e$} \\
4_S & \to & c \; 3_S \; \bullet & \\ \hline\hline
\end{array}$}

\nput[labelsep=5pt]{90}{I0}{$\hat{I}_0$}

\nput[labelsep=5pt]{90}{I1}{$\hat{I}_1$}

\nput[labelsep=5pt]{90}{I4}{$\hat{I}_4$}

\nput[labelsep=5pt]{-180}{I5}{$\hat{I}_5$}

\nput[labelsep=5pt]{-180}{I8}{$\hat{I}_8$}

\nput[labelsep=5pt]{-180}{I10}{$\hat{I}_{10}$}

\nput[labelsep=5pt]{-180}{I11}{$\hat{I}_{11}$}

\nput[labelsep=5pt]{-180}{I10sink}{$\hat{I}_{10\,sink}$}

\nput[labelsep=5pt]{-180}{I11sink}{$\hat{I}_{11\,sink}$}

\nput[labelsep=0pt]{180}{I0}{$\hat{\mathcal{P}} \; \to$}

\nput[labelsep=5pt]{-90}{I11sink}{\parbox{3cm}{\centering reduce-reduce conflict \par $\iff$ \par convergence conflict}}

\nput[labelsep=5pt]{90}{gram}{\parbox{3cm}{\centering right-linearized \par grammar for Ex. \ref{exConvergenceProperty}}}

\ncline{I0}{I1} \naput{$a$}

\ncline{I1}{I4} \naput{$a$}

\ncline{I1}{I5} \nbput{$b$}

\ncline{I4}{I8} \nbput{$b$}

\nccurve[angleA=-20,angleB=20,ncurvA=2.5,ncurvB=2.5]{I4}{I4} \nbput{$a$}

\ncline{I5}{I10} \nbput{$c$}

\ncline{I8}{I11} \nbput{$c$}

\ncline{I10}{I10sink} \nbput{$3_S$}

\ncline{I11}{I11sink} \nbput{$3_S$}

\end{psmatrix}}
\vspace{1cm}
\end{center}
\caption{Part of the (traditional - with marked rules) pilot of the right-linearized grammar for Ex. \ref{exConvergenceProperty}; the reduce-reduce conflict in m-state $\hat{I}_{11 \, sink}$ matches the convergence conflict of the edge $I_8 \stackrel c \to I_{11}$ of $\mathcal{P}$ (Fig. \ref{figPilotWithConvergence}, p.\pageref{figPilotWithConvergence}). }\label{figMultiCandidatesInBaseRightLin}
\end{figure}
\par
\subsection*{Proof of property \ref{PropDisjointGuideSetsConv}}
If the guide sets of a $PCFG$ are disjoint, then the machine net satisfies the $ELL \, (1)$ condition of Definition \ref{defELL(1)condition}.
\begin{proof}
Since the $ELL \, (1)$ condition consists of the three properties of the $ELR \, (1)$ pilot: (1) absence of left recursion; (2) $STP$, i.e., absence of multiple transitions; and (3) absence of shift-reduce and reduce-reduce conflicts, we will prove that the presence of disjoint guide sets in the $PCFG$ implies that all these three conditions hold.
\begin{enumerate}
\item We prove that if a grammar (represented as a net) is left recursive then its guide sets are not disjoint. If the grammar is left recursive
then $\exists \, n > 1$ such that in the $PCFG$ there are $n$ call edges $0_{A_1} \dashrightarrow 0_{A_2}$,
$0_{A_2} \dashrightarrow 0_{A_3}$, $\ldots$, $0_{A_n} \dashrightarrow 0_{A_1}$. Then since $\nexists \, A \in V$ such that $L_A \, (G) = \set{ \; \varepsilon \; }$, it holds $\exists \, a \in \Sigma$ and $\exists \, A_i \in V$ such that there is a shift edge $0_{A_i} \stackrel a \to p_A$ and $a \in Gui \, (0_{A_i} \dashrightarrow 0_{A_{i + 1}})$, hence the guide sets for these two edges in the $PCFG$ are not disjoint. Notice that the presence of left recursion due to rules of the kind $A \to X \; A \; \ldots$ with $X$ a nullable nonterminal, can be ruled out because this kind of left recursion leads to a shift-reduce conflict, as discussed in Section \ref{subsectELL(1)condition} and shown in Figure \ref{FigLeftRecDerivNoELL1}.
\item \label{i2} We prove that the presence of multiple transitions implies that the guide sets in the $PCFG$
are not all disjoint. This is done by induction, through starting from the initial m-state of the pilot
automaton (which has an empty base) and showing that all reachable m-states of the pilot automaton
satisfy the Single Transition Property ($STP$), unless the guide sets of the $PCFG$ are not all disjoint.
We also note that the transitions from the m-states satisfying $STP$ lead to m-states the base of which is a
singleton set: we call such m-states Singleton Base and we say that they satisfy the Singleton Base
Property ($SBP$).
\par
\underline{Induction base}: Assume there exists a multiple transition from the initial m-state
$I_0 = closure \, (\left\langle 0_S, \, \set{ \; \dashv \; } \right\rangle)$. Hence $I_0$ includes $n > 1$ candidates
$\left\langle 0_{X_1}, \, \pi_1 \right\rangle$, $\left\langle 0_{X_2}, \, \pi_2 \right\rangle$, $\ldots$,
$\left\langle 0_{X_n}, \, \pi_n \right\rangle$, and for some $h$ and $k$ with $1 \leq h < k \leq n$,
the machine net includes transitions $0_{X_h} \stackrel X \to p_{X_h}$ and $0_{X_k} \stackrel X \to p_{X_k}$,
with $X$ being a terminal symbol s.t. $X = a \in \Sigma$ or being a nonterminal one s.t. $X = Z \in V$.
Let us first consider the case $X = a \in \Sigma$ and assume the candidate $\left\langle 0_{X_k}, \, \pi_k \right\rangle$ derives from candidate $\left\langle 0_{X_h}, \, \pi_h \right\rangle$ through a (possibly iterated) closure operation;
then the $PCFG$ includes the call edges $0_{X_h} \dashrightarrow 0_{X_{h+1}}$, $\ldots$, $0_{X_{k-1}} \dashrightarrow 0_{X_k}$, the inclusions
$\set{ \; a \; } \subseteq Gui \, (0_{X_{k-1}} \dashrightarrow 0_{X_k}) \subseteq \, \ldots \, \subseteq Gui \, (0_{X_{h}} \dashrightarrow 0_{X_{h+1}})$ hold, and there are two non-disjoint guide sets $Gui \, (0_{X_h} \stackrel a \to p_{X_h}) = \set{ \; a \; }$ and $Gui \, (0_{X_{h}} \dashrightarrow 0_{X_{h+1}}) \supseteq \set{ \; a \; }$.
\par	
Assuming instead that there is a candidate $\left\langle 0_{X_j}, \, \pi_j \right\rangle$ such that both $\left\langle 0_{X_h},\,  \pi_h \right\rangle$ and $\left\langle 0_{X_k}, \, \pi_k \right\rangle$ are derived by the closure operation through distinct paths, then the $PCFG$ includes the call edges $0_{X_j} \dashrightarrow 0_{X_{jh1}}$, $\ldots$, $0_{X_{jh-1}} \dashrightarrow 0_{X_{jh}}$ with $0_{X_{jh}} = 0_{X_h}$ and $0_{X_j} \dashrightarrow 0_{X_{jk1}} \ldots 0_{X_{jk-1}} \dashrightarrow 0_{X_{jk}}$ with $0_{X_{jk}} = 0_{X_k}$; thus the following inclusions hold: $\set{ \; a \; } \subseteq Gui \, (0_{X_{jh-1}} \dashrightarrow 0_{X_{jh}}) \subseteq \ldots \subseteq Gui \, (0_{X_j} \dashrightarrow 0_{X_{jh1}})$ and $\set{ \; a \; } \subseteq Gui \, (0_{X_{jk-1}} \dashrightarrow 0_{X_{jk}}) \subseteq \ldots \subseteq Gui \, (0_{X_j} \dashrightarrow 0_{X_{jk1}})$; therefore the two guide sets $Gui \, (0_{X_j} \dashrightarrow 0_{X_{jh1}})$ and $Gui \, (0_{X_j} \dashrightarrow 0_{X_{jk1}})$ are not disjoint.
\par	
Let us now consider the case $X = Z \in V$. Then if the candidate $\left\langle 0_{X_k}, \, \pi_k \right\rangle$
derives from $\left\langle 0_{X_h}, \, \pi_h \right\rangle$, the $PCFG$ includes the call edges
$0_{X_h} \dashrightarrow 0_{X_{h+1}}$,  $\ldots$, $0_{X_{k-1}} \dashrightarrow 0_{X_k}$, as well as
$0_{X_h} \dashrightarrow 0_Z$ and $0_{X_k} \dashrightarrow 0_Z$, hence
$Ini \, (0_Z) \subseteq Gui \, (0_{X_h} \dashrightarrow 0_Z)$ and
$Ini \, (0_Z) \subseteq Gui \, (0_{X_k} \dashrightarrow 0_Z) \subseteq Gui \, (0_{X_h} \dashrightarrow 0_{X_{h+1}})$,
and the two guide sets $Gui \, (0_{X_h} \dashrightarrow 0_Z)$ and $Gui \, (0_{X_h} \dashrightarrow 0_{X_{h+1}})$
are not disjoint. The case of both candidates $\left\langle 0_{X_h}, \, \pi_h \right\rangle$
and $\left\langle 0_{X_k}, \, \pi_k \right\rangle$ deriving from a common candidate
$\left\langle 0_{X_j}, \, \pi_j \right\rangle$ through distinct sequences of closure operations,
is similarly dealt with and leads to the conclusion that in the $PCFG$ there are two call edges originating
from state $0_{X_j}$, the guide sets of which are not disjoint. This concludes the induction base.
\par
\underline{Inductive step}: Consider a non-initial, singleton base m-state $I$, such that a multiple transition $\vartheta \, (I, \, X)$ is defined. Since $I$ has a singleton base, of the two candidates from where the multiple transition originates, at least one is in $I_{ \vert closure}$. The case of both candidates in the closure is at all similar to the one treated in the base case of the induction. Therefore, we consider the case of one candidate in the base with a non-initial state $q_A$ and one in the closure with an initial state $0_Y$. Then if $X = a \in \Sigma$ the $PCFG$ has: states $r_A$, $r_Y$ and $0_{Y_1}$, $\ldots$, $0_{Y_n}$, with $n \geq 1$ and $Y_n = Y$; shift edges $q_A \stackrel a \to r_A$ and $0_Y \stackrel a \to r_Y$; and call edges $q_A \dashrightarrow 0_{Y_1}$, $\ldots$, $0_{Y_{n-1}} \dashrightarrow 0_Y$ such that $\set{ \; a \; } \subseteq Ini \, (0_Y) \subseteq Gui \, (0_{Y_{n-1}} \dashrightarrow 0_Y) \subseteq \ldots \subseteq Gui \, (q_A \dashrightarrow 0_{Y_1})$ and $\set { \; a \; } = Gui \, \left(q_A \stackrel a \to r_A\right)$. Thus the two guide sets $Gui \, \left(q_A \stackrel a \to r_A\right)$ and $Gui \, (q_A \dashrightarrow 0_{Y_1})$ are not disjoint. Otherwise if $X = Z \in V$, the $PCFG$ includes the nonterminal shift edges $q_A \stackrel Z \to r_A$ and $0_{Y_n} \stackrel Z \to r_{Y_n}$ with $r_{Y_n} \in Q_{Y_n}$, and the call edges $q_A \dashrightarrow 0_{Y_1}$, $\ldots$, $0_{Y_{n-1}} \dashrightarrow 0_{Y_n}$, and also the two call edges $q_A \dashrightarrow 0_Z$ and $0_{Y_n} \dashrightarrow 0_Z$. Then it holds $Gui \, (q_A \dashrightarrow 0_Z) \, \cap \, Gui \, (q_A \dashrightarrow 0_{Y_1}) \supseteq Ini \, (Z)$, and	the two guide sets $Gui \, (q_A \dashrightarrow 0_Z)$ and $Gui \, (q_A \dashrightarrow 0_{Y_1})$ are not disjoint. This concludes the induction.
\item We prove that the presence of shift-reduce or reduce-reduce conflicts implies that the guide sets are not disjoint. We can assume that
all m-states are singleton base, hence if there are two conflicting candidates at most one of them
is in the base, as in the point (\ref{i2}) above. So:
\begin{enumerate}
\item First we consider reduce-reduce conflicts and the two cases of candidates
being both in the closure or only one.
\begin{enumerate}
\item If both candidates are in the closure then there are $n > 1$ candidates
$\left\langle 0_{X_1}, \, \pi_1 \right\rangle$, $\ldots$, $\left\langle 0_{X_n}, \, \pi_n \right\rangle$ such that the candidates $\left\langle 0_{X_1}, \, \pi_1 \right\rangle$ and $\left\langle 0_{X_n}, \, \pi_n \right\rangle$ are conflicting, hence for some $a \in \Sigma$ it holds $a \in Gui \, (0_{X_1} \to)$ and $a \in Gui \, (0_{X_n} \to)$. If the candidate $\left\langle 0_{X_n}, \, \pi_n \right\rangle$ derives from candidate $\left\langle 0_{X_1}, \, \pi_1 \right\rangle$ through a sequence of closure operations, then the $PCFG$ includes the chain of call edges $0_{X_1} \dashrightarrow 0_{X_2}$, $\ldots$, $0_{X_{n-1}} \dashrightarrow 0_{X_n}$ with $a \in Gui \, (0_{X_{n-1}} \dashrightarrow 0_{X_n}) \subseteq \ldots \subseteq Gui \, (0_{X_{1}} \dashrightarrow 0_{X_{2}})$, therefore $a \in Gui \, (0_{X_1} \to) \, \cap \, Gui \, (0_{X_1} \dashrightarrow 0_{X_2})$ and these two guide sets are not disjoint. If instead there is a candidate $\left\langle 0_{X_j}, \, \pi_j \right\rangle$ such that both $\left\langle 0_{X_1}, \, \pi_1 \right\rangle$ and $\left\langle 0_{X_n}, \, \pi_n \right\rangle$ are obtained from that candidate through distinct paths, it can be shown that in the $PCFG$ there are two call edges departing from state $0_{X_j}$, the guide sets of which are not disjoint, as in the point (\ref{i2}) above.
\item The case where one candidate is in the base and one is in the closure, is quite similar to the one in the previous point:
there are a candidate $\left\langle f_A, \, \pi \right\rangle$ and $n \geq 1$ candidates $\left\langle 0_{X_1}, \, \pi_1 \right\rangle$, $\ldots$, $\left\langle 0_{X_n}, \, \pi_n \right\rangle$ such that the candidates $\left\langle f_A, \, \pi \right\rangle$ and $\left\langle 0_{X_n}, \, \pi_n \right\rangle$ are conflicting, hence for some $a \in \Sigma$ it holds  $a \in Gui \, (f_A \to)$ and $a \in Gui \, (0_{X_n} \to)$, the $PCFG$ includes the call edges $f_A \dashrightarrow 0_{X_1}$, $\ldots$, $0_{X_{n-1}} \dashrightarrow 0_{X_n}$, whence $a \in Gui \, (f_A \to) \, \cap \, Gui \, (f_A \dashrightarrow 0_{X_1})$ and these two guide sets are not disjoint.
\end{enumerate}
\item Then we consider shift-reduce conflicts and the three cases that arise depending on whether the conflicting candidates
are both in the closure or only one, or whether there is one candidate that is both shift and reduction.
\begin{enumerate}
\item \label{c1i} If there are two conflicting candidates, both in the closure of m-state $I$, then either the closure of $I$ includes
$n > 1$ candidates $\left\langle 0_{X_1}, \, \pi_1 \right\rangle$, $\ldots$, $\left\langle 0_{X_n}, \, \pi_n \right\rangle$ and the $PCFG$ includes the call edges $0_{X_1} \dashrightarrow 0_{X_2}$, $\ldots$, $0_{X_{n-1}} \dashrightarrow 0_{X_n}$, or the closure includes three candidates $\left\langle 0_{X_j}, \, \pi_j \right\rangle$, $\left\langle 0_{X_h}, \, \pi_h \right\rangle$ and $\left\langle 0_{X_k}, \, \pi_k \right\rangle$ such that $\left\langle 0_{X_h}, \, \pi_h \right\rangle$ and $\left\langle 0_{X_k}, \, \pi_k \right\rangle$ derive from $\left\langle 0_{X_j}, \pi_j \right\rangle$ through distinct chains of closure operations.
\par
We first consider a linear chain of closure operations from $\left\langle 0_{X_1}, \, \pi_1 \right\rangle$ to $\left\langle 0_{X_n}, \, \pi_n \right\rangle$. Let us first consider the case where $\left\langle 0_{X_1}, \, \pi_1 \right\rangle$ is a reduction candidate (hence $0_{X_1}$ is a final state), and $\exists \, a \in \Sigma$ such that $a \in \pi_1$ and $\exists \, q \in Q_{X_{n}}$ such that $0_{X_n} \stackrel a \to q$. Then it holds $a \in Gui \, \left(0_{X_{n-1}} \dashrightarrow 0_{X_n}\right) \subseteq \ldots \subseteq Gui \, \left(0_{X_1} \dashrightarrow 0_{X_2}\right)$, therefore $a \in Gui \, (0_{X_1} \to) \, \cap \, Gui \, (0_{X_1} \dashrightarrow 0_{X_2})$ and these two guide sets are not disjoint. Let us then consider the symmetric case where $0_{X_n}$ is a final state (hence $\left\langle 0_{X_n}, \, \pi_n \right\rangle$ is a reduction candidate), and $\exists \, a \in \Sigma$ such that $a \in \pi_n$ and $\exists \, q \in Q_{X_1}$ such that $0_{X_1} \stackrel a \to q$. Then it holds $a \in Gui \, \left(0_{X_{n-1}} \dashrightarrow 0_{X_n}\right) \subseteq \ldots \subseteq Gui \, \left(0_{X_1} \dashrightarrow 0_{X_2}\right)$, therefore $\set{ \; a \; } = Gui \, \left(0_{X_1} \stackrel a \to q\right) \, \cap \, Gui \, \left(0_{X_1} \dashrightarrow 0_{X_2}\right)$ and these two guide sets are not disjoint.
\par
The other case, where two candidates $\left\langle 0_{X_h}, \, \pi_h \right\rangle$ and $\left\langle 0_{X_k}, \, \pi_k \right\rangle$ derive from a third one $\left\langle 0_{X_j}, \, \pi_j \right\rangle$ through distinct chains of closure operations, is treated in a similar way and leads to the conclusion that in the $PCFG$ there are two call edges departing from state $0_{X_j}$, the guide sets of which are not disjoint.
\par
\item If there are two conflicting candidates, one in the base and one in the closure, then we consider the two cases that arise
depending on whether the reduction candidate is in the closure or in the base, similarly to point \ref{c1i} above.
\par
Let us first consider the case where the base contains the candidate $\langle p_A, \, \pi \rangle$, the closure includes $n \geq 1$ candidates $\langle 0_{X_1}, \, \pi_1 \rangle$, $\ldots$, $\langle 0_{X_n}, \, \pi_n \rangle$, state $0_{X_n}$ is final hence $\langle 0_{X_n}, \, \pi_n \rangle$ is a reduction candidate, and $\exists \, a \in \Sigma$ such that $a \in \pi_n$ and $\exists \, q \in Q_A$ such that $p_A \stackrel a \to q$. Then since $0_{X_n}$ is final, it holds $Nullable \, (X_n)$ and $\set{ \; a \; } \in Gui \, (0_{X_{n-1}} \dashrightarrow 0_{X_n}) \subseteq \ldots \subseteq Gui \, (0_{p_A} \dashrightarrow 0_{X_1})$, whence $\set{ \; a \; } = Gui \, \left(p_A \stackrel a \to q\right) \, \cap \, Gui \, (p_A \dashrightarrow 0_{X_1})$ and these two guide sets are not disjoint.
\par
Let us then consider the symmetric case where the base contains a candidate $\langle f_A, \, \pi \rangle$ with $f_A \in F_A$ (hence $\langle f_A, \, \pi\rangle$ is a reduction candidate), the closure includes $n \geq 1$ candidates $\langle 0_{X_1}, \, \pi_1 \rangle$, $\ldots$, $\langle 0_{X_n}, \, \pi_n \rangle$, and $\exists \, a \in \Sigma$ such that $a \in \pi$ and $\exists \, q \in Q_{X_n}$ such that $0_{X_n} \stackrel a \to q$. Then it holds $a \in Gui \, (0_{X_{n-1}} \dashrightarrow 0_{X_n}) \subseteq \ldots \subseteq Gui \, (f_A \dashrightarrow 0_{X_1})$, therefore $\set{ \; a \; } \subseteq Gui \, (f_A \to) \, \cap \, Gui \, (f_A \dashrightarrow 0_{X_1})$ and these two guide sets are not disjoint.
\par	
\item If there is one candidate $\langle p_A, \, \pi \rangle$ that is both shift and reduce, then $p_A \in F_A$, and
$\exists \, a \in \Sigma$ such that $a \in \pi$ and $\exists \, q_A \in Q_A$ such that $p_A \stackrel a \to q_A$. Then it holds $\set{ \; a \; } = Gui \, (p_A \to) \, \cap \, Gui \, \left(p_A \stackrel a \to q_A\right)$ and so there are two guide sets on edges departing from the same $PCFG$ node, that are not disjoint.
\end{enumerate}
\end{enumerate}
\end{enumerate}
This concludes the proof of the Property.
\end{proof}
\subsection*{Proof of lemma \ref{EarleyCorrect}}
If it holds $\left\langle \; q_A, \; j \; \right\rangle \in E \, [i]$, which implies inequality $j \leq i$, with $q_A \in Q_A$, i.e., state $q_A$ belongs to the machine $M_A$ of nonterminal $A$, then it holds $\left\langle \; 0_A, \; j \; \right\rangle \in E \, [j]$ and the right-linearized grammar $\hat{G}$ admits a leftmost derivation $0_A \stackrel \ast \Rightarrow x_{j+1} \, \ldots \, x_i \, q_A$ if $j < i$ or $0_A \stackrel \ast \Rightarrow q_A$ if $j = i$.
\begin{proof}
By induction on the sequence of insertion operations performed in the vector $E$ by the Analysis Algorithm.
\begin{description}
\item[Base] $\langle 0_S, \, 0 \rangle \in E \, [0]$ and the property stated by the Lemma is trivially satisfied:
$\langle 0_S, \, 0 \rangle \in E \, [0]$ and $0_S \stackrel \ast \Rightarrow 0_S$.
\par
\item[Induction] We examine the three operations below:
\par
\emph{TerminalShift}: if $\langle q_A, \, j \rangle  \in E \, [i]$ results from a TerminalShift operation,
then it holds $\exists \, q'_A$ such that $\delta \, (q'_A, \, x_i) = q_A$ and $\langle q'_A, \, j \rangle  \in E \, [i-1]$, as well as
$\langle 0_A, \, j \rangle  \in E \, [j]$ and $0_A \stackrel \ast \Rightarrow x_{j+1} \, \ldots \, x_{i-1} \; q'_A$ by the inductive hypothesis;
hence $0_A \stackrel \ast \Rightarrow x_{j+1} \, \ldots \, x_{i-1} \; x_i \; q_A$.
\par
\emph{Closure}: if $\langle 0_B, \, i \rangle \in E \, [i]$ then the property is trivially satisfied:
$\langle 0_B, \, i \rangle \in E \, [i]$ and $0_B \stackrel \ast \Rightarrow 0_B$.
\par
\emph{NonterminalShift}: suppose first that $j < i$; if $\langle q_{F_A}, \, j \rangle \in E \, [i]$ then $\langle 0_A, \, j \rangle \in E \, [j]$ and $0_A \stackrel \ast \Rightarrow x_{j+1} \, \ldots \, x_i \; q_{F_A}$ by the inductive hypothesis; furthermore, if  $\langle q_B, \, k \rangle \in E \, [j]$ then $\langle 0_B, \, k \rangle \in E \, [k]$ and $0_B \stackrel \ast \Rightarrow x_{k+1}  \ldots \, x_j \; q_B$ by the inductive hypothesis; also if $\delta \, (q_B, \, A) = q'_B$ then $q_B \Rightarrow 0_A \; q'_B$; and since $\langle q'_B, \, k \rangle$ is added to $E \, [i]$, it is eventually proved that
$\langle q'_B, \, k \rangle \in E \, [i]$ implies $\langle 0_B, \, k \rangle \in E \, [k]$ and:
\begin{eqnarray*}
0_B & \stackrel \ast \Rightarrow & x_{k+1} \, \ldots \, x_j \; q_B \\
& \stackrel \ast \Rightarrow & x_{k+1} \, \ldots \, x_j \; 0_A \; q'_B \\
& \stackrel \ast \Rightarrow & x_{k+1} \, \ldots \, x_j \; x_{j+1} \, \ldots \, x_i \; q_{F_A} \; q'_B \\
& \Rightarrow & x_{k+1} \, \ldots \, x_i \; q'_B
\end{eqnarray*}
\end{description}
If $j = i$ then the same reasoning applies, but there is not any \emph{TerminalShift} since the derivation does not generate any terminal. In particular, for the \emph{NonterminalShift} case, we have
that if $\langle q_{F_A}, \, i \rangle \in E \, [i]$, then $\langle 0_A, \, i \rangle \in E \, [i]$ and $0_A \stackrel \ast \Rightarrow q_{F_A}$ by the inductive hypothesis; and the rest follows as before but with $k = i$ and $0_B \stackrel \ast \Rightarrow q_B$.
\par
The reader may easily complete by himself with the two remaining proof cases where $k < j = i$ or $k = j < i$, which are combinations.
\par
This concludes the proof of the Lemma.
\end{proof}
\subsection*{Proof of lemma \ref{EarleyComplete}}
Take an \emph{EBNF} grammar $G$ and a string $x = x_1 \, \dots \, x_n$ of length $n$ that belongs to language $L \, (G)$. In the right-linearized grammar $\hat{G}$, consider any leftmost derivation $d$ of a prefix $x_1 \, \dots \, x_i$ ($i \leq n$) of $x$, that is:
\[
d \colon 0_S \stackrel + \Rightarrow x_1 \, \ldots \, x_i \, q_A \, W
\]
with $q_A \in Q_A$ and $W \in Q_A^\ast$. The two points below apply:
\begin{enumerate}
\item \label{Earleycase1app}
if it holds $W \neq \varepsilon$, i.e., $W = r_B \, Z$ for some $r_B \in Q_B$, then it holds $\exists \, j \; 0 \leq j \leq i$ and $\exists \, p_B \in Q_B$ such that the machine net has an arc $p_B \stackrel A \to r_B$ and grammar $\hat{G}$ admits two leftmost derivations $d_1 \colon 0_S \stackrel + \Rightarrow x_1 \, \ldots \, x_j \, p_B \, Z$ and $d_2 \colon 0_A \stackrel + \Rightarrow x_{j+1} \, \ldots \, x_i \, q_A$, so that derivation $d$ decomposes as follows:
\[\def\arraystretch{1.5}
\begin{array}{rclcl}
d \colon 0_S & \stackrel {d_1} \Rightarrow & x_1 \, \ldots \, x_j \, p_B \, Z & \stackrel {p_B \to 0_A \, r_B} \Longrightarrow & x_1 \, \ldots \, x_j \, 0_A \, r_B \, Z \\
& \stackrel {d_2} \Rightarrow & x_1 \, \ldots \, x_j  \, x_{j+1} \, \ldots \, x_i \, q_A \, r_B \, Z & = & x_1 \, \ldots \, x_i \, q_A \, W
\end{array}
\]
as an arc $p_B \stackrel A \to r_B$ in the net maps to a rule $p_B \to 0_A \, r_B$ in grammar $\hat{G}$
\item
this point is split into two steps, the second being the crucial one: \label{Earleycase2app}
\begin{enumerate}
\item
if it holds $W = \varepsilon$, then it holds $A = S$, i.e., nonterminal $A$ is the axiom, $q_A \in Q_S$ and $\left\langle \; q_A, \; 0 \, \right\rangle
\in E \, [i]$ \label{Earleycase2aapp}
\item if it also holds $x_1 \ldots x_i \in L \, (G)$, i.e., the prefix also belongs to language $L \, (G)$,
then it holds $q_A = f_S \in F_S$, i.e., state $q_A = f_S$ is final for the axiomatic machine $M_S$, and the prefix is accepted by the Earley algorithm \label{Earleycase2bapp}
\end{enumerate}
\end{enumerate}
Limit cases: if it holds $i = 0$ then it holds $x_1 \, \ldots \, x_i = \varepsilon$; if it holds $j = i$ then it holds $x_{j+1} \, \ldots \, x_i = \varepsilon$; and if it holds $x = \varepsilon$ (so $n = 0$) then both cases hold, i.e., $j = i = 0$.
\par
If the prefix coincides with the whole string $x$, i.e., $i = n$, then step \eqref{Earleycase2bapp} implies that string $x$, which by hypothesis belongs to language $L \, (G)$, is accepted by the Earley algorithm, which therefore is complete. \qed
\begin{proof}
By induction on the length of the derivation $S \stackrel \ast \Rightarrow x$.
\begin{description}
\item[Base] Since $0_S \stackrel \ast \Rightarrow 0_S$ the thesis (case (\ref{Earleycase1app})) is satisfied by taking $i = 0$.
\item[Induction] We examine a few cases:
\begin{itemize}
\item If $0_S \stackrel \ast \Rightarrow x_1 \, \ldots \, x_i \; q_A \; W \Rightarrow x_1 \, \ldots \, x_i \; 0_X \; r_A \; W$
(because $q_A \stackrel X \to r_A$) then the closure operation adds to $E \, [i]$ the item $\langle 0_X, \, i \rangle$, hence the thesis (case \eqref{Earleycase1app}) holds by taking for $j$ the value $i$, for $q_A$ the value $0_X$, for $q_B$ the value $r_A$ and for $W$ the value $r_A \, Z$.
\item If $0_S \stackrel \ast \Rightarrow x_1 \, \ldots \, x_i \; q_A \; W \Rightarrow x_1 \, \ldots \, x_i \; x_{i+1} \; r_A \; W$
(because $q_A \stackrel {x_{i+1}} \to r_A$) then the operation \emph{TerminalShift} ($E$, $i + 1$) adds to $E \, [i + 1]$ the item $\langle r_A, \, j \rangle$, hence the thesis (case \eqref{Earleycase1app}) holds by taking for $j$	the same value and for $q_A$ the value $r_A$.
\item If $0_S \stackrel \ast \Rightarrow x_1 \, \ldots \, x_i \; q_A \; W \Rightarrow x_1 \, \ldots \, x_i \; W$ (because
$q_A = f_A \in F_A$ and $f_A \to \varepsilon$) then:
\begin{itemize}
\item If $Z = \varepsilon$, $q_A = f_S \in F_S$ and $0_S \stackrel \ast \Rightarrow x_1 \, \ldots \, x_i \; f_S \Rightarrow x_1 \,
\ldots \, x_i$, then the current derivation step is the last one and the string is accepted because the pair $\langle f_S, \, 0 \rangle \in E \, [i]$ by the inductive hypothesis.
\item If $W = q_B \; Z$, then by applying the inductive hypothesis to $x_1 \, \ldots \, x_j$, the following facts hold:
$\exists \, k \; 0 \leq k \leq j$, $\exists \, p_C \in Q_C$ such that $0_S \stackrel \ast \Rightarrow x_1 \, \ldots \, x_k \; p_C \; Z'$, $p_C \stackrel B \to q_C$, $0_B \stackrel \ast \Rightarrow x_{k+1} \, \ldots \, x_j \; p_B$, $p_B \stackrel A \to q_B$, and the vector $E$ is such that $\langle p_C, \, h \rangle \in E \, [k]$, $\langle 0_B, \, k \rangle \in E \, [k]$, $\langle p_B, \, k \rangle \in E \, [j]$, $\langle 0_A, \, j \rangle \in E \, [j]$ and $\langle q_A, \, j \rangle \in E \, [i]$; then the nonterminal shift operation in the \emph{Completion} procedure adds to $E \, [i]$ the pair $\langle q_B, \, k \rangle$.
\par
First, we notice that from $0_B \stackrel \ast \Rightarrow x_{k+1} \, \ldots \, x_j \, p_B$, $p_B \stackrel A \to q_B$ and $0_A \stackrel \ast \Rightarrow x_{j+1} \ldots x_i f_A$, it follows that $0_B \stackrel \ast \Rightarrow x_{k+1} \, \ldots \, x_i \; q_B$. Next, since $p_C \stackrel B \to q_C$, it holds:
\[
0_S \stackrel \ast \Rightarrow x_1 \, \ldots \, x_k \; p_C \; Z \Rightarrow x_1 \, \ldots \, x_k \; 0_B \; q_C \; Z'
\]
therefore:
\[
0_S \stackrel \ast \Rightarrow x_1 \, \ldots \, x_k \; x_{k+1} \, \ldots \, x_i \; q_B \; q_C \; Z'
\]
and the thesis holds by taking for $j$ the value $k$, for $q_A$ the value $q_B$, for $q_B$ the value $q_C$ and for $W$ the value $q_C \; Z'$. The above situation is schematized in Figure \ref{figTraceEarleyCompletenessProof}.
\end{itemize}
\end{itemize}
\end{description}
\begin{figure}[h!]
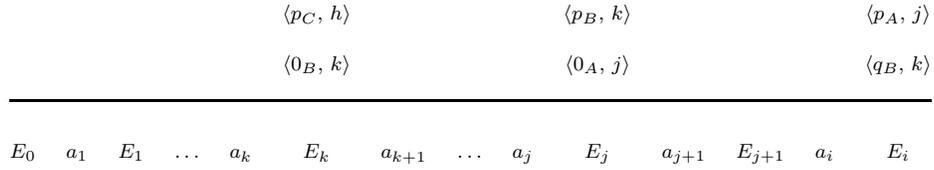

\begin{center}
\scalebox{0.95}{
\psset{border=0pt,nodesep=0pt,labelsep=5pt, arrows=->, rowsep=10pt, colsep=12pt}
\begin{psmatrix}

&&&&& \rnode{k1}{$\langle p_C, \, h \rangle$} &&&&  \rnode{j1}{$\langle p_B, \, k \rangle$} &&&& \rnode{i1}{$\langle p_A, \, j \rangle$} \\

&&&&& \rnode{k2}{$\langle 0_B, \, k \rangle$} &&&& \rnode{j2}{$\langle 0_A, \, j \rangle$} &&&& \rnode{i2}{$\langle q_B, \, k \rangle$} \\ \hline \\

$E_0$ & $a_1$ & $E_1$ & $\ldots$ & $a_k$ & $E_k$ & $a_{k+1}$ & $\ldots$ & $a_j$& $E_j$ & $a_{j+1}$ & $E_{j+1}$ & $a_{i}$ & $E_{i}$
\end{psmatrix}}
\end{center}
\caption{Schematic trace of tabular parsing.} \label{figTraceEarleyCompletenessProof}
\end{figure}
This concludes the proof of the Lemma.
\end{proof}
\end{document}